\documentclass[reqno]{amsart}
\usepackage{amsfonts,latexsym,enumerate}
\usepackage{amsmath}
\usepackage{amscd}
\usepackage{float,amsmath,amssymb,mathrsfs,bm,multirow,graphics}
\usepackage[dvips]{graphicx}
\usepackage{overpic}
\usepackage{amsaddr}
\usepackage[numbers,sort&compress]{natbib}
\usepackage{tikz}
\usepackage{epstopdf}
\usepackage{subfigure}
\usepackage{enumerate}
\usepackage{xcolor}
\definecolor{lightblue}{RGB}{216,128,144} 
\colorlet{mred}{lightblue}
\colorlet{mgreen}{lightblue}
\colorlet{mblue}{lightblue}

\newcommand{\R}{{\Bbb R}}

\newcommand{\CC}{{\Bbb C}}

\newcommand{\diag}{\text{\upshape diag\,}}

\newtheorem{theorem}{Theorem}[section]
\newtheorem{proposition}[theorem]{Proposition}
\newtheorem{lemma}[theorem]{Lemma}

\newtheorem{assumption}[theorem]{Assumption}
\newtheorem{remark}[theorem]{Remark}

\newtheorem{RHproblem}[theorem]{RH problem}
\numberwithin{equation}{section}

\begin{document}
	\title[Long-time asymptotics of the Tzitz\'eica equations]
	{Long-time asymptotics of the Tzitz\'eica equation on the line}
	
	\author{Lin Huang$^{* \dagger}$,~Deng-Shan Wang$^{* \ddagger}$ and Xiaodong Zhu$^{* \ddagger,\S}$}
	\address{$^{\dagger}$School of Science, Hangzhou Dianzi University, Zhejiang 310018, China. \\ $^{\ddagger}$Laboratory of Mathematics and Complex Systems (Ministry of Education),  \\ School of Mathematical Sciences,
		Beijing Normal University,Beijing 100875, China.\\
    ${\S}$ SISSA, via Bonomea 265, 34136, Trieste, Italy}
	\email{lin.huang@hdu.edu.cn, dswang@bnu.edu.cn, xzhu@sissa.it}

	
	\date{\today}
	
	
	
	\begin{abstract}
		In this paper, the renowned Riemann-Hilbert method is employed to study the long-time asymptotics of pure radiation solution to the initial value problem of Tzitz\'eica equation on the line, which is an important issue that remains unsolved. Initially, our analysis focuses on elucidating the properties of two reflection coefficients, which are determined by the initial values. Subsequently, leveraging these reflection coefficients, we construct a Riemann-Hilbert problem that is a powerful tool to articulate the solution of the Tzitz\'eica equation. Finally, the nonlinear steepest descent method is applied to the oscillatory Riemann-Hilbert problem, which enables us to delineate the long-time asymptotic behaviors of solutions to the Tzitz\'eica equation across various regions. Moreover, it is shown that the leading-order terms of asymptotic formulas match well with full numerical simulations.

	\end{abstract}
	
	\maketitle
	
	\tableofcontents

	\section{\bf Introduction}
	In a seminal series of papers published between 1907 and 1910, the distinguished Romanian geometer Tzitz\'eica \cite{tzitzeica} initiated a pioneering study of a special class of hyperbolic surfaces in Euclidean three-space whose second fundamental form is non-degenerate. 
	Tzitz\'eica discovered that the ratio 
	\(
	I=\frac{K}{d^4},
	\)
	where $K$ denotes the Gauss curvature and $d$ is the Euclidean distance from a fixed point to the tangent plane, remains invariant under equiaffine transformations of Euclidean 3-space. 
	A surface for which $I$ is constant is therefore called a {\it Tzitz\'eica surface} \cite{Krivoshapko-Ivanov-1984}. 
	In affine differential geometry, such a surface is also known as a {\it proper affine sphere}, since the affine distance from the origin is a nonzero constant \cite{Inoguchi-2018}. 
	
	For an indefinite proper affine sphere with negative affine mean curvature, the Gauss--Codazzi equations with Blaschke metric $h=2e^{u} dXdT$ in isothermal coordinates reduce to the nonlinear wave equation
	\begin{align}\label{TTorg}
		u_{X,T}=e^{u}-e^{-2u},
	\end{align}
	which is known as the {\it Tzitz\'eica equation}, where $u$ is a real-valued function. 
	In the context of integrable systems, Mikhailov  \cite{M-1981} showed that the Tzitz\'eica equation \eqref{TTorg} arises as a reduction of the periodic two-dimensional Toda lattice \cite{FB-1980} of type $A_2^{(2)}$,
	\[
	u_{n,X,T}= e^{u_n-u_{n-1}} - e^{u_{n+1}-u_n},
	\]
	with period three under the constraint $u_1=u$, $u_2=-u$, and $u_3=0$, which corresponds to the affine root system.  Moreover, as shown in \cite{Dunajski-2009}, the Tzitz\'eica equation \eqref{TTorg} arises as a reduction of the anti-self-dual Yang--Mills equations on $\mathbb{R}^{2,2}$ with gauge group $SL(3,\mathbb{R})$.
	
	\par
By introducing the ``light-cone'' coordinates $(x,t)$ defined by
\begin{align*}
	x = X + T, \qquad t = T - X,
\end{align*}
the Tzitz\'eica equation \eqref{TTorg} can be rewritten in the form
\begin{align}\label{Tt}
	u_{tt} - u_{xx} = e^{-2u} - e^u.
\end{align}
This equation has appeared in various mathematical and physical contexts, highlighting its importance in the theory of integrable systems. 
Dodd and Bullough \cite{DB-1977} discovered two nontrivial conservation laws associated with this equation. 
Later, Zhiber and Shabat \cite{ZS-1979} showed that the equation admits an infinite Lie--B\"acklund symmetry group, which further reveals its rich integrable structure. 
For this reason, the equation is sometimes also referred to as the Bullough--Dodd--Zhiber--Shabat equation.
	\par	
	Mikhailov \cite{M-1981} discovered a Lax pair for equation \eqref{Tt}, which opened the way to the construction of various exact solutions, including soliton solutions, finite-gap solutions, and algebro--geometric solutions \cite{Brezh,cmg-1999,KSY-1999,wgz-2015}.
	The integrable properties of this equation, such as soliton solutions, conservation laws, symmetries, B\"acklund transformations, and Darboux transformations, have attracted considerable attention in the literature \cite{Babalic-Constantinescu-Gerdjikov,KW-1981,cmg-1999,SS-1993,Brezh,NW-1997,ZG-2005}. 
	These studies have greatly advanced the understanding of the Tzitz\'eica equation.
	\par	
	The well-posedness of the Tzitz\'eica equation is an important topic in the study of integrable systems and nonlinear partial differential equations. 
	It concerns the existence, uniqueness, and continuous dependence of solutions on the initial data. 
	In 2017, Jevnikar and Yang studied the Tzitz\'eica equation in its elliptic form and investigated the occurrence of blow-up phenomena together with conditions ensuring the existence of solutions \cite{JY-2017}. 
	{More recently, Riemann--Hilbert techniques have also proved useful in the study of blow-up phenomena for related integrable models. 
		For instance, Charlier \cite{Charlier-2024-blow} employed the Riemann--Hilbert approach to construct blow-up solutions for the ``bad'' Boussinesq equation and demonstrated that a wide range of asymptotic blow-up scenarios can occur.}
	To the best of our knowledge, the well-posedness of the Tzitz\'eica equation has not yet been investigated from the perspective adopted in this work. 
	In this paper, we establish a connection between solutions of equation~\eqref{Tt} and an associated Riemann--Hilbert problem, which provides a new framework for studying the well-posedness of the equation.
	\par
	The Tzitz\'eica equation \eqref{Tt} and the sine--Gordon equation
	\begin{align}\label{SG}
		u_{tt}-u_{xx}+\sin(u)=0
	\end{align}
	are two typical examples of integrable systems arising from classical differential geometry \cite{Udriste}. 
	The sine--Gordon equation \eqref{SG} \cite{Rogers-2002} appears as the Gauss--Codazzi equation for surfaces with constant Gauss curvature $-1$ in $\mathbb{R}^3$, and hence describes pseudo-spherical surfaces. 
	Over the past decades, the sine--Gordon equation \eqref{SG} has been extensively studied by various analytical approaches. 
	For instance, Cheng, Venakides, and Zhou \cite{Zhou-Cpde-singordan} investigated the long-time asymptotics of pure-radiation solutions of the sine--Gordon equation, while Lu and Miller \cite{Lu-Miller-2022} studied the Dubrovin universality near the critical point in the semiclassical sine--Gordon equation. 
	
	The integrability-based analysis of the sine--Gordon equation is relatively tractable since it possesses a $2\times2$ Lax pair. 
	In contrast, the Tzitz\'eica equation \eqref{Tt} admits a $3\times3$ Lax pair, which leads to substantial difficulties in the development of the inverse spectral theory. 
	To the best of our knowledge, the asymptotic analysis of pure-radiation solutions to the Tzitz\'eica equation \eqref{Tt} remains an open problem.
	\par	
	In this paper, we study the long-time asymptotic behavior of solitonless (pure-radiation) solutions of the Tzitz\'eica equation \eqref{Tt} with initial data in the Schwartz class,
	\begin{align}
		\label{TT}
		\begin{cases}
			u_{tt}-u_{xx}=e^{-2u}-e^{u},\\
			u(x,0)=u_0(x)\in\mathcal{S}(\mathbb{R}),\quad 
			u_t(x,0)=u_1(x)\in\mathcal{S}(\mathbb{R}).
		\end{cases}
	\end{align}
	
	Our analysis begins with the study of the direct and inverse scattering problems associated with a third-order spectral problem \cite{Beals-Deift-Tomei,Beals-Coifman-1984}. 
	This leads to the development of an inverse scattering transform (IST) framework for the initial-value problem. 
	Within this framework, the solution of the Tzitz\'eica equation can be formulated in terms of a $3\times3$ matrix Riemann--Hilbert problem. 
	
	To analyze the long-time behavior of the solution, we employ the nonlinear steepest descent method introduced by Deift and Zhou \cite{Deift-Zhou-1993} and further developed by Lenells \cite{lenells-2015,Lenells-2018}. 
	Using this approach, we derive asymptotic formulas describing the long-time behavior of solutions. 
	For simplicity, we restrict our analysis to initial data in the Schwartz space $\mathcal S(\mathbb{R})$.
	\par	
The Riemann--Hilbert method has become a powerful analytical tool with important applications in several areas of mathematics, including integrable systems, random matrix theory, and orthogonal polynomials \cite{Deift1999}. 
Significant progress has been made in the study of the long-time asymptotic behavior of solutions to integrable equations associated with higher-order Lax pairs. 
In particular, the inverse spectral theory for integrable systems of Gelfand--Dickey type with third-order Lax pairs has been systematically developed by Charlier, Lenells, and one of the present authors \cite{Charlier-Lenells-2021,CL-2021,Charlier-Lenells-Wang-2021}. 

Other notable examples include the Degasperis--Procesi equation \cite{BS-2013,Boutet-de-Monvel-1}, the ``good'' and ``bad'' Boussinesq equations \cite{CL-2023-sectorI,CL-2023-sectorIV,CL-2023-soliton,CL-2023-sectorV,Charlier-Lenells-Wang-2021,Deift1982}, and the Lenells equation \cite{CL-2021}; see also \cite{Boutet-de-Monvel-2,Constantin-4,GWC-2021,GL-2018,HL-2020,Ds-Xd-2022,XF-2022,YF-2023,Ds-Xd-2026} for further developments.
	\par
	The paper is organized as follows. 
	In Section \ref{Main-Results}, we present the main results of this work. 
	Section \ref{spectralanalysis} is devoted to the spectral analysis based on the Lax pair associated with the Tzitz\'eica equation, which leads to the proof of Theorem \ref{RHth}. 
	Finally, in Section \ref{longtimeasymptotic}, we prove Theorem \ref{uasy} describing the long-time asymptotic behavior of solutions to the Tzitz\'eica equation.
	
	\section{\bf Main Results}\label{Main-Results}
	
	This section lists the main results of the present work. For the initial value problem \eqref{TT}, direct scattering analysis shows that one can define the scattering matrices $s(\lambda)=(s_{ij}(\lambda))_{3\times3}$ and $s^{A}(\lambda)=(s^{A}_{ij}(\lambda))_{3\times3}$ in \eqref{Scattering-Ma} and \eqref{Scattering-Ma-A} below, respectively. With these scattering data in mind,
	this paper delineates its core contributions through the formulation and proof of three main theorems, which emerge from a foundational theoretical framework, predicated on a series of basic assumptions.
	Below, we outline these assumptions, which are essential for the derivation of our main results:
	
	{\begin{assumption}\label{solitonless}
			Let $D_1 := \{\lambda \in \mathbb{C} \mid 0 < \arg(\lambda) < \frac{\pi}{3}\}$ and 
			$D_4 := \{\lambda \in \mathbb{C} \mid \pi < \arg(\lambda) < \frac{5\pi}{3}\}$, 
			see Figure~\ref{sigma0}. 
			For the initial value problem \eqref{TT}, we assume that 
			$s_{11}(\lambda) \neq 0$ for $\lambda \in \overline{D}_1 \setminus \{0\}$ and 
			$s_{11}^A(\lambda) \neq 0$ for $\lambda \in \overline{D}_4 \setminus \{0\}$. 
			That is, we assume that the scattering data contain no discrete spectrum, 
			so that problem \eqref{TT} corresponds to the solitonless case and we only 
			consider pure radiation solutions of the Tzitz\'eica equation.
	\end{assumption}
\begin{remark}
	In this manuscript, we only focus on the pure radiation case. 
	The analysis can also be generalized to the general case with solitons, 
	in which one may study the corresponding soliton resolution conjecture. 
	The soliton resolution conjecture for the bad Boussinesq equation was studied in~\cite{Charlier-Lenells-2024-soliton}.
\end{remark}
}
	\par
	The subsequent discussion will elucidate the rationality of the aforementioned assumption, demonstrating its validity by selecting a specific initial value followed by numerical calculations.
	\par
	Our first result reveals the characteristics of two spectral functions, denoted as $r_1(\lambda)$ and $r_2(\lambda)$. These functions are interpreted as the reflection coefficients associated with the equation \eqref{Tt}, determined by its initial conditions. Furthermore, $r_1(\lambda)$ and $r_2(\lambda)$ are conceptualized as the nonlinear Fourier transform of the initial data. The properties of these spectral functions are crucial for formulating the Riemann-Hilbert (RH) problem and, subsequently, for the accurate reconstruction of the solution to equation \eqref{TT} from the Riemann-Hilbert framework.
	\par
	To be specific,	the reflection coefficients $r_1(\lambda)$ and $r_2(\lambda)$ are defined by  \begin{align}\label{r1r2def}
		\begin{cases}
			r_1(\lambda):=\frac{s_{12}(\lambda)}{s_{11}(\lambda)},\quad \lambda \in (0,\infty),\\
			r_2(\lambda):=\frac{s_{12}^A(\lambda)}{s_{11}^A(\lambda)},\quad \lambda \in (-\infty,0).
		\end{cases}
	\end{align}
	In Proposition \ref{sprop} below, we demonstrate that the matrix entries $s_{11}(\lambda)$ and $s_{12}(\lambda)$, as defined in equation (\ref{r1r2def}), are smooth functions over the interval $\lambda \in (0,\infty)$. Analogously, Proposition \ref{sAprop} establishes the smoothness of the entries $s^A_{11}(\lambda)$ and $s^A_{12}(\lambda)$, also referenced in equation \eqref{r1r2def}, for $\lambda \in (-\infty,0)$. Consequently, the reflection coefficients $r_1(\lambda)$ and $r_2(\lambda)$ exhibit smoothness within their respective domains, except for potential discontinuities at points where $s_{11}(\lambda)$ and $s^A_{11}(\lambda)$ approaches zero. These zero points are indicative of the emergence of solitons. However, our analysis confines itself to scenarios devoid of solitons, focusing exclusively on pure radiation solutions (see Assumption \ref{solitonless}).
	
{\begin{theorem}\label{r1r2prop}
		Suppose that $u_0,u_1\in\mathcal{S}(\R)$ and satisfy Assumption~\ref{solitonless}. 
		Then the reflection coefficients $r_1(\lambda)$ and $r_2(\lambda)$ are well defined 
		for $\lambda\in (0,\infty)$ and $\lambda\in(-\infty,0)$, respectively, and satisfy 
		the following properties:
		\begin{enumerate}
			\item The functions $r_1(\lambda)$ and $r_2(\lambda)$ are smooth for $\lambda$ in their domain and decay rapidly as $\lambda\to\infty$.
			\item The functions $r_1(\lambda)$, $r_2(\lambda)$ and their derivatives decay rapidly as $\lambda\to0$.
		\end{enumerate}
	\end{theorem}}
	\begin{proof}
		The theorem follows the definitions of $r_1(\lambda)$ and $r_2(\lambda)$, as shown in equation \eqref{r1r2def}, alongside a thorough analysis of the properties adhered to by the scattering matrices $s(\lambda)$ and $s^{A}(\lambda)$ in \eqref{Scattering-Ma} and \eqref{Scattering-Ma-A}, respectively.
		These properties are meticulously detailed in Proposition \ref{sprop} and  Proposition \ref{sAprop}, respectively.
	\end{proof}
	
	\begin{remark} It is remarked that in his pioneering literature \cite{Zhou-JDE-1995}, Zhou had already meticulously demonstrated that as $\lambda \to 0$, the reflection coefficients $r_1(\lambda)$ and $r_2(\lambda)$ converge to zero.
	\end{remark}
	
	The second principal conclusion of this study delineates the establishment of a substantive linkage between the solution of the Tzitz\'eica equation, characterized by Schwartz class initial values, and the resolution of a specific Riemann-Hilbert problem. This problem is defined through the ``reflection coefficients"  \(r_1(\lambda)\), \(r_2(\lambda)\) involved in Theorem \ref{r1r2prop}, and a set of designated phase functions.

	\begin{RHproblem}\label{MRHp}
		To identify a \(3 \times 3\) piecewise analytic matrix-valued function, denoted as \(M(x, t, \lambda)\), which possesses characteristics outlined below:
		\begin{itemize}
			\item $M(x,t,\lambda)$ is analytic in $\CC\setminus{\Sigma}$, where
			$\Sigma=
			\cup_{j=1}^6 e^{i(j-1)\pi/3}\R_+,$
			the orientation is shown in Figure \ref{sigma0}.
			\item  $M(x,t,\lambda)$ is analytic for $\CC\setminus\Sigma$; and for $\lambda$ approaches $\Sigma$ from the left and right, the limits of $M(x,t,\lambda)$ exist. Denote the limits as $M_+(x,t,\lambda)$ and $M_-(x,t,\lambda)$, respectively, and they have the following relationship
			$$
			M_+(x,t,\lambda)=M_-(x,t,\lambda)V(x,t,\lambda).
			$$
			\item  As $\lambda\to\infty,$ for $\lambda\in \CC\setminus{\Sigma}$, \ $M(x,t,\lambda)= I+\mathcal{O}\left(\frac{1}{\lambda}\right)$.
			\item  As $\lambda\to0,$ for $\lambda\in\CC\setminus{\Sigma}$,\ $M(x,t,\lambda)= G(x,t)+\mathcal{O}(\lambda)$.
			\item  $
			M(x, t, \lambda)=\mathcal{A}^{-1} M(x, t, \omega \lambda) \mathcal{A}={\mathcal{B}} {M(x, t, {\lambda}^{*})}^{*} \mathcal{B}^{-1}, \quad \lambda \in \mathbb{C} \backslash \Sigma.
			$
		\end{itemize}
	\end{RHproblem}
{Here, if $\lambda\in e^{i(j-1)\pi/3}\mathbb{R}_+$ for $j=1,2,\dots,6$, then 
$V(x,t,\lambda)=v_j(x,t,\lambda)$ denotes the corresponding jump matrix, 
whose explicit form is given in Lemma~\ref{jumpvj}. 
The function $G(x,t)$ is defined in~\eqref{G}, and the matrices 
$\mathcal{A}$ and $\mathcal{B}$ are given in~\eqref{mathcal-A} 
and~\eqref{mathcal-B}, respectively.}
	
	%
	\begin{theorem}\label{RHth}
		Let $u(x,t)$ be a solution belonging to the Schwartz class of the initial value problem given in \eqref{TT}, defined for an existence time {$T \in (0, \infty]$} with initial data $u_0, u_1 \in \mathcal{S}(\mathbb{R})$ satisfying Assumptions \ref{solitonless}. Define the reflection coefficients $r_1(\lambda)$ and $r_2(\lambda)$ with respect to $u_0, u_1$ as per \eqref{r1r2def}.
		It is then established that the Riemann-Hilbert problem \ref{MRHp} admits a unique solution $M(x,t,\lambda)$ for each point in the domain $(x,t) \in \mathbb{R} \times [0,T)$. Furthermore, the solution $u(x,t)$ of Tzitz\'eica equation for all $(x,t) \in \mathbb{R} \times [0,T)$ can be expressed by
		\begin{equation}\label{usolution}
			u(x,t)=\lim_{\lambda \to 0} \log \left[ (\omega,\omega^2,1) M(x,t,\lambda) \right]_{13},
		\end{equation}
		where $\omega=e^{2i\pi/3}$.
	\end{theorem}
	
	\begin{proof}
		See Section \ref{uissolution}.
	\end{proof}
	
	\begin{lemma}
		Under the assumptions of Theorem \ref{RHth}, the solution $M=M(x,t,\lambda)$ of RH problem \ref{MRHp} is unique, if it exits.
	\end{lemma}
	
	\begin{proof}
		Since the determinant
		of solution $M(x,t;\lambda)$ for RH problem \ref{MRHp} is analytic for $\lambda\in\CC\setminus\{0\}$ with a removable singularity at $\lambda=0$ and $M=I$ as $\lambda\to\infty$, it implies that $\det M=1$ and $M$ is invertible. Consequently, if $M'$ is another solution of RH problem \ref{MRHp}, $M'M^{-1}$ is analytic for $\lambda\in\CC$ and tends to $I$ as $\lambda\to\infty$, it follows that $M'M^{-1}=I$, which denotes that $M'=M$.
	\end{proof}
	Having established the intricate link between the solutions of the Tzitz\'eica equations, framed by Schwartz class initial value conditions, and the Riemann-Hilbert problem, we can now delve into the study of the long-time asymptotic behaviors of the solutions. This constitutes the third significant result obtained in this paper, as depicted below:

	\begin{theorem}\label{uasy}
		{Under the assumptions of Theorem \ref{RHth} and that the
		solution $u(x,t)$ exists globally in time, i.e., $T=+\infty$.} Then the solution $u(x,t)$ to the Tzitz\'eica equation in \eqref{TT}, as defined in \eqref{usolution}, exhibits the following asymptotic behaviors as time $t \to \infty$ (see Figure \ref{sector} for the asymptotic sectors {\rm I-IV} in the $x$-$0$-$t$ half-plane):
		\par
		\noindent {\bf Sector I \& II}: If $\left|\frac{x}{t}\right|\geq 1$, as $t\to\infty$, the function $u(x,t)$ rapidly tends to zero. Specifically, in Sector II where $1\leq \left|\frac{x}{t}\right| < \infty$, the solution $u(x,t)$ behaves as {$\mathcal{O}(t^{-N})$} for sufficiently large $t$. Conversely, as $\left|\frac{x}{t}\right|$ tends to $\infty$ in Sector I, the solution $u(x,t)$ behaves as { $\mathcal{O}(|x|^{-N})$} for sufficiently large $|x|$, where $c$ is a positive constant.
		\par		
	    \noindent{ {\bf Sector III}: If $\left|\frac{x}{t}\right|$ tends to $1$ from within $\left|\frac{x}{t}\right|<1$, it represents a transition region and for any integer $N>1$, the leading-order term of $u(x,t)$ is same with (\ref{zm}), but with an error of $\mathcal{O}\left({|t|^{-N}+\frac{C_N(\lambda_0)\ln t}{t}}\right)$ where $C_N(\lambda)$ is a nonnegative smooth function, which vanishes to any order at $\lambda=0$ and $\lambda=\infty$.}
		\par		
		\noindent {\bf Sector IV}:  If $\left|\frac{x}{t}\right|< 1$, then the long-time asymptotics of the Tzitz\'eica equation in problem \eqref{TT} is
		\begin{equation}\label{zm}
			\small{
				\begin{aligned}
					\small u(x,t)=\log(1+3^{-\frac{1}{4}}&\sqrt{\frac{2(1+\lambda_0^2)}{t\lambda_0}}\left(\sqrt{\nu_1}\cos\left(\frac{5\pi i}{12}-\frac{2\sqrt{3}t\lambda_0}{1+\lambda_0^2}-\nu_1\ln\left(\frac{6\sqrt{3}t\lambda_0}{1+\lambda_0^2}\right)+s_1\right)\right.\\
					&\small \left.\quad -\sqrt{\nu_4}\cos\left(\frac{13\pi i}{12}-\frac{2\sqrt{3}t\lambda_0}{1+\lambda_0^2}
					-\nu_4\ln\left(\frac{6\sqrt{3}t\lambda_0}{1+\lambda_0^2}\right)+s_2\right)\right)+\mathcal O\left(\frac{\ln t}{t}\right),
			\end{aligned}}
		\end{equation}
		where $s_1$ and $s_2$ are specified by \eqref{s}, while $\lambda_0$, $\nu_1$, and $\nu_4$ are defined by \eqref{lambda0def}, \eqref{nu1def}, and \eqref{nu4def}, respectively.
	\end{theorem}
	
	\begin{proof}
		See Lemma \ref{lemma-tends-to-0}, Section \ref{SectorI} and Section \ref{SectorII}.
	\end{proof}
	\begin{figure}[!h]
		\centering
		\begin{overpic}[height=.3\textwidth,width=.65\textwidth]{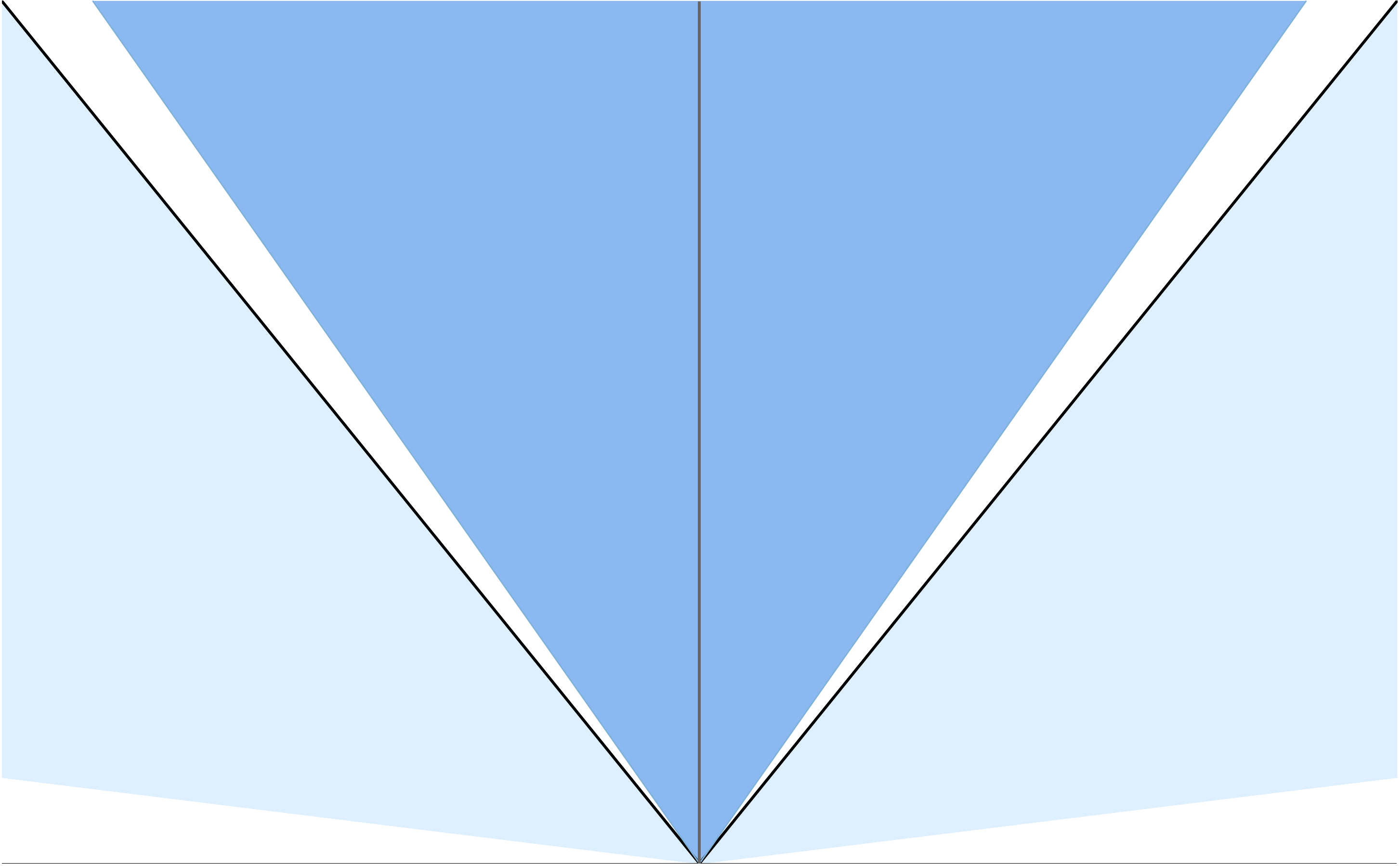}
			\put(8,0.5){${\rm I}$}
			\put(90,0.5){${\rm I}$}
			\put(100,-0.6){$x$}
			\put(100,-0.6){$x$}
			\put(85,15.5){${\rm II}$}
			\put(15,15.5){${\rm II}$}
			\put(88,39){${\rm III}$}
			\put(7,39){${\rm III}$}
			\put(48,32){${\rm IV}$}
			\put(50,47){$t$}
			\put(49,-5){$0$}
			\put(50,0){\vector(0,1){46.5}}
			\put(50,0){\vector(1,0){50}}
			\put(50,0){\line(-1,0){50}}
		\end{overpic}
		\caption{{\protect\small
				The asymptotic sectors I-IV in the $x$-$0$-$t$ half-plane.}}
		\label{sector}
	\end{figure}
	{\begin{lemma}\label{lemma-tends-to-0}
		As $\left|\frac{x}{t}\right|\to 1 $, from the inside of the light cone, the solution $u(x,t)$ of sector $\rm III$ in Theorem \ref{uasy} can match the long-time asymptotic formula of sector $\rm II$.
	\end{lemma}}

	\begin{proof}
		Consider the scenarios as $\frac{x}{t} \to 1^-$ for $x > 0$ and $\frac{x}{t} \to -1^+$ for $x < 0$. In these limits, we analyze the behavior of the critical point $\lambda_0 = \sqrt{\frac{x+t}{x-t}}$ given in (\ref{lambda0def}).
		
		\textbf{Case 1:} $\frac{x}{t} \to 1^-$.
		
		As $\frac{x}{t} \to 1^-$, the expression for $\lambda_0$ approaches to zero, i.e., $\lambda_0 \to 0$. This behavior impacts the leading-order term of the asymptotic expression (\ref{zm}), which is $\sqrt{\frac{2(1+\lambda_0^2)\nu_j}{t\lambda_0}}$ for $j = 1, 4$. As $\lambda_0 \to 0$, the reflection coefficients $r_1(\lambda)$ and $r_2(\lambda)$ tend to zero rapidly. Consequently, $u(x,t)$ vanishes in this limit, which can be observed from the decay properties of the kernel in (\ref{Fredholm integral equation}) and the fact that $\sqrt{\frac{2(1+\lambda_0^2)\nu_j}{t\lambda_0}} \to 0$ for $j = 1, 4$.
		
		\textbf{Case 2:} $\frac{x}{t} \to -1^+$.
		
		For this case, $\lambda_0$ approaches to infinity, i.e., $\lambda_0 \to \infty$. Similar to the previous case, $u(x,t)$ also vanishes as the reflection coefficients decay rapidly when $\lambda \to \infty$.
		\par	
		Moreover, for the {\bf Case 1}, the jump matrix on $(-\lambda_0,\lambda_0)$ effectively diminishes and tends to the identity matrix. This simplification results in the rational decomposition of reflection coefficient $r_1(\lambda)$, i.e., $r_{1,a}$ and $r_{1,r}$ approaching to zero near the critical points $\pm \lambda_0$. Consequently, the jump matrix on $(-\lambda_0,\lambda_0)$ can be approximated as $I + \mathcal{O}(\frac{1}{|t|^l})$, with $l \geq 1$, indicating that the solution in Sector III with $x>0$ converges to the trivial solution as expected from the decay properties of $r_1(\lambda)$ and $r_2(\lambda)$ near the boundary $\lambda \to 0$. For the {\bf Case 2}, following the same procedure as {\bf Case 1}, the solution in Sector III with $x<0$ converges to the trivial solution near the boundary $\lambda \to \infty$.
	\end{proof}
	
	\par
	\noindent{\bf Numerical results.}
	To validate the accuracy of Theorem \ref{uasy}, we introduce an initial value problem characterized by a Gaussian wave packet, specifically defined as
	\begin{equation}\label{eq:initial_data}
		u_0 = u(x,0) = -\frac{1}{10} e^{-\frac{x^2}{2}}, \quad \text{and} \quad u_t(x,0) = 0.
	\end{equation}
	This initial condition ensures that the reflection coefficients meet the requirements of Assumption \ref{solitonless}, namely, $s_{11}(\lambda)\neq 0$ and $s_{11}^A(\lambda)\neq 0$ for $\lambda \neq 0$.
	\par		
	Figures \ref{t20} and \ref{t50} depict the comparisons between the asymptotic solutions posited in Theorem \ref{uasy} and the outcomes derived from numerical simulations with the initial conditions specified in \eqref{eq:initial_data} at times $t = 20$ and $t = 50$, respectively. These figures illustrate the asymptotic predictions with dashed red lines, whereas the numerical results are presented through solid blue lines. These visual comparisons affirm that the large-time asymptotic solutions provide a close approximation to the numerical solutions within an acceptable margin of error.		
	Additionally, an analysis of Figures \ref{t20} and \ref{t50} reveals that for $|x|>t$ the solution $u(x,t)$ in Sector I approaches zero. This behavior corroborates the theoretical anticipation for Sector I, where the solution is expected to decay rapidly as $t$ escalates.
	\par		
	In summary, these numerical investigations reinforce the validity of Theorem \ref{uasy}, underscoring the reliability and precision of the asymptotic expressions delineated therein.
	%
	%

	\begin{figure}[h]
		\centering
		\includegraphics[width=0.9\textwidth]{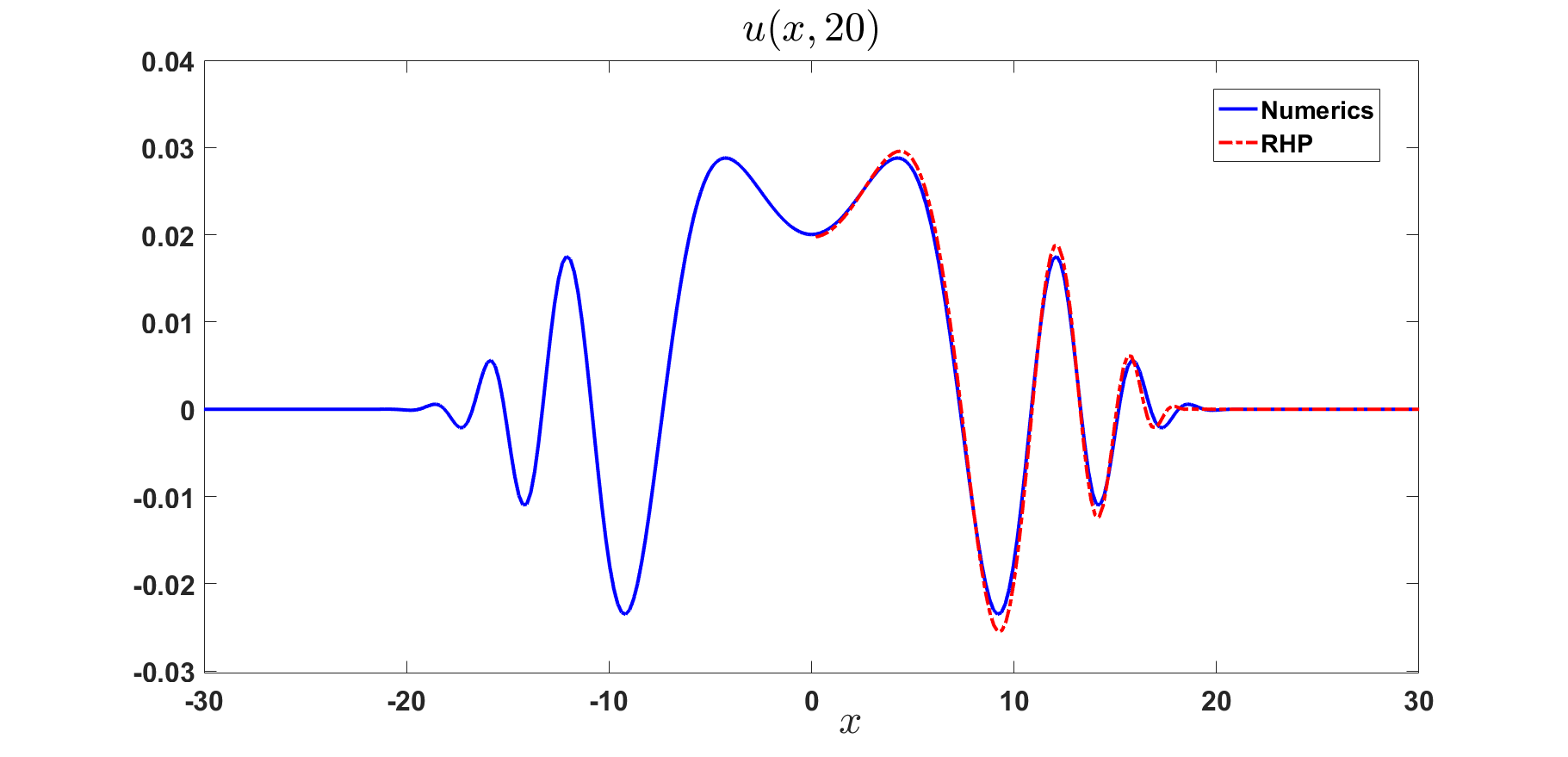}
		\caption{{\small
				The comparison of theoretical results given in Theorem \ref{uasy} and
				full numerical simulations of the Tzitz\'eica equation (\ref{Tt}) with initial condition
				(\ref{eq:initial_data}) at time $t = 20$.}}
		\label{t20}
	\end{figure}
	\begin{figure}[h]
		\centering
		\includegraphics[width=0.9\textwidth]{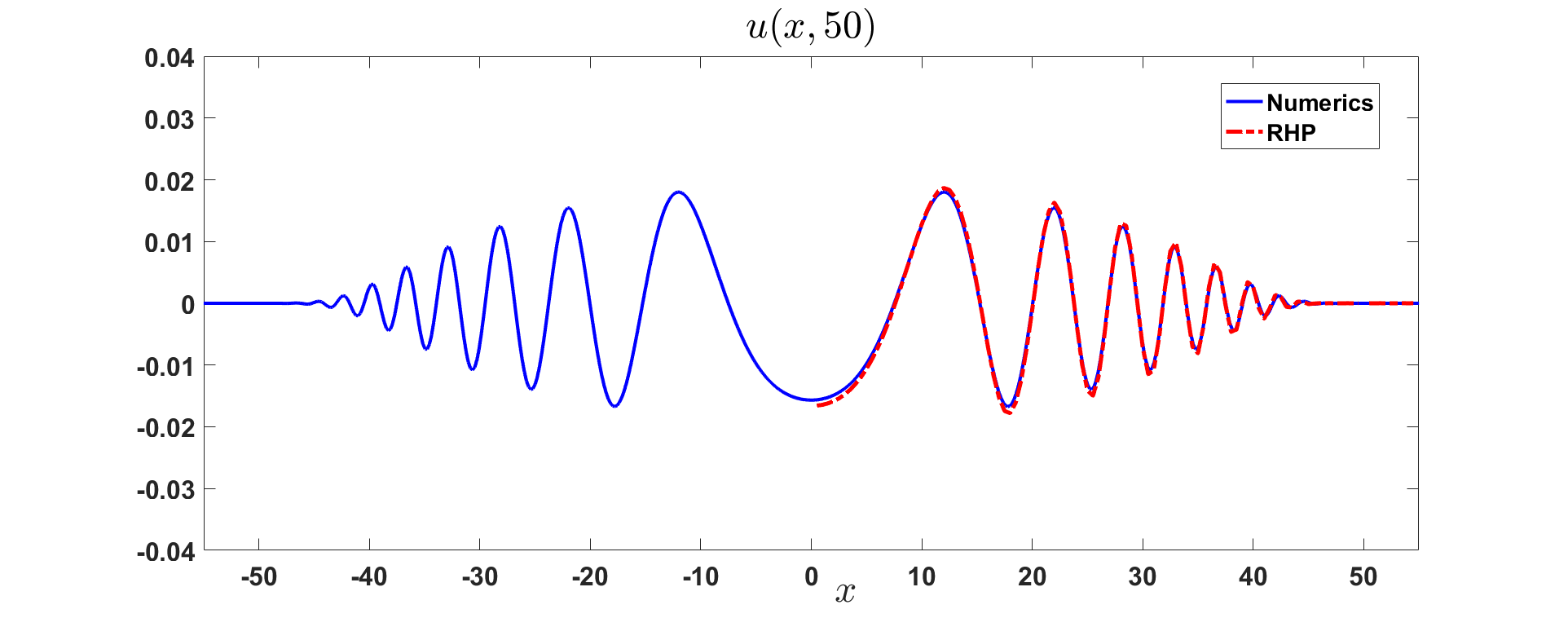}
		\caption{{\small
				The comparison of theoretical results given in Theorem \ref{uasy} and
				full numerical simulations of the Tzitz\'eica equation (\ref{Tt}) with initial condition
				(\ref{eq:initial_data}) at time $t = 50$.}}
		\label{t50}
	\end{figure}

	\section{\bf Spectral analysis}\label{spectralanalysis}
	
	This section focuses on the spectral analysis and inverse scattering transform of the Tzitz\'eica equation \eqref{Tt} based on its Lax pair.

	\subsection{Lax pair}
	The Tzitz\'eica equation \eqref{Tt} has the Lax pair of the form
	\begin{equation}\label{lax pair}
		\left\{\begin{array}{l}
			\phi_x(x, t, \lambda)=L(x, t, \lambda) \phi(x, t, \lambda), \\
			\phi_t(x, t, \lambda)=Z(x, t, \lambda) \phi(x, t, \lambda),
		\end{array}\right.
	\end{equation}
	where the matrices $L(x, t, \lambda)$ and $Z(x, t, \lambda)$ satisfy
	\begin{equation}\label{lax_space0}
		L(x, t, \lambda)=\frac{\lambda}{2} J+U_0+\frac{1}{\lambda} U_1,
	\end{equation}
	\begin{equation}\label{lax_time0}
		Z(x, t, \lambda)=\frac{\lambda}{2} J+U_0-\frac{1}{ \lambda} U_1,
	\end{equation}
	with
	\begin{align*}
		J=\left(\begin{array}{ccc}
			\omega & 0 & 0 \\
			0 & \omega^2 & 0 \\
			0 & 0 & 1
		\end{array}\right),\
		U_0=\frac{i \sqrt{3}\left(u_x+u_t\right)}{6}\left(\begin{array}{ccc}
			0 & 1 & -1 \\
			-1 & 0 & 1 \\
			1 & -1 & 0
		\end{array}\right),\\
		U_1=\frac{1}{6}\left(\begin{array}{ccc}
			\omega^2\left(2 e^u+e^{-2 u}\right) & e^{-2 u}-e^u & \omega\left(e^{-2 u}-e^u\right) \\
			e^{-2 u}-e^u & \omega\left(2 e^u+e^{-2 u}\right) & \omega^2\left(e^{-2 u}-e^u\right) \\
			\omega\left(e^{-2 u}-e^u\right) & \omega^2\left(e^{-2 u}-e^u\right) & 2 e^u+e^{-2 u}
		\end{array}\right).
	\end{align*}
	
	\subsection{Direct scattering problem}
	Let us consider the expressions for $l_j$ and $z_j$
	defined as follows: \begin{align*}
		l_j=\frac{\omega^j\lambda+(\omega^j\lambda)^{-1}}{2},\quad z_j=\frac{\omega^j\lambda-(\omega^j\lambda)^{-1}}{2},\quad j=1,2,3,
	\end{align*}
	and denote $\mathcal{L}=\text{diag}(l_1,l_2,l_3)=\frac{1}{2}(\lambda J+\frac{J^2}{\lambda})$ and $\mathcal{Z}=\text{diag}(z_1,z_2,z_3)=\frac{1}{2}(\lambda J-\frac{J^2}{\lambda})$. The matrix functions $L$ and $Z$ in  (\ref{lax pair})  can be written as follows
	$$
	L:=\mathcal{L}+L_1,\quad Z:=\mathcal{Z}+Z_1,
	$$
	where $L_1,Z_1$ are given by
	$$
	L_1=U_0+\frac{U_1-\frac{1}{2}J^2}{\lambda},\ Z_1=U_0-\frac{U_1+\frac{1}{2}J^2}{\lambda}.
	$$
	Since $u_0,u_1 \in \mathcal{S}(\R)$, it is straightforward to check that the matrices \(L_1\) and \(Z_1\) have the following asymptotic properties
	\begin{equation}
		\lim_{|x| \rightarrow \infty} L_1 = \lim_{|x| \rightarrow \infty} Z_1 = 0.
	\end{equation}
	Noticing that $ u$ is a real function, it can be directly verified that the matrix-valued functions $L$ and $Z$ satisfy the $\mathbb{Z}_3$ symmetry
	\begin{equation}\label{mathcal-A}
		L(\lambda)=\mathcal{A}^{-1} L(\omega \lambda) \mathcal{A}, \quad Z(\lambda)=\mathcal{A}^{-1} Z(\omega \lambda) \mathcal{A}, \quad \mathcal{A}=\left(\begin{array}{ccc}
			0 & 1 & 0 \\
			0 & 0 & 1 \\
			1 & 0 & 0
		\end{array}\right),
	\end{equation}
	and	the $\mathbb{Z}_2$ symmetry
	\begin{equation}\label{mathcal-B}
		L(\lambda)=\mathcal{B} \overline{L(\bar{\lambda})} \mathcal{B}^{-1}, \quad Z(\lambda)=\mathcal{B} \overline{Z(\bar{\lambda})} \mathcal{B}^{-1}, \quad \mathcal{B}=\left(\begin{array}{ccc}
			0 & 1 & 0 \\
			1 & 0 & 0 \\
			0 & 0 & 1
		\end{array}\right).
	\end{equation}
	To facilitate the analysis, we introduce the eigenfunction \(\Phi\) defined by the transformation
	\[
	\phi = \Phi e^{\mathcal{L} x + \mathcal{Z} t},
	\]
	then the Lax pair \eqref{lax pair} can be rewritten as
	\begin{align}\label{lax equation}
		\Phi_x-[\mathcal{L}, \Phi]=L_1 \Phi, \\
		\Phi_t-[\mathcal{Z}, \Phi]=Z_1 \Phi.
	\end{align}
	\par
	The solutions to the equation (\ref{lax equation}) can be formalized through the introduction of two \(3 \times 3\) matrix-valued functions which are determined by solving specific linear Volterra integral equations. Explicitly, define \(\Phi_{+}(x, \lambda)\) and \(\Phi_{-}(x, \lambda)\) as follows
	\begin{equation}\label{voltera0}
		\begin{aligned}
			& \Phi_{+}(x, \lambda)=I-\int_x^{\infty} e^{(x-y) \widehat{\mathcal{L}(\lambda)}}\left(L_1(y, \lambda) \Phi_{+}(y, \lambda)\right) d y, \\
			& \Phi_{-}(x, \lambda)=I+\int_{-\infty}^x e^{(x-y) \widehat{\mathcal{L}(\lambda)}}\left(L_1(y, \lambda) \Phi_{-}(y, \lambda)\right) d y,
		\end{aligned}
	\end{equation}
	where $e^{\hat{\mathcal{L}}}$ is an operator that acts on a $3\times 3$ matrix $A$ by $e^{\hat{\mathcal{L}}}A=e^{ \mathcal{L}}Ae^{-\mathcal{L}}$.
	\par
	Furthermore, decompose the complex plane by $\Sigma$, $i.e.$,
	\begin{align}\label{Sigmadef}
		\Sigma:=\{\lambda\in\mathbb{C}|Re(l_i)=Re(l_j),1\le i<j\le3\},
	\end{align}
	in fact, $\Sigma$ can also be rewritten as
	$\Sigma:=\{\R,\omega\R,\omega^2\R\}$.
	Notice that $\Sigma$ divides the complex plane $\CC$ into open sets $\{{D}_j\}_{j=1}^6$ (see Figure \ref{sigma0}) and suppose
	$S:=\{\lambda\in\mathbb{C}|0<arg(\lambda)<\frac{2\pi}{3}\}.$
	\begin{figure}[!h]
		\centering
		\begin{overpic}[width=.55\textwidth]{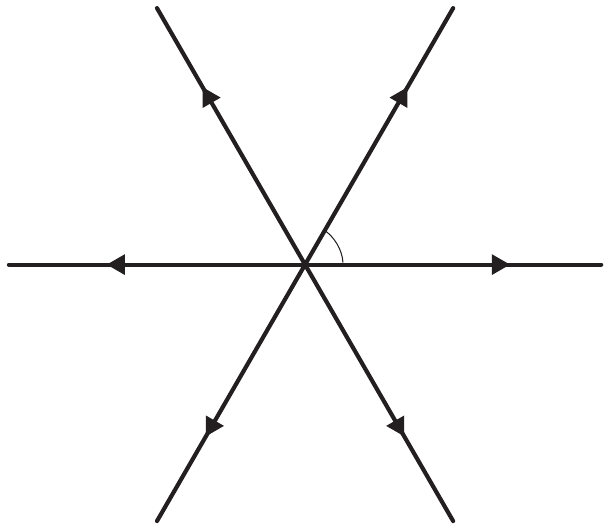}
			\put(57,48){$\frac{\pi}{3}$}
			\put(103,42){$\Sigma$}
			\put(75,55){${D}_1$}
			\put(48,70){${D}_2$}
			\put(23,55){${D}_3$}
			\put(23,25){${D}_4$}
			\put(48,15){${D}_5$}
			\put(75,25){${D}_6$}
		\end{overpic}
		\caption{{\protect\small
				The contour $\Sigma$ decomposes the $\lambda$ plane into six parts.}}
		\label{sigma0}
	\end{figure}
	
	\begin{proposition}
		Suppose the initial data $u_0(x),u_1(x)\in\mathcal{S}(\R)$, then the matrix-valued Jost functions $\Phi_+(x,\lambda)$ and $\Phi_-(x,\lambda)$ have the properties:
		\begin{enumerate}
			\item  $\Phi_+(x,\lambda)$ is well-defined in the closure of  $(S,\omega^2 S,\omega S)\setminus\{0\}$, and $\Phi_-(x,\lambda)$ is well-defined in the closure of  $(-S,-\omega^2 S,-\omega S)\setminus\{0\}$. Moreover, $\Phi_+(\cdot,\lambda)$ and $\Phi_-(\cdot,\lambda)$ are smooth and rapidly decay in the closure of their domains with determinant equal to 1.
			
			\item $\Phi_+(x,\cdot)$ and $\Phi_-(x,\cdot)$ are analytic in interior of their domains, and any order partial derivative of $\lambda$ can be continuous to the closure of their domains.
			
			\item $\Phi_+(x,\lambda)$  and $\Phi_-(x,\lambda)$ satisfy the following symmetric:
			$$
			\begin{aligned}
				\Phi_{\pm}(x, \lambda)=\mathcal{A}^{-1} \Phi_{\pm}(x, \omega \lambda) \mathcal{A}=\mathcal{B} {\Phi_{\pm}^*(x, {\lambda}^*)} \mathcal{B}^{-1},
			\end{aligned}
			$$
			with $\lambda$ in their domains.
			
			\item Assuming that the initial conditions \(u_0\) and \(u_1\) have compact support, the functions \(\Phi_+\) and \(\Phi_-\) are well-defined and analytic over the complex plane \(\mathbb{C} \setminus \{0\}\).
		\end{enumerate}
	\end{proposition}
	\begin{proof}
		The proof of this proposition is a straightforward analysis of the Volterra integral equations in \eqref{voltera0}.
	\end{proof}
	\begin{remark}
		For a comprehensive analysis and foundational methodologies, the reader is referred to seminal works such as those by Charlier and Lenells \cite{Charlier-Lenells-2021}, Huang and Lenells \cite{HLNonlinearFourier}, and further elaborations in \cite{Ds-Xd-2022}. At first inspection, the kernel of the integral equation in (\ref{voltera0}) appears to possess singularities at \(\lambda = 0\). This situation is reminiscent of, yet distinct from, the scenario encountered in the analysis of the Boussinesq equation, where the kernel exhibits second order poles at \(\lambda = 0\) as detailed in \cite{Charlier-Lenells-2021}. Notably, in the context of the Tzitz\'eica equation, the ensuing section will demonstrate that the Jost solutions can be analytically continued to \(\lambda = 0\), thereby indicating a significant divergence in behavior when compared to the solution characteristics of Boussinesq equation.
	\end{remark}

	\subsubsection{The behavior of $Jost$ functions as $\lambda\to\infty$.}
	Next consider the properties of the eigenfunctions \(\Phi_+\) and \(\Phi_-\) as $\lambda \to +\infty$. For the Lax equation (\ref{lax equation}) for \(\Phi\), we proceed by conducting a Wentzel--Kramers--Brillouin (WKB) expansion of \(\Phi\) as \(\lambda \rightarrow \infty\), which yields the series
	\[
	\Phi_{\pm} = I + \frac{\Phi^{(1)}_{\pm}}{\lambda} + \frac{\Phi^{(2)}_{\pm}}{\lambda^2} + \cdots,
	\]
	where the coefficients \(\Phi^{(n)}_{\pm}\) are determined by the recursive formula
	\[
	\left\{\begin{array}{l}
		\left[\frac{J}{2}, \Phi^{(n+1)}_{\pm}\right] + \left[\frac{J^2}{2}, \Phi^{(n-1)}_{\pm}\right] = (\partial_x \Phi^{(n)}_{\pm})^{(o)} - \left(\mathrm{U}_0 \Phi^{(n)}_{\pm}\right)^{(o)} - \left(\mathrm{U}_1 \Phi^{(n-1)}_{\pm}\right)^{(o)}, \\
		(\partial_x \Phi^{(n)}_{\pm})^{(d)} = \left(\mathrm{U}_0 \Phi^{(n)}_{\pm}\right)^{(d)} + \left(\mathrm{U}_1 \Phi^{(n-1)}_{\pm}\right)^{(d)},
	\end{array}\right.
	\]
	where \(\Phi^{(0)}_{\pm} = I\) and \(\Phi^{(-1)}_{\pm} = \mathbf{0}\). The notation \((\star)^{(o)}\) denotes the off-diagonal part of a matrix, whereas \((\star)^{(d)}\) refers to the diagonal part.
	\par	
	The following proposition describes the asymptotic properties of the Jost functions \(\Phi_+\) and \(\Phi_-\)  when $\lambda$ tends to infinity.
	
	
	\begin{proposition}
		Suppose $u_0, u_1 \in \mathcal{S}(\mathbb{R})$ and as $\lambda \rightarrow \infty$, the functions $\Phi_+$ and $ \Phi_-$ coincide to all orders with their expansions, respectively. Precisely, for an integer $p \geq 0$, the functions
		$$
		\begin{aligned}
			& (\Phi_+)_{(p)}(x, \lambda):=I+\frac{\Phi^{(1)}_{+}(x)}{\lambda}+\cdots+\frac{\Phi^{(p)}_{+}(x)}{\lambda^p}, \\
			& (\Phi_-)_{(p)}(x, \lambda):=I+\frac{\Phi^{(1)}_{-}(x)}{\lambda}+\cdots+\frac{\Phi^{(p)}_{-}(x)}{\lambda^p}, \\
		\end{aligned}
		$$
		are well-defined and, for each integer $j \geq 0$,
		$$
		\begin{aligned}
			& \left|\frac{\partial^j}{\partial \lambda^j}\left(\Phi_+-(\Phi_+)_{(p)}\right)\right| \leq \frac{f_{+}(x)}{|\lambda|^{p+1}}, \quad x \in \mathbb{R},\quad \lambda \in\left( \overline{\mathrm{S}}, \omega^2 \overline{\mathrm{S}}, \omega\overline{\mathrm{S}}\right),\quad|\lambda| \geq 2, \\
			& \left|\frac{\partial^j}{\partial \lambda^j}\left(\Phi_--(\Phi_-)_{(p)}\right)\right| \leq \frac{f_{-}(x)}{|\lambda|^{p+1}}, \quad x \in \mathbb{R},\quad \lambda \in\left(- \overline{\mathrm{S}},-\omega^2 \overline{\mathrm{S}},-\omega \overline{\mathrm{S}}\right),\quad|\lambda| \geq 2,
		\end{aligned}
		$$
		where positive functions $f_{+}(x)$ and $f_{-}(x)$ are bounded and smooth for $x \in \mathbb{R}$, and they rapidly decay as $x \rightarrow+\infty$ and $x \rightarrow-\infty$, respectively.
	\end{proposition}
	\begin{proof}
		This proof is similar to the case of the sine-Gordon equation and is omitted here, see \cite{Charlier-Lenells-2021,HLNonlinearFourier}.
	\end{proof}
	\subsubsection{The behavior of $Jost$ functions as $\lambda\to0$.}
	The presence of a pole at \(\lambda = 0\) in the kernel of the Jost functions signifies that the Volterra integral equation framework is not ideally suited for analyzing the behavior of Jost functions as \(\lambda \to 0\). To circumvent this limitation, we propose a gauge transformation defined by \(\phi(x, t; \lambda) = G(x,t) \psi(x, t; \lambda)\), where \(G(x, t)\) is given by
	\begin{equation}\label{G}
		G(x, t) = \frac{1 + e^u + e^{2u}}{3e^u} \left(\begin{array}{ccc}
			1 & \frac{\omega(e^u - 1)}{e^u - \omega^2} & \frac{\omega^2(e^u - 1)}{e^u - \omega} \\
			\frac{\omega^2(e^u - 1)}{e^u - \omega} & 1 & \frac{\omega(e^u - 1)}{e^u - \omega^2} \\
			\frac{\omega(e^u - 1)}{e^u - \omega^2} & \frac{\omega^2(e^u - 1)}{e^u - \omega} & 1
		\end{array}\right).
	\end{equation}
	The function $\psi(x,t;\lambda)$ satisfies a new Lax pair as following
	\begin{align*}
		\begin{cases}
			\psi_x(x, t;\lambda)=G^{-1}(LG-G_x)\psi(x, t;\lambda)=\tilde L(x,t;\lambda)\psi(x, t;\lambda), \\
			\psi_t(x, t;\lambda)=G^{-1}(ZG-G_t) \psi(x, t;\lambda)=\tilde Z(x,t;\lambda)\psi(x, t;\lambda),
		\end{cases}
	\end{align*}
	where
	\begin{align*}
		\tilde{L}(x, t, \lambda)=\frac{1}{2 \lambda} J^2+\tilde U_0+\lambda \tilde U_1,\
		\tilde{Z}(x, t, \lambda)=\frac{1}{2 \lambda} J^2+\tilde U_0-\lambda \tilde U_1,
	\end{align*}
	and
	\begin{align*}
		&\tilde U_0=\frac{i \sqrt{3}\left(u_x-u_t\right)}{6}\left(\begin{array}{ccc}
			0 & 1 & -1 \\
			-1 & 0 & 1 \\
			1 & -1 & 0
		\end{array}\right),\\
		&\tilde U_1=\frac{\lambda}{6}\left(\begin{array}{ccc}
			\omega\left(2 e^u+e^{-2 u}\right) & e^{-2 u}-e^u & \omega^2\left(e^{-2 u}-e^u\right) \\
			e^{-2 u}-e^u & \omega^2\left(2 e^u+e^{-2 u}\right) & \omega\left(e^{-2 u}-e^u\right) \\
			\omega^2\left(e^{-2 u}-e^u\right) & \omega\left(e^{-2 u}-e^u\right) & 2 e^u+e^{-2 u}
		\end{array}\right).
	\end{align*}
	Hence one has
	$$
	\tilde L\sim \frac{1}{2}(\frac{J^2}{\lambda}+\lambda J),\quad \tilde Z\sim \frac{1}{2}(\frac{J^2}{\lambda}+\lambda J), \quad x\to\pm\infty,
	$$
	and suppose that
	\begin{align*}
		&\tilde{\mathcal{L}}=\frac{1}{2}(\frac{J^2}{\lambda}+\lambda J),\quad
		\tilde{\mathcal{Z}}=\frac{1}{2}(\frac{J^2}{\lambda}-\lambda J),\\
		&\tilde L_1:=\tilde U_0+\lambda \tilde U_1,\quad
		\tilde Z_1:=\tilde U_0-\lambda \tilde U_1.
	\end{align*}
	\par
	Introducing a transformation $\psi = \Psi e^{\mathcal{L}x + \mathcal{Z}t}$, we derive a new Lax pair for $\Psi$ of the form
	\begin{equation}\label{new-Lax-pair}
		\left\{\begin{array}{l}
			\Psi_x-[\mathcal{L}, \Psi]=\tilde{L}_1 \Psi, \\
			\Psi_t-[\mathcal{Z}, \Psi]=\tilde{Z}_1 \Psi.
		\end{array}\right.
	\end{equation}
	To further analyze the function $\Psi$, formulate the Volterra integral equations as follows
	\begin{align*}
		& \Psi_{+}(x, \lambda)=I-\int_x^{\infty} e^{(x-y) \widehat{\mathcal{L}(\lambda)}}\left(\tilde L_1(y, \lambda) \Psi_{+}(y, \lambda)\right) d y, \\
		& \Psi_{-}(x, \lambda)=I+\int_{-\infty}^x e^{(x-y) \widehat{\mathcal{L}(\lambda)}}\left(\tilde L_1(y, \lambda) \Psi_{-}(y, \lambda)\right) d y.
	\end{align*}
	\noindent This formulation ensures that the kernel of the Volterra integral equations is not affected by the poles as $\lambda \to 0$. Consequently, we propose a formal expansion of function $\Psi$ for $\lambda \to 0$, given by
	$$
	\Psi_{\pm}(x,\lambda)=\Psi_{\pm}^{(0)}(x)+\Psi_{\pm}^{(1)}(x)\lambda+\cdots,\quad \text{as} \ \lambda\to0.
	$$
	Combining the formal expansion above with the $x$-part of Lax pair (\ref{new-Lax-pair}), we arrive at the following system of equations
	$$
	\left\{\begin{array}{l}
		{\left[\frac{J}{2}, \Psi^{(n-1)}_{\pm}\right]+\left[\frac{J^2}{2}, \Psi^{(n+1)}_{\pm}\right]=(\partial_x  \Psi^{(n)}_{\pm})^{(o)}-\left( \tilde{U}_0  \Psi^{(n)}_{\pm}\right)^{(o)}-\left( \tilde{U}_1  \Psi^{(n-1)}_{\pm}\right)^{(o)},} \\
		(\partial_x  \Psi^{(n)}_{\pm})^{(d)}=\left( \tilde{U}_0 \Psi^{(n)}_{\pm}\right)^{(d)}+\left( \tilde{U}_1 \Psi^{(n-1)}_{\pm}\right)^{(d)},
	\end{array}\right.
	$$
	where $\Psi_{\pm}^{(0)}(x) = I$ and $\Psi_{\pm}^{(-1)}(x) = 0$.
	%
	In this context, the behavior of $\Phi(x, \lambda)$ as $\lambda \to 0$ is described through the transformation $\Phi = G\Psi$, leading to the expansion:
	$$
	\Phi_{\pm}(x,\lambda)=G(x)+G\Psi_{\pm}^{(1)}(x)\lambda+\cdots,\quad \text{as} \ \lambda\to0.
	$$
	\par
	\par
	\begin{remark} Here are several properties of the function $G(x,t)$:
		\begin{itemize}
			
			{\item  The function $G(x,t)$ satisfies the symmetries
			$$
			\mathcal{A}^{-1}G(x,t)\mathcal{A}=G,\quad \mathcal{B}G(x,t)\mathcal{B}=G(x,t).
			$$}
			\item 		Notice that
			\begin{align*}
				\left(\begin{array}{lll}\omega & \omega^2 & 1\end{array}\right) G(x, t)&=e^u\left(\begin{array}{lll}\omega & \omega^2 & 1\end{array}\right),\\
				\left(\begin{array}{lll}\omega & \omega^2 & 1\end{array}\right) G(x, t)^{-1}&=e^{-u}\left(\begin{array}{lll}\omega & \omega^2 & 1\end{array}\right).
			\end{align*}	
			Due to these properties, the reconstruction of the solution to the Tzitz\'eica equation as \(\lambda\) approaches zero becomes a viable endeavor.
			
		\end{itemize}
	\end{remark}

	\begin{remark}
		In contrast to the scenario encountered with the Boussinesq equation, the Jost functions pertaining to the isospectral problem for the Tzitz\'eica equation exhibit regularity at \(\lambda = 0\). This distinction underscores a fundamental difference in the analytical properties of the solutions associated with these two equations, highlighting the unique behavior of the Tzitz\'eica equation in the context of singularities at the origin of the spectral parameter.
		
	\end{remark}
	\subsection{The scattering matrix}
	Suppose the initial potential functions \(u_0\) and \(u_1\) in (\ref{TT}) have compact support, then there exists a matrix function \(s(\lambda)\) dependent of $x$ and $t$ such that the following relationship holds
	$$
	\Phi_+(x, \lambda)=\Phi_-(x, \lambda) e^{x \widehat{\mathcal{L}(\lambda)}} s(\lambda), \quad \lambda \in \mathbb{C}\setminus\{0\},
	$$
	which characterizes the analytic continuation of the scattering data across the complex plane, excluding the origin. In the more general case where \(\{u_0,u_1\} \in \mathcal{S}(\mathbb{R})\), the aforementioned relationship for \(s(\lambda)\) remains valid within its domain.
	\begin{proposition}\label{sprop}
		Suppose $u_0,u_1 \in\mathcal{S}(\R)$, then the matrix function $s(\lambda)$  has the following properties:
		\begin{enumerate}
			\item  The domain of function $s(\lambda)$ is
			$$
			\left(\begin{array}{ccc}
				\bar S & \mathbb{R}_{+} & \omega \mathbb{R}_{+} \\
				\mathbb{R}_{+} & \omega^2 \bar S & \omega^{2} \mathbb{R}_{+} \\
				\omega \mathbb{R}_{+} & \omega^{2} \mathbb{R}_{+} & \omega\bar{S}
			\end{array}\right)\setminus\{0\},
			$$
			where $\overline{S}$ means the closure of set $S$, and $\R_{+}$  and $\R_-$ represent the positive and negative real axis respectively. The function $s(\lambda)$ is continuous to the boundary of its domain but only analytic in the interior of its domain.
			
			\item The behaviors of $s(\lambda)$ as $\lambda\to\infty$ and $\lambda\to0$, respectively are
			$$
			s(\lambda)=I-\sum_{j=1}^{N} \frac{s_{j}}{\lambda^{j}}+\mathcal{O}\left(\frac{1}{\lambda^{N+1}}\right),\quad \lambda \rightarrow \infty,
			$$
			and
			$$
			s(\lambda)=I+s^{(1)} \lambda+\cdots, \quad \lambda \rightarrow 0.
			$$
			
			\item  Function $s(\lambda)$ satisfies the symmetries
			$$
			s(\lambda)=\mathcal{A}^{-1} s(\omega \lambda) \mathcal{A}=\mathcal{B} s^*({\lambda}^*) \mathcal{B}^{-1}.
			$$
		\end{enumerate}
	\end{proposition}
	
	\begin{proof}
		Let us consider the  \(\lambda\) lies within the intersection of the relevant domain. The relationship between the Jost functions \(\Phi_+\) and \(\Phi_-\) can be expressed as:
		$$
		\Phi_+(x, \lambda)=\Phi_-(x, \lambda) e^{x \widehat{\mathcal{L}(\lambda)}} s(\lambda).
		$$
		It is noteworthy that as \(x \to -\infty\), \(\Phi_-(x, \lambda)\) converges rapidly towards the identity matrix \(I\). This observation allows us to represent \(s(\lambda)\) in the following integral
		\begin{equation}\label{Scattering-Ma}
			s(\lambda)=I-\int_{-\infty}^{\infty} e^{-y \widehat{\mathcal{L}(\lambda)}}\left(L_1(y, \lambda) \Phi_{+}(y, \lambda)\right) d y.
		\end{equation}
		Through a detailed analysis of the exponential terms \(e^{l_i - l_j}\), it is directly to show that
		$$
		\lambda \in \left(\begin{array}{ccc}
			\bar S & \mathbb{R}_{+} & \omega \mathbb{R}_{+} \\
			\mathbb{R}_{+} & \omega^2 \bar S & \omega^{2} \mathbb{R}_{+} \\
			\omega \mathbb{R}_{+} & \omega^{2} \mathbb{R}_{+} & \omega\bar{S}
		\end{array}\right)\setminus\{0\}.
		$$
		Furthermore, the analysis of the behavior of \(\Phi(x, \lambda)\) as \(\lambda \to 0\), and letting $x=0$, it shows that
		$$
		G(0)+G(0)\Psi_{+}^{(1)}(0)\lambda+\mathcal{O}(\lambda)=(G(0)+G(0)\Psi_{-}^{(1)}(0)\lambda+\mathcal{O}(\lambda))s(\lambda)
		$$
		which allows us to find the expansion of $s(\lambda)$ as $\lambda\to 0$. Consequently, the expansion of \(s(\lambda)\) is given by
		$$
		s(\lambda)=I+s^{(1)}\lambda+\cdots\quad \lambda\to0,
		$$
		where the coefficients \(s^{(j)}\) are determined by the relation $s^{(j)}=\left(\Psi_-^{(j)}\right)^{-1}\Psi_+^{(j+1)}$.
	\end{proof}
	\begin{remark}
		In fact, by the third property of scattering matrix \(s(\lambda)\), the domain of $s(\lambda)$ can be analytic continuous to the point $\lambda=0$.
	\end{remark}
	\subsection{The cofactor matrix}
	
	Let \(M^A\) denote the cofactor of matrix \(M\), defined by \(M^A = (M^{-1})^T\). Utilizing the space part of Lax pair presented in \eqref{lax equation}, we derive a corresponding Lax representation for the cofactor matrix associated with the Jost functions, given by
	\[
	\Phi_x^A + [\mathcal{L}, \Phi^A] = -L_1^T \Phi^A.
	\]
	To express the Volterra integral equations for \(\Phi^A\), analogous formulations to those previously discussed are employed, which yields
	\[
	\begin{aligned}
		& \Phi_{+}^A(x, \lambda) = I + \int_x^{\infty} e^{-(x-y) \widehat{\mathcal{L}(\lambda)}}\left(L_1^T \Phi_{+}^A\right)(y, \lambda) \, dy, \\
		& \Phi_{-}^A(x, \lambda) = I - \int_{-\infty}^x e^{-(x-y) \widehat{\mathcal{L}(\lambda)}}\left(L_1^T \Phi_{-}^A\right)(y, \lambda) \, dy.
	\end{aligned}
	\]
	The application of a similar analysis to these Volterra integral equations for \(\Phi^A\) as conducted for \(\Phi\) leads to a new proposition regarding the Jost functions, denoted \(\Phi^A_{\pm}\).

	
	\begin{proposition}
		Suppose the initial data $u_0(x),u_1(x)\in\mathcal{S}(x)$, then cofactor matrix of Jost functions $\Phi_+^A(x,\lambda)$ and $\Phi_-^A(x,\lambda)$ have the properties:
		\begin{enumerate}
			\item $\Phi_+^A(x,\lambda)$ is well-defined in the closure of  $(-S,-\omega^2 S,-\omega S)\setminus\{0\}$, and $\Phi_-^A(x,\lambda)$ is well-defined in the closure of  $(S,\omega^2 S,\omega S)\setminus\{0\}$, the determinant of $\Phi^A_{\pm}$ are always equal to $1$, and $\Phi_+^A(\cdot,\lambda)$ and $\Phi_-^A(\cdot,\lambda)$ are smooth and rapidly decay in the closure of  their domains.
			
			\item  $\Phi_+^A(x,\cdot)$ and $\Phi_-^A(x,\cdot)$ are analytic in interior of  their domains, but $k$-order derivative of them can be continuous to the closure of their domains.
			
			\item  $\Phi_+^A(x,\lambda)$  and $\Phi_-^A(x,\lambda)$ satisfy the following symmetries
			$$
			\begin{aligned}
				\Phi_{\pm}^A(x, \lambda)=\mathcal{A}^{-1} \Phi_{\pm}^A(x, \omega \lambda) \mathcal{A}=\mathcal{B} {(\Phi_{\pm}^A)^*(x, {\lambda}^*)} \mathcal{B}^{-1},
			\end{aligned}
			$$
			where $\lambda$ is in the domains of $\Phi_+^A(x,\lambda)$  and $\Phi_-^A(x,\lambda)$.
			
			\item  Assuming that the initial conditions \(u_0\) and \(u_1\) have compact support,  $\Phi_+^A$ and $\Phi_-^A$ are well-defined and analytic for $\lambda\in \CC\setminus\{0\}$.
		\end{enumerate}
	\end{proposition}

	\begin{proposition}
		Suppose $u_0, u_1 \in \mathcal{S}(\mathbb{R})$ and as $\lambda \rightarrow \infty, \Phi_+^A$ and $ \Phi_-^A$ coincide to any orders with their expansions like {$(\Phi_{\pm}^A)_{(p)}$}, respectively. Precisely, for an integer $p \geq 0$, the functions
		$$
		\begin{aligned}
			& (\Phi_+^A)_{(p)}(x, \lambda):=I+\frac{(\Phi^A_{+})^{(1)}(x)}{\lambda}+\cdots+\frac{(\Phi^A_{+})^{(p)}(x)}{\lambda^p}, \\
			& (\Phi_-^A)_{(p)}(x, \lambda):=I+\frac{(\Phi^A_{-})^{(1)}(x)}{\lambda}+\cdots+\frac{(\Phi^A_{-})^{(p)}(x)}{\lambda^p}, \\
		\end{aligned}
		$$
		are well-defined and, for any integer $j \geq 0$, one has
		$$
		\begin{aligned}
			& \left|\frac{\partial^j}{\partial \lambda^j}\left(\Phi_+^A-(\Phi_+^A)_{(p)}\right)\right| \leq \frac{g_{+}(x)}{|\lambda|^{p+1}}, \quad x \in \mathbb{R},\quad \lambda \in\left( -\overline{\mathrm{S}}, -\omega^2 \overline{\mathrm{S}}, -\omega\overline{\mathrm{S}}\right),\quad |\lambda| \geq 2, \\
			& \left|\frac{\partial^j}{\partial \lambda^j}\left(\Phi_-^A-(\Phi_-^A)_{(p)}\right)\right| \leq \frac{g_{-}(x)}{|\lambda|^{p+1}}, \quad x \in \mathbb{R},\quad \lambda \in\left( \overline{\mathrm{S}},\omega^2 \overline{\mathrm{S}},\omega\overline{\mathrm{S}}\right),\quad |\lambda| \geq 2,
		\end{aligned}
		$$
		where the positive functions $g_{+}(x)$ and $g_{-}(x)$ are bounded and smooth on $\mathbb{R}$, and decay rapidly as $x \to +\infty$ and $x \to -\infty$, respectively.
	\end{proposition}
	
	By the definition of $M^A$, it follows naturally that \(\Phi^A = G^A \Psi^A\). Consequently, the function \(\Psi^A\) is governed by the equation
	$$
	\Psi_x^A+[\mathcal{L}, \Psi^A]=-(G^{-1}LG-G^{-1}G_x-\mathcal{L}^T)\Psi^A=-\tilde L_1^T \Psi^A.
	$$
	The Volterra integral equations for the cofactor matrix functions $\Psi^A$ are expressed by
	$$
	\begin{aligned}
		& \Psi_{+}^A(x, \lambda)=I+\int_x^{\infty} e^{-(x-y) \widehat{\mathcal{L}(\lambda)}}\left(\tilde L_1^T \Psi_{+}^A\right)(y, \lambda) d y, \\
		& \Psi_{-}^A(x, \lambda)=I-\int_{-\infty}^x e^{-(x-y) \widehat{\mathcal{L}(\lambda)}}\left(\tilde L_1^T \Psi_{-}^A\right)(y, \lambda) d y.
	\end{aligned}
	$$
	Analogous to the analysis for \(\Psi_{\pm}\) as \(\lambda \to 0\), the expansion for \(\Psi^A_{\pm}\) near \(\lambda = 0\) is given by
	{
	$$
	\Psi_{\pm}^A(x,\lambda)=I+(\Psi_{\pm}^A)^{(1)}(x)\lambda+\cdots,\quad \text{as} \ \lambda\to 0.
	$$}
	
	The analysis of the scattering matrix \(s^A(\lambda)\) parallels that of \(s(\lambda)\), with the relationship between the Jost functions \(\Phi_+^A\) and \(\Phi_-^A\) given by
	$$
	\Phi_+^A(x, \lambda)=\Phi_-^A(x, \lambda) e^{-x \widehat{\mathcal{L}(\lambda)}} s^A(\lambda), \quad \lambda \in \mathbb{C}\setminus\{0\},
	$$
	where the matrix-valued function $s^A(\lambda)$ is
	\begin{equation}\label{Scattering-Ma-A}
		s^A(\lambda)=I+\int_{-\infty}^{\infty} e^{y \widehat{\mathcal{L}(\lambda)}}\left(L_1^T(y, \lambda) \Phi_{+}^A(y, \lambda)\right) d y.
	\end{equation}
	\par	
	In the following proposition, we provide a brief overview of the relevant properties of \(s^A(\lambda)\). The detailed derivation and proof process are omitted.
	\begin{proposition}\label{sAprop}
		Let us consider the initial data \(u_0, u_1 \in \mathcal{S}(\mathbb{R})\), then the matrix function \(s^A(\lambda)\) exhibits the ensuing properties:
		\begin{enumerate}
			\item  The domain of function $s^A(\lambda)$ is
			$$
			\left(\begin{array}{ccc}
				-\bar S & \mathbb{R}_{-} & \omega \mathbb{R}_{-} \\
				\mathbb{R}_{-} & -\omega^2 \bar S & \omega^{2} \mathbb{R}_{-} \\
				\omega \mathbb{R}_{-} & \omega^{2} \mathbb{R}_{-} & -\omega\bar{S}
			\end{array}\right)\setminus\{0\}.
			$$
			\item  The behaviors of $s^A(\lambda)$ as $\lambda\to\infty$ and $\lambda\to0$, respectively, are
			$$
			s^A(\lambda)=I-\sum_{j=1}^{N} \frac{s_{j}^A}{\lambda^{j}}+\mathcal{O}\left(\frac{1}{\lambda^{N+1}}\right),\quad \lambda \rightarrow \infty,
			$$
			and
			$$
			s^A(\lambda)=I+(s^A)^{(1)} \lambda+\cdots, \quad \lambda \rightarrow 0.
			$$
			\item The matrix function $s^A(\lambda)$ satisfies the symmetries
			$$
			s^A(\lambda)=\mathcal{A}^{-1} s^A(\omega \lambda) \mathcal{A}=\mathcal{B} (s^A)^*({\lambda}^*) \mathcal{B}^{-1}.
			$$
		\end{enumerate}
	\end{proposition}
	\subsection{The eigenfunctions $M_n$}
	In this subsection, we construct the piecewise analytic eigenfunctions $M_n$ restricted in the region $D_n=\{\lambda:\frac{(n-1)\pi}{3}<\arg(\lambda)<\frac{n\pi}{3}\}$ for $n=1,2,\cdots,6$.
	\par
	For $\lambda\in D_n$ and $j=1,2,3$, define a $3\times3$ matrix function by Fredholm integral equation
	\begin{equation}\label{Fredholm integral equation}
		\left(M_n\right)_{i j}(x, \lambda)=\delta_{i j}+\int_{\gamma_{i j}^n}\left(e^{\left(x-y\right) \widehat{\mathcal{L}(\lambda)}}\left(L_1 M_n\right)\left(y, \lambda\right)\right)_{i j} d y,
	\end{equation}
	where the contours $\gamma_{i j}^n$ for $i, j=1,2,3$ and $n=1,2, \ldots, 6$, are defined as
	$$
	\gamma_{i j}^n=\left\{\begin{array}{ll}
		(-\infty, x), & \operatorname{Re} l_i(\lambda)<\operatorname{Re} l_j(\lambda), \\
		(+\infty, x), & \operatorname{Re} l_i(\lambda) \geq \operatorname{Re} l_j(\lambda),
	\end{array} \quad \text { for } \quad \lambda \in D_n.\right.
	$$
	Denote the zeros of the Fredholm determinants related to the Fredholm integral equations as $\mathcal{N}$  and suppose $\tilde {\mathcal N}:=\mathcal N\cup\{0\}$. Moreover, the $3\times3$ matrix-valued function defined by (\ref{Fredholm integral equation}) has the following propositions.
	
	\begin{proposition}
		Assume $u_0,u_1\in\mathcal{S}(\R)$, then Fredholm integral equation (\ref{Fredholm integral equation}) defines  $3\times3$ matrix valued functions $M_n$ that have the following properties:
		\begin{enumerate}
			\item  The matrix-valued function $M_n(x,\lambda)$ is defined in $\R\times \bar D_n\setminus\tilde{\mathcal{N}}$ . Furthermore, $M(\cdot,\lambda)$ is a smooth function for $\lambda\in D_n\setminus\tilde{\mathcal{N}}$ and satisfies the Lax pair.
			\item  For any $x\in\R$, the matrix-valued function $M_n(x,\cdot)$ is analytic in $\lambda\in D_n\setminus\tilde{\mathcal{N}}$ and continuous on $\lambda\in \bar D_n\setminus\tilde{\mathcal{N}}$.
			\item  $\forall \epsilon>0$, $M_n(x,\lambda)$ is a bounded matrix function for $x\in\R$ and $\lambda\in D_n,\text{dist}(\lambda,\tilde{\mathcal{N}})\ge\epsilon$.
			\item  $\det M_n(x,\lambda)=1$ for $x\in\R$ and $\lambda\in \bar D_n\setminus{\tilde{\mathcal N}}$.
			\item  $\forall x\in\R$, define the piecewise analytic matrix function $M(x,\lambda):=M_n(x,\lambda)$ for $\lambda\in D_n$ and the matrix function $M(x,\lambda)$ is suited to the following symmetries
			$$
			M(x, \lambda)=\mathcal{A}^{-1} M(x, \omega \lambda) \mathcal{A}=\mathcal{B} {M^*(x, {\lambda}^*)} \mathcal{B}^{-1},\quad \lambda\in\CC\setminus{\tilde{\mathcal{N}}}.
			$$
		\end{enumerate}
	\end{proposition}

	\begin{proof}
	Without loss of generality, we focus on the first column of $M_1(x;\lambda)$ and denote by $m_i(x;\lambda)$ the {$(i,1)$--entry} of $M_1(x;\lambda)$.
		$$
		m(x,\lambda)_{i}=(\delta)_{i1}+\int_{\R}\sum_{j=1}^3K(x,y;\lambda)_{ij}m_j(y,\lambda)dy,
		$$
		where
		$$
		K\left(x, y; \lambda\right)_{i j}=\pm H(x-y)e^{(x-y)(l_i-l_1)}(L_1)_{i j},
		$$
		the $\pm\infty$ with respect to the choice of $\gamma_{ij}$,
		and $|L_1(x)|$ could be dominated by a Schwartz function $b(x)$. Moreover, the Fredholm determinant associated to the first column is
		$$
		f(\lambda)=\sum_{m=0}^{\infty}\frac{(-1)^m}{m !} \sum_{i_1, i_2, \ldots, i_m=1}^3 \int_{\mathbb{R}^m}  K^{(m)}\left(\begin{array}{l}
			x_1, i_1, x_2, i_2, \ldots, x_m, i_m \\
			x_1, i_1, x_2, i_2, \cdots, x_m, i_m
		\end{array} ; \lambda\right) d x_1 d x_2 \cdots d x_m
		$$
		with
		$$
		K^{(m)}\left(\begin{array}{l}
			x_1, i_1, x_2, i_2, \ldots, x_m, i_m \\
			y_1, i_1^{\prime}, y_2, i_2^{\prime}, \cdots, y_m, i_m^{\prime}
		\end{array} ; \lambda\right)=\operatorname{det}\left(\begin{array}{ccc}
			K\left(x_1, y_1, \lambda\right)_{i_1 i_1^{\prime}} & \cdots & K\left(x_1, y_m, \lambda\right)_{i_1 i_m^{\prime}} \\
			\vdots & & \vdots \\
			K\left(x_m, y_1, \lambda\right)_{i_m i_1^{\prime}} & \cdots & K\left(x_m, y_m, \lambda\right)_{i_m i_m^{\prime}}
		\end{array}\right).
		$$
		\par	
		Using the Hardamard's inequality, it is immediate to know that
		$$
		|K^{(m)}|\le m^{m/2}\Pi_{j=1}^m b(y_j).
		$$
		Consequently, the  Fredholm determinant $f$ related to the first column of (\ref{Fredholm integral equation}) is a analytic function for $\lambda\in {D}_1\setminus{\tilde{\mathcal{N}}}$ and continuous to the closure of $ {D}_1\setminus{\tilde{\mathcal{N}}}$. Furthermore, since the potential matrix $\tilde L_1$ is $\mathcal{O}(1)$ as $\lambda\to\infty$, the Fredholm determinant $f$ is bounded as $\lambda\to\infty$ so that the zeros of  $f$ is finite in the region $D_1$. 		
		
	\end{proof}
	
	\begin{lemma}\label{lemmaST}
		Suppose $u_0,u_1\in\mathcal{S}(\R)$ with compact support, then
		$$
		\begin{aligned}
			M_n(x, \lambda) & =\Phi_{-}(x, \lambda) e^{x \mathcal{L}(\lambda)} S_n(\lambda) e^{-x \mathcal{L}(\lambda)} \\
			& =\Phi_{+}(x, \lambda) e^{x \mathcal{L}(\lambda)} T_n(\lambda) e^{-x \mathcal{L}(\lambda)}, \quad n=1,2, \ldots, 6,
		\end{aligned}
		$$
		where the matrices $T_n(\lambda)$ and $S_n(\lambda)$ can be expressed by the elements of the matrix $s(\lambda)$ and the $(ij)$-th minor $m_{ij}(s)$ of $s(\lambda)$, which are expressed by
		$$
		\begin{array}{rlr}
			S_1(\lambda)=\left(\begin{array}{ccc}
				s_{11} & 0 & 0 \\
				s_{21} & \frac{m_{33}(s)}{s_{11}} & 0 \\
				s_{31} & \frac{m_{23}(s)}{s_{11}} & \frac{1}{m_{33}(s)}
			\end{array}\right),
			& S_2(\lambda)=\left(\begin{array}{ccc}
				s_{11} & 0 & 0 \\
				s_{21} & \frac{1}{m_{22}(s)} & \frac{m_{32}(s)}{s_{11}} \\
				s_{31} & 0 & \frac{m_{22}(s)}{s_{11}}
			\end{array}\right), \\
			S_3(\lambda)=\left(\begin{array}{ccc}
				\frac{m_{22}(s)}{s_{33}} & 0 & s_{13} \\
				\frac{m_{12}(s)}{s_{33}} & \frac{1}{m_{22}(s)} & s_{23} \\
				0 & 0 & s_{33}
			\end{array}\right),
			& S_4(\lambda)=\left(\begin{array}{ccc}
				\frac{1}{m_{11}(s)} & \frac{m_{21}(s)}{s_{33}} & s_{13} \\
				0 & \frac{m_{11}(s)}{s_{33}} & s_{23} \\
				0 & 0 & s_{33}
			\end{array}\right), \\
			S_5(\lambda)=\left(\begin{array}{ccc}
				\frac{1}{m_{11}(s)} & s_{12} & -\frac{m_{31}(s)}{s_{22}} \\
				0 & s_{22} & 0 \\
				0 & s_{32} & \frac{m_{11}(s)}{s_{22}}
			\end{array}\right),
			& S_6(\lambda)=\left(\begin{array}{ccc}
				\frac{m_{33}(s)}{s_{22}} & s_{12} & 0 \\
				0 & s_{22} & 0 \\
				-\frac{m_{13}(s)}{s_{22}} & s_{32} & \frac{1}{m_{33}(s)}
			\end{array}\right),
		\end{array}
		$$
		and
		$$
		\begin{array}{lll}
			T_1(\lambda)=\left(\begin{array}{ccc}
				1 & -\frac{s_{12}}{s_{11}} & \frac{m_{31}(s)}{m_{33}(s)} \\
				0 & 1 & -\frac{m_{32}(s)}{m_{33}(s)} \\
				0 & 0 & 1
			\end{array}\right),
			& T_2(\lambda)=\left(\begin{array}{ccc}
				1 & -\frac{m_{21}(s)}{m_{22}(s)} & -\frac{s_{13}}{s_{11}} \\
				0 & 1 & 0 \\
				0 & -\frac{m_{23}(s)}{m_{22}(s)} & 1
			\end{array}\right),
		\end{array}
		$$
		$$
		\begin{array}{lll}
			T_3(\lambda)=\left(\begin{array}{ccc}
				1 & -\frac{m_{21}(s)}{m_{22}(s)} & 0 \\
				0 & 1 & 0 \\
				-\frac{s_{31}}{s_{33}} & -\frac{m_{23}(s)}{m_{22}(s)} & 1
			\end{array}\right),
			& T_4(\lambda)=\left(\begin{array}{ccc}
				1 & 0 & 0 \\
				-\frac{m_{12}(s)}{m_{11}(s)} & 1 & 0 \\
				\frac{m_{13}(s)}{m_{11}(s)} & -\frac{s_{32}}{s_{33}} & 1
			\end{array}\right), \\
			T_5(\lambda)=\left(\begin{array}{ccc}
				1 & 0 & 0 \\
				-\frac{m_{12}(s)}{m_{11}(s)} & 1 & -\frac{s_{23}}{s_{22}} \\
				\frac{m_{13}(s)}{m_{11}(s)} & 0 & 1
			\end{array}\right),
			& T_6(\lambda)=\left(\begin{array}{ccc}
				1 & 0 & \frac{m_{31}(s)}{m_{33}(s)} \\
				-\frac{s_{21}}{s_{22}} & 1 & -\frac{m_{32}(s)}{m_{33}(s)} \\
				0 & 0 & 1
			\end{array}\right).
		\end{array}
		$$
	\end{lemma}
	\begin{proof}
		Notice that
		\begin{align}\label{Sn}
			\left\{\begin{array}{l}
				S_n(\lambda)=\lim _{x \rightarrow-\infty} e^{-x \widehat{\mathcal{L}(\lambda)}} M_n(x, \lambda), \\
				T_n(\lambda)=\lim _{x \rightarrow \infty} e^{-x \widehat{\mathcal{L}(\lambda)}} M_n(x, \lambda),
			\end{array} \quad \lambda \in \bar{D}_n \backslash \tilde{\mathcal{N}},\right.
		\end{align}
		and
		$$
		s(\lambda) T_n(\lambda)=S_n(\lambda).
		$$
		Furthermore, it is obvious that
		$$
		\begin{aligned}
			& \left(S_n(\lambda)\right)_{i j}=0 \quad \text { if } \quad \gamma_{i j}^n=(-\infty, x), \\
			& \left(T_n(\lambda)\right)_{i j}=\delta_{i j} \quad \text { if } \quad \gamma_{i j}^n=(\infty, x).
		\end{aligned}
		$$
		The equality expressed above gives rise to nine distinct algebraic equations. For illustrative purposes, we shall focus exclusively on the region denoted as \(D_1\). This analysis involves the matrix \(\gamma^1\), as detailed in equation \eqref{Fredholm integral equation} below
		$$
		\gamma^1=\left(\begin{array}{lll}
			\gamma_2 & \gamma_1 & \gamma_1 \\
			\gamma_2 & \gamma_2 & \gamma_1 \\
			\gamma_2 & \gamma_2 & \gamma_2
		\end{array}\right),
		$$
		where $\gamma_1=(-\infty, x)$ and $\gamma_2=(+\infty, x)$. Thus the equation \eqref{Sn} for $n=1$ shows that
		$$
		S_1=\lim _{x \rightarrow-\infty} e^{-x \widehat{L(\lambda)}} M_1(x, \lambda)=\left(\begin{array}{ccc}
			x_{11} & 0 & 0 \\
			x_{21} & x_{22} & 0 \\
			x_{31} & x_{32} & x_{33}
		\end{array}\right),
		$$
		and
		$$
		T_1=\lim _{x \rightarrow \infty} e^{-x \widehat{L(\lambda)}} M_1(x, \lambda)=\left(\begin{array}{ccc}
			1 & y_{12} & y_{13} \\
			0 & 1 & y_{23} \\
			0 & 0 & 1
		\end{array}\right).
		$$
		The proof for the remaining cases $S_i(\lambda), T_i(\lambda)$ for $i=2,3,\cdots 5$ follows a methodology analogous to that employed for the previous discussions.
	\end{proof}

	\begin{lemma}
		Let $u_0, u_1 \in \mathcal{S}(\mathbb{R})$ and
		denote by \(\left\{s(\lambda), M_n(x, \lambda)\right\}\) the scattering data and eigenfunctions related with the \((u_0, u_1)\). For any $u_0,u_1\in \mathcal{S}(\mathbb{R})$, there are two function sequences $u_0^{(i)},u_1^{(i)}\in C_0^{\infty}(\R)$, which converge to \(u_0, u_1\) uniformly, respectively. Then it follows
		$$
		\begin{aligned}
			& \lim _{i \rightarrow \infty} s_{i}(\lambda)=s(\lambda), \quad \lambda \in\left(\begin{array}{ccc}
				\overline{\mathrm{S}} & \mathbb{R}_{+} & \omega \mathbb{R}_{+} \\
				\mathbb{R}_{+} & \omega^2 \overline{\mathrm{S}} & \omega^2 \mathbb{R}_{+} \\
				\omega \mathbb{R}_{+} & \omega \mathbb{R}_{+} & \omega^2\overline{\mathrm{S}}
			\end{array}\right)\backslash\{0\}, \\
			& \lim _{i \rightarrow \infty}\left(s^A\right)_{i}(\lambda)=s^A(\lambda), \quad \lambda \in\left(\begin{array}{ccc}
				- \overline{\mathrm{S}} & \mathbb{R}_{-} & \omega \mathbb{R}_{-} \\
				\mathbb{R}_{-} & -\omega^2 \overline{\mathrm{S}} & \omega^2 \mathbb{R}_{-} \\
				\omega \mathbb{R}_{-} & \omega^2 \mathbb{R}_{-} & -\omega^2\overline{\mathrm{S}}
			\end{array}\right)\backslash\{0\},
		\end{aligned}
		$$

		$$
		\begin{aligned}
			& \lim _{i \rightarrow \infty} \Phi_+^{i}(x, \lambda)=\Phi_+(x, \lambda), \quad x \in \mathbb{R}, \lambda \in\left( \overline{\mathrm{S}}, \omega^2 \overline{\mathrm{S}}, \omega \overline{\mathrm{S}}\right) \backslash\{0\}, \\
			& \lim _{i \rightarrow \infty} \Phi_-^{i}(x, \lambda)=\Phi_-(x, \lambda), \quad x \in \mathbb{R}, \lambda \in\left(-\overline{\mathrm{S}},-\omega^2 \overline{\mathrm{S}},-\omega^2\overline{\mathrm{S}}\right) \backslash\{0\}, \\
			& \lim _{i \rightarrow \infty} M_n^{i}(x, \lambda)=M_n(x, \lambda), \quad x \in \mathbb{R}, k \in \bar{D}_n \backslash \tilde{\mathcal{N}},\quad n=1,2, \ldots, 6.
		\end{aligned}
		$$
	\end{lemma}
	\begin{lemma}\label{jumpvj}
		The jump matrices for the matrix functions $M_n(x,t,\lambda)$ associated with a Riemann-Hilbert problem are given by
		$$
		\begin{aligned}
			&\small v_1=\left(\begin{array}{ccc}
				1 & -r_1(\lambda) e^{-\vartheta_{21}} & 0 \\
				r_1^*(\lambda) e^{\vartheta_{21}} & 1-\left|r_1(\lambda)\right|^2 & 0 \\
				0 & 0 & 1
			\end{array}\right),
			v_2=\left(\begin{array}{ccc}
				1 & 0 & 0 \\
				0 & 1-r_2(\omega \lambda) r_2^*(\omega \lambda) & -r_2^*(\omega \lambda) e^{-\vartheta_{32}} \\
				0 & r_2(\omega \lambda) e^{\vartheta_{32}} & 1
			\end{array}\right), \\
			&\small v_3=\left(\begin{array}{ccc}
				1-r_1\left(\omega^2 \lambda\right) r_1^*\left(\omega^2 \lambda\right) & 0 & r_1^*\left(\omega^2 \lambda\right) e^{-\vartheta_{31}} \\
				0 & 1 & 0 \\
				-r_1\left(\omega^2 \lambda\right) e^{\vartheta_{31}} & 0 & 1
			\end{array}\right),
			v_4=\left(\begin{array}{ccc}
				1-\left|r_2(\lambda)\right|^2 & -r_2^*(\lambda) e^{-\vartheta_{21}} & 0 \\
				r_2(\lambda) e^{\vartheta_{21}} & 1 & 0 \\
				0 & 0 & 1
			\end{array}\right) \text {, } \\
			&\small v_5=\left(\begin{array}{ccc}
				1 & 0 & 0 \\
				0 & 1 & -r_1(\omega \lambda) e^{-\vartheta_{32}} \\
				0 & r_1^*(\omega \lambda) e^{\vartheta_{32}} & 1-r_1(\omega \lambda) r_1^*(\omega \lambda)
			\end{array}\right),
			v_6=\begin{pmatrix}
				1 & 0 & r_2\left(\omega^2 \lambda\right) e^{-\vartheta_{31}} \\
				0 & 1 & 0 \\
				-r_2^*\left(\omega^2 \lambda\right) e^{\vartheta_{31}} & 0 & 1-r_2\left(\omega^2 \lambda\right) r_2^*\left(\omega^2 \lambda\right)
			\end{pmatrix}, \\
			&
		\end{aligned}
		$$
		where $\vartheta_{ij}=(l_i-l_j)x+(z_i-z_j)t$ and $l_j=\frac{\omega^j\lambda+(\omega^j\lambda)^{-1}}{2},z_j=\frac{\omega^j\lambda-(\omega^j\lambda)^{-1}}{2}$.
	\end{lemma}

	\begin{lemma}
		Suppose $u_0, u_1 \in \mathcal{S}(\mathbb{R})$, then the functions $M_n$ and $M_n^A \equiv\left(M_n^{-1}\right)^T$ for $n=1,4$ can be given by the Jost functions $\Phi_{\pm}$, $\Phi_{\pm}^A$ and scattering matrices $s(\lambda)$, $s^A(\lambda)$ as follows
		
		\resizebox{1\textwidth}{!}{
			$
			M_1=\small \left(\begin{array}{ccc}
				(\Phi_+)_{11} & \frac{(\Phi_-^A)_{31} (\Phi_+^A)_{23}-(\Phi_-^A)_{21} (\Phi_+^A)_{33}}{s_{11}} & \frac{(\Phi_-)_{13}}{s_{33}^A} \\
				(\Phi_+)_{21} & \frac{(\Phi_-^A)_{11} (\Phi_+^A)_{33}-(\Phi_-^A)_{31} (\Phi_+^A)_{13}}{s_{11}} & \frac{(\Phi_-)_{23}}{s_{33}^A} \\
				(\Phi_+)_{31} & \frac{(\Phi_-)_{21}^A (\Phi_+)_{13}^A-(\Phi_-)_{11}^A (\Phi_+)_{23}^A}{s_{11}} & \frac{(\Phi_-)_{33}}{s_{33}^A}
			\end{array}\right), \quad
			M_1^A=\small \left(\begin{array}{ccc}
				\frac{(\Phi_-)_{11}^A}{s_{11}} & \frac{(\Phi_+)_{31} (\Phi_-)_{23}-(\Phi_+)_{21} (\Phi_-)_{33}}{s_{33}^A} & (\Phi_+^A)_{13} \\
				\frac{(\Phi_-)_{21}^A}{s_{11}} & \frac{(\Phi_+)_{11} (\Phi_-)_{33}-(\Phi_+)_{31} (\Phi_-)_{13}}{s_{33}^A} & (\Phi_+^A)_{23} \\
				\frac{(\Phi_-)_{31}^A}{s_{11}} & \frac{(\Phi_+)_{21} (\Phi_-)_{13}-(\Phi_+)_{11} (\Phi_-)_{23}}{s_{33}^A} & (\Phi_+^A)_{33}
			\end{array}\right),
			$}
		and
		
		\resizebox{1\textwidth}{!}{
			$
			\small	M_4=\small \left(\begin{array}{ccc}
				\frac{(\Phi_-)_{11}}{s_{11}^A} & \frac{(\Phi_-^A)_{23} (\Phi_+^A)_{31}-(\Phi_+^A)_{21} (\Phi_-^A)_{33}}{s_{33}} & (\Phi_+)_{13} \\
				\frac{(\Phi_-)_{21}}{s_{11}^A} & \frac{(\Phi_+^A)_{11} (\Phi_-^A)_{33}-(\Phi_-^A)_{13} (\Phi_+^A)_{31}}{s_{33}} & (\Phi_+)_{23}\\
				\frac{(\Phi_-)_{31}}{s_{11}^A} & \frac{(\Phi_+^A)_{21} (\Phi_-^A)_{13}-(\Phi_+^A)_{11} (\Phi_-^A)_{23}}{s_{33}} & (\Phi_+)_{33}
			\end{array}\right), \quad
			\small	M_4^A=\small \left(\begin{array}{ccc}
				(\Phi_+^A)_{11} & \frac{(\Phi_+)_{23} (\Phi_-)_{31}-(\Phi_-)_{21} (\Phi_+)_{33}}{s_{11}^A} & \frac{(\Phi_-^A)_{13}}{s_{33}} \\
				(\Phi_+^A)_{21} & \frac{(\Phi_-)_{11} (\Phi_+)_{33}-(\Phi_+)_{13} (\Phi_-)_{31}}{s_{11}^A} & \frac{(\Phi_-^A)_{23}}{s_{33}} \\
				(\Phi_+^A)_{31} & \frac{(\Phi_-)_{21} (\Phi_+)_{13}-(\Phi_-)_{11} (\Phi_+)_{23}}{s_{11}^A} & \frac{(\Phi_-^A)_{33}}{s_{33}}
			\end{array}\right).
			$}
	\end{lemma}

	\begin{remark}
		As the assumption illustrated, the diagonal part of  scattering matrix is not equal to zero and then as $\lambda\to p_j$ with $p_j\in\mathcal{N}$, the matrix-valued function satisfies that
		$$
		(\lambda-p_j)M(x,t;\lambda)=0,\ \text{as}\ \lambda\to p_j.
		$$
		It is immediate to know that $M(x,t;\lambda)$ can be analytically continued to the zeros of Fredholm determination.
	\end{remark}
	
	\begin{lemma}
		Suppose $u_0,u_1\in\mathcal{S}(\R)$, then as $\lambda \to 0$, the behavior of $M(x,t;\lambda)$ has the same expansion with $\Phi_{+}$ as $\lambda\to0$. Specially, the leading term in the expansion is $G(x)$.
	\end{lemma}

	\begin{proof}
		In the process of proof, we extensively employ the methodologies outlined in \cite{Lenells-2018} and \cite{HLNonlinearFourier}.
		\par		
		For convenience, denote $\mathcal{M}_p(x,\lambda):=G(x)+G\Psi_{+}^{(1)}(x)\lambda+\cdots+G\Psi_{+}^{(p)}(x)\lambda^{p+1}$ and then it is sufficient to show that
		$$
		|M(x,\lambda)-\mathcal{M}_p(x,\lambda)|\le C|\lambda|^p, \quad x\in\R,\ \lambda\in\CC\setminus\Gamma,\ |\lambda|<\epsilon,
		$$
		with $\epsilon<1$ small enough.
		
		Let $D_n^{\epsilon}$ be $D_n\cap\{|\lambda|<\epsilon\}$ and then since $G$ is a matrix function in Schwartz space, then the maximal norm of $\mathcal{M}_p$ is less than $1$. As a result, $\mathcal{M}_p^{-1}(x,\lambda)$ exists and is uniformly bounded for $\lambda\in D_n^{\epsilon}$ and $x \in \R$.
		
		Now, assume that $L_{(p)}:=(\partial_x \mathcal{M}_p+\mathcal{M}_p\mathcal{L})\mathcal{M}_p^{-1}$ and $\Delta=L-L_{(p)}$, then we can get the Fredholm integral equation for $\mathcal{M}_p^{-1}M$ as
		$$
		(\mathcal{M}_p^{-1}M)_x=\mathcal{M}_p^{-1}\Delta M+[\mathcal{L},\mathcal{M}_p^{-1}M]
		$$
		heuristically, in which the each entry of $\mathcal{M}_{p}^{-1} M_n$ satisfies the Fredholm equation
		$$
		\left(\mathcal{M}_{p}^{-1} M_n\right)(x, \lambda)_{i j}=\delta_{i j}+\int_{\gamma_{i j}^n}\left(e^{\left(x-y\right) \hat{\mathcal{L}}}\left(M_{(p)}^{-1} \Delta M_n\right)\left(y, \lambda\right)\right)_{i j} d y.
		$$
		Let $H(x)$ be a step function with $H(x)=1$ for $x>0$ and $H(x)=0$ for $x\le0$, and without loss of generality, we focus on the first column of $M_1$ and suppose $m(x,\lambda)_{i}=(M_1)_{i1}$, then rewrite the Fredholm equation as
		$$
		m(x,\lambda)_{i}=(\mathcal{M}_p)_{i1}+\int_{\R}\sum_{j=1}^3K(x,y;\lambda)_{ij}m_j(y,\lambda)dy,
		$$
		where
		$$
		K\left(x, y; \lambda\right)_{i j}=\sum_{s=1}^3 \mathcal{M}_{p}(x, \lambda)_{i s} H_s\left(x, y\right) e^{\left(l_s-l_1\right)\left(x-y\right)}\left(\mathcal{M}_{p}^{-1} \Delta\right)\left(y, \lambda\right)_{s j},
		$$
		and
		$$
		H_s\left(x, y\right)= \begin{cases}H\left(x-y\right) & \text { if } \operatorname{Re} l_i(\lambda)<\operatorname{Re} l_1(\lambda), \\ -H\left(y-x\right) & \text { if } \operatorname{Re} l_i(\lambda)\ge\operatorname{Re} l_1(\lambda).\end{cases}
		$$
		Furthermore, rewrite $\Delta$ as
		$$
		\Delta=\left(L_1 \mathcal{M}_{p}+\left[\mathcal{L}, \mathcal{M}_{p}\right]-\partial_x \mathcal{M}_{p}\right)\mathcal{M}_{p}^{-1}.
		$$
		Since $\mathcal{M}_{p}^{-1}$ is uniformly bounded for $\lambda\in D_{\epsilon}^n$, and noticing that $L_1$ and $G$ are Schwartz matrix-valued functions, it follwos that there is a Schwartz function $\alpha(x)$ such that
		$$
		|\Delta(x,\lambda)|\le C \alpha(x) |\lambda|^{p+1},\quad x\in\R,\ \lambda\in D_n^{\epsilon},
		$$
		and
		$$
		\left|K\left(x, y; \lambda\right)_{i j}\right| \leq {\alpha\left(y\right)}{|\lambda|^{p+1}}, \quad x, y \in \mathbb{R},\quad  \lambda \in {D}_1^{\epsilon}.
		$$
		Let $K^{(0)}=1$ and define
		$$
		K^{(m)}\left(\begin{array}{l}
			x_1, i_1, x_2, i_2, \ldots, x_m, i_m \\
			y_1, i_1^{\prime}, y_2, i_2^{\prime}, \cdots, y_m, i_m^{\prime}
		\end{array} ; \lambda\right)=\operatorname{det}\left(\begin{array}{ccc}
			K\left(x_1, y_1, \lambda\right)_{i_1 i_1^{\prime}} & \cdots & K\left(x_1, y_m, \lambda\right)_{i_1 i_m^{\prime}} \\
			\vdots & & \vdots \\
			K\left(x_m, y_1, \lambda\right)_{i_m i_1^{\prime}} & \cdots & K\left(x_m, y_m, \lambda\right)_{i_m i_m^{\prime}}
		\end{array}\right).
		$$
		\par
		Using the Hardamard's inequality,
		$$
		|\det(A)|^2\le\prod_i^{m}\sum_{j}^{m}|A_{ij}|^2,
		$$
		it is immediate to show that
		$$
		|K^{(m)}|\le \left(\prod_i^{m}\sum_{j}^{m}|K_{ij}^{m}(x_i,y_j,\lambda)|^2\right)^{\frac{1}{2}}
		=m^{\frac{m}{2}}\prod_j^{m}\alpha(y_j)|\lambda|^{p+1}.
		$$
		Define the Fredholm determinant and Fredholm minor as
		$$
		\begin{aligned}
			& f(\lambda)=\sum_{m=0}^{\infty} f^{(m)}(\lambda), \quad \lambda \in {D}_1^\epsilon, \\
			& F\left(x, y, \lambda\right)_{i i^{\prime}}=\sum_{m=0}^{\infty} F^{(m)}\left(x, y, \lambda\right)_{i i^{\prime}}, \quad x,y\in\R,\lambda \in {D}_1^\epsilon,
		\end{aligned}
		$$
		where
		$$
		f^{(m)}(\lambda)=\frac{(-1)^m}{m !} \sum_{i_1, i_2, \ldots, i_m=1}^3 \int_{\mathbb{R}^m}  K^{(m)}\left(\begin{array}{l}
			x_1, i_1, x_2, i_2, \ldots, x_m, i_m \\
			x_1, i_1, x_2, i_2, \cdots, x_m, i_m
		\end{array} ; \lambda\right) d x_1 d x_2 \cdots d x_m,
		$$
		and
		$$
		\small F^{(m)}(x,y,\lambda)_{ii'}=\frac{(-1)^m}{m !} \sum_{i_1, i_2, \ldots, i_m=1}^3 \int_{\mathbb{R}^m}  K^{(m+1)}\left(\begin{array}{l}
			x,i,x_1, i_1,  \ldots, x_m, i_m \\
			y,i',y_1, i_1,\cdots, y_m, i_m
		\end{array} ; \lambda\right) d x_1 d x_2 \cdots d x_m.
		$$
		Therefore, the estimates for $f(\lambda)$ and $F(x,y,\lambda)$ satisfy that
		\begin{align*}
			&\left|f^{(m)}(\lambda)\right| \leq \frac{3^m m^{m / 2}\|\alpha(y)\|_{L^1(\mathbb{R})}^m|\lambda|^{(p+1)m}}{m !}, \quad \lambda \in {D}_1^\epsilon, m \geq 0,\\
			&\left|F^{(m)}\left(x, y, \lambda\right)_{i i^{\prime}}\right| \leq \frac{3^m(m+1)^{(m+1) / 2}\|\alpha\|_{L^1(\mathbb{R})}^m \alpha\left(y\right)|\lambda|^{(p+1)(m+1)}}{m !}, \quad x, y \in \mathbb{R}, \lambda \in D_1^\epsilon, m \geq 0,
		\end{align*}
		which immediately indicates that
		$$
		\begin{aligned}
			& |f(\lambda)-1| \leq {C}{|\lambda|^{p+1}}, \quad \lambda \in {D}_1^{\epsilon}, \\
			& \left|F\left(x, y, k\right)\right| \leq C {\alpha(y)}{|\lambda|^{p+1}}, \quad x, y \in \mathbb{R}, \lambda \in {D}_1^{\epsilon} .
		\end{aligned}
		$$
		Notice that the Fredholm determinant is not zero for $\lambda\in D_1^{\epsilon}$, then by the Fredholm analysis, it yields
		$$
		m(x,\lambda)_{i}=(\mathcal{M}_p)_{i1}+\frac{1}{f(\lambda)}\int_{\R}\sum_{j=1}^3F(x,y;\lambda)_{ij}(\mathcal{M}_p)_{j1}(y,\lambda)dy,
		$$
		and
		$$
		|m(x,\lambda)_i-(\mathcal{M}_p)_{i1}|\le C|\lambda|^{p+1},\quad \lambda\in D_1^{\epsilon}.
		$$
	\end{proof}
	\subsection{Proof of Theorem \ref{RHth}}\label{uissolution}
	Introduce the eigenfunctions $\{M_n(x,t;\lambda)\}_{n=1}^6$ dependent of time by inserting $L_1(x;\lambda)$ into $L_1(x,t;\lambda)$ in the Fredholm integral equations and denote the sectionally analytic function $M(x,t;\lambda)$ as $M(x,t;\lambda)=M_n(x,t;\lambda)$ for $\lambda\in D_n$. Otherwise, the behaviour of $M(x,t;\lambda)$ has the same expansion with $\Phi_+$ which implies that
	\begin{equation}\label{recover formula}
		u(x, t)=\lim _{\lambda \rightarrow 0} \log \left[\left(\omega, \omega^2, 1\right) M(x, t, \lambda)\right]_{13}.
	\end{equation}
	Now, it suffices to show the Theorem \ref{RHth}, if the time-dependent meromorphic function $M(x,t;\lambda)$ satisfies the Lax equation and jump conditions, see Lemma \ref{Lemma-formula-1} and Lemma \ref{Lemma-formula-2} below.
	\begin{lemma}\label{Lemma-formula-1}
		For $\lambda \in D_n$, the function $M_n(x,t;\lambda)$ is smooth for $(x,t)\in\R\times[0,T)$, and satisfies the Lax pair (\ref{lax equation}).
	\end{lemma}
	\begin{proof}
		Denote {$\hat{M}_n(x,t;\lambda):=M_n(x,t;\lambda)e^{\mathcal{L}x+\mathcal{Z}t}$}, it follows that {$\hat{M}_n(x,t;\lambda)$} satisfies the space part of Lax pair and according to the compatibility conditions, we have $(\partial_t\hat{M}_n-Z\hat{M}_n)_x=L(\partial_t\hat{M}_n-Z\hat{M}_n)$. Consequently, $\partial_t\hat{M_n}-Z\hat{M_n}$ also solves the space part of Lax pair, but the boundary condition is 0 as $x\to\pm\infty$, w.r.t. $\gamma_{ij}^n$. By the Fredholm theory, for $\lambda\in D_n\setminus\tilde{\mathcal{N}}$, the Fredholm determinant is nonzero, which implies that the integral equation only has trivial zero solution, that is $\partial_t\hat{M}_n-Z\hat{M}_n\equiv0$. Furthermore, it follows from the boundedness of $M_n$ and the behavior of $\lambda\to 0$ that $M_n(x,t;\lambda)$ satisfies the Lax pair for the Tzitz\'eica equation (\ref{Tt}).
	\end{proof}
	\begin{lemma}\label{Lemma-formula-2}
		For $(x,t)\in\R \times [0,T)$, the time-dependent function $M(x,t;\lambda)$ is a sectionally analytic function for $\lambda\in \CC\setminus \Sigma$, and satisfies the same jump conditions with function $M(x;\lambda)$.
	\end{lemma}
	\begin{proof}
		The existence and analyticity of $M_n(x,t;\lambda)$ can be gotten by the standard Fredholm analysis. Moreover, we have proven that $M(x,t;\lambda)$ satisfies the Lax pair, which implies that
		$$
		M_{n+1}(x, t;\lambda)=M_n(x, t;\lambda) \mathrm{e}^{\widehat{\mathcal L(\lambda)} x+\widehat{\mathcal Z(\lambda)} t}\left(M_n(0,0; \lambda)^{-1} M_{n+1}(0,0; \lambda)\right),
		$$
		for $n=1,2,\cdots,6$ and $M_7=M_1$.
	\end{proof}

	\section{\bf Long-time asymptotics for the Tzitz\'eica equation}\label{longtimeasymptotic}
		First, we focus on the critical points of the RH problem \eqref{MRHp} and recall the dispersion relations 
		$\vartheta_{21}$, $\vartheta_{31}$ and $\vartheta_{32}$ appearing in the RH problem:
		\[
		\begin{aligned}
			\vartheta_{21}
			&=(l_2-l_1)x+(z_2-z_1)t  \\
			&=\frac{(\omega^2-\omega)t}{2}\big[(\lambda-\lambda^{-1}) \xi+(\lambda+\lambda^{-1})\big],
		\end{aligned}
		\]
		where $\xi=\frac{x}{t}$. The stationary points of the phase function $\vartheta_{21}$ are
		\begin{align}\label{lambda0def}
			\lambda_0=\pm\sqrt{\frac{|x-t|}{|x+t|}}.
		\end{align}
		Moreover, the stationary points of $\vartheta_{31}$ are given by 
		$\omega\lambda_0=\pm\omega\sqrt{\frac{|x-t|}{|x+t|}}$, while the stationary points of 
		$\vartheta_{32}$ are $\omega^2\lambda_0=\pm\omega^2\sqrt{\frac{|x-t|}{|x+t|}}$.
		
		\par
		
		By convention, we denote
		\[
		B(a,\lambda_0)=B_{\lambda_0}(a):=\{\lambda \in \mathbb{C}\,|\,|\lambda-a|<\lambda_0\},
		\]
		which represents an open disk in the complex plane centered at $a$ with radius $\lambda_0$, and
		\[
		\partial B(a,\lambda_0)=\partial B_{\lambda_0}(a):=\{\lambda\in\mathbb{C}\,|\,|\lambda-a|=\lambda_0\},
		\]
		which is the boundary of $B_{\lambda_0}(a)$. Moreover, the signature of the real part of the phase function 
		$\vartheta_{21}$ is as follows:
		\begin{enumerate}
			\item For $|\xi|<1$, $\Re \vartheta_{21}>0$ when $\lambda \in B_{\lambda_0}(0) \cap \Im \lambda<0$ and $\lambda \in B^c_{\lambda_0}(0) \cap \Im \lambda>0$, while $\Re \vartheta_{21}<0$ elsewhere.
			\item For $\xi\geq 1$, $\Re \vartheta_{21}>0$ when $\Im \lambda>0$, and $\Re \vartheta_{21}<0$ when $\Im \lambda<0$.
			\item For $\xi\leq -1$, $\Re \vartheta_{21}>0$ when $\Im \lambda<0$, and $\Re \vartheta_{21}<0$ when $\Im \lambda>0$.
		\end{enumerate}

		Figure \ref{sign} describes the signature of real part of phase function $\vartheta_{21}$ for $|\xi|<1$.
		
		\begin{figure}[!h]
			\centering
			\begin{overpic}[width=.65\textwidth]{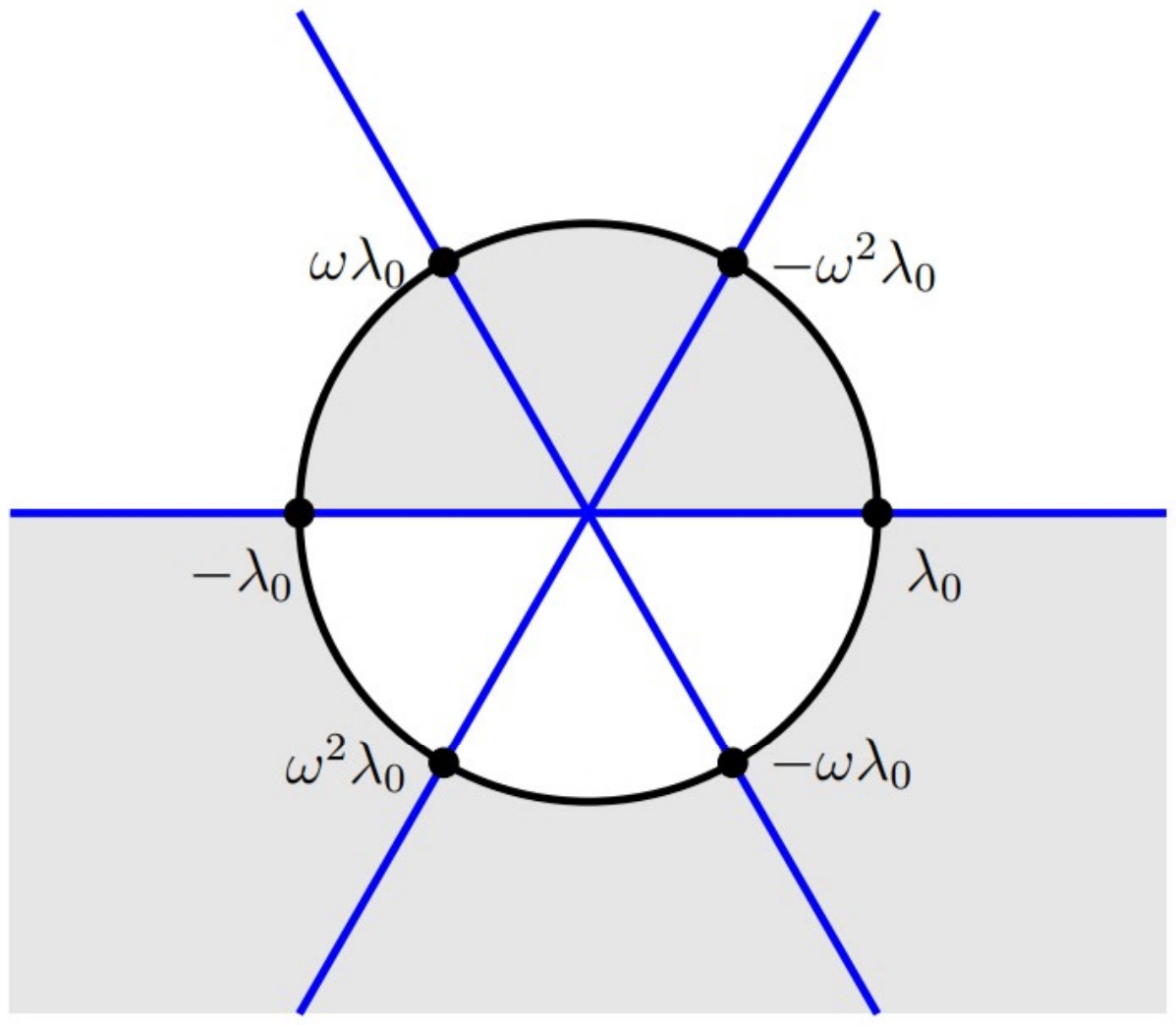}
			\end{overpic}
			\caption{{\protect\small
					The signature of real part of $\vartheta_{21}$ for $|\xi|<1$:
					the shaded region represents $\Re\vartheta_{21}<0$ and the white region represents $\Re\vartheta_{21}>0$.
			}}
			\label{sign}
		\end{figure}

	\subsection{Long-time behavior outside the light cone $|\frac{x}{t}|>1$}
	
	For $\xi>1$, we have that
	
	$$
	\begin{aligned}
		v_1&=\left(\begin{array}{ccc}
			1 & 0 & 0 \\
			r_1^*(\lambda) e^{\vartheta_{21}} & 1 & 0 \\
			0 & 0 & 1
		\end{array}\right)\left(\begin{array}{ccc}
			1 & -r_1(\lambda) e^{-\vartheta_{21}} & 0 \\
			0 & 1 & 0 \\
			0 & 0 & 1
		\end{array}\right)\\&:=v^{lower}_{1}v^{upper}_{1},
	\end{aligned}
	$$
	
	and
	
	$$
	\begin{aligned}
		v_4&=\left(\begin{array}{ccc}
			1 & -r_2^*(\lambda) e^{-\vartheta_{21}} & 0 \\
			0 & 1 & 0 \\
			0 & 0 & 1
		\end{array}\right)\left(\begin{array}{ccc}
			1 & 0 & 0 \\
			r_2(\lambda) e^{\vartheta_{21}} & 1 & 0 \\
			0 & 0 & 1
		\end{array}\right)\\&:=v^{upper}_{4}v^{lower}_{4}.
	\end{aligned}
	$$
	%
	
	Without loss of generality, assume that the reflection coefficient $r_{1}(\lambda)$ is analytic in the region $\{\lambda: 0 < \arg \lambda < \frac{\pi}{3}, \Re \lambda > 0\}$ and continuous to the boundary. Similarly, assume $r_2(\lambda)$ to be analytic in the region $\{\lambda: \pi < \arg \lambda < \frac{4\pi}{3}, \Re \lambda < 0\}$ and continuous to the boundary. Under these assumptions, it follows that the norm $\|r_1(\lambda)e^{-\vartheta_{21}}\|_{\infty}$ decays exponentially as $t \to \infty$ for $\Im \lambda > 0$. Consequently, the jump matrices $v_1$ and $v_4$ approach the identity matrix $I$ in this limit, which implies that the solution to the Riemann-Hilbert problem \ref{MRHp} converges to identity matrix $I$. Combining the recovery formula within this framework, it is deduced that $u(x,t) \to 0$ rapidly as $t \to \infty$. Furthermore, the analyticity and symmetry of the other jump matrices allow to similar factorizations. As $\frac{|x|}{t} \to \infty$ and $t \to \infty$, corresponding to Sector ${\rm I}$ in Theorem \ref{uasy}, the phase functions in the jump matrices are predominantly influenced by $e^{-c|x|}$. This dominance leads to the conclusion that $u(x,t) \to 0 $ in this sector, providing a clear asymptotic behavior as the spatial variable $x$ or time $t$ becomes large.

	\begin{lemma}\label{1xi}
		For $\xi>1$, as $t\to\infty$, the solution of RH problem \ref{MRHp} satisfies
		$
		\|M(x,t;\lambda)-I\|\to0,\
		$
		rapidly, and the solution of Tzitz\'eica equation (\ref{Tt}) rapidly decays to $0$.
	\end{lemma}
	\begin{proof}
		Rewrite
		$
		\vartheta_{21}=\frac{-\sqrt{3}it}{2}[(\lambda-\lambda^{-1}) \xi+(\lambda+\lambda^{-1})]:=2it\theta_{21}(\lambda).
		$
		Since for $\lambda\neq0$, $\theta_{21}$ is a monotone function of $\lambda$, we can view $\lambda$ as the function of $\theta_{21}$, and then
		suppose that $g(\lambda)=r_1(\lambda)(\lambda^2-i)^2$. Thus, it follows
		
		$$
		\begin{aligned}
			r_1(\lambda)e^{-\vartheta_{21}(\lambda)}&=\frac{e^{-2it\theta_{21}}}{(\lambda^2-i)^2}\frac{1}{\sqrt{2\pi}}\int_{\R}\hat g(s)e^{is\theta_{21}}ds\\
			&=\frac{e^{-2it\theta_{21}}}{(\lambda^2-i)^2}\frac{1}{\sqrt{2\pi}}\left( \int_{-\infty}^t\hat g(s)e^{is\theta_{21}}ds+\int^{\infty}_t\hat g(s)e^{is\theta_{21}}ds \right)\\			
			&:=r_{1,a}e^{-\vartheta_{21}(\lambda)}+r_{1,r}e^{-\vartheta_{21}(\lambda)},
		\end{aligned}
		$$
		where we have used the fact that
		$$
		\hat g(s)=\frac{1}{\sqrt{2\pi}}\int_{\R}g(\lambda(\theta))e^{-is\theta_{21}}d\theta.
		$$
		It is clear that $r_{1,a}e^{-\vartheta_{21}(\lambda)}$ has an analytic continuation to the upper half plane and by using the Schwartz reflection principle, the $r_{1,a}^*e^{\vartheta_{21}}$ is analytic in the lower half plane. Furthermore, $r_{1,a}e^{-\vartheta_{21}(\lambda)}$ decays exponentially as $t\to\infty$ in the $L^1\cap L^\infty$ norm. On the other hand, $r_{1,r}e^{-\vartheta_{21}(\lambda)}$ still relies on the positive half real line, {and for any integer $N>0$ $r_{1,r}e^{-\vartheta_{21}(\lambda)}=\mathcal{O}\left(t^{-N}\right)$ } as $t\to\infty$.
		\par		
		In the same way, one can split $r_2(\lambda)e^{\vartheta_{21}}$ into $r_{2,a}(\lambda)e^{\vartheta_{21}}$ and $r_{2,r}(\lambda)e^{\vartheta_{21}}$, where $r_{2,r}(\lambda)e^{\vartheta_{21}}$ relies on the negative half line, rapidly decays as $t\to\infty$ and $r_{2,a}(\lambda)e^{\vartheta_{21}}$ can be analytic continuation to the lower plane with exponentially decaying as $t\to\infty$.
		\par	
		
		As a result, the jump matrices of RH problem \ref{MRHp} can be factorized into
		$$
		V(x,t;\lambda)=\begin{cases}
			\begin{aligned}
				&\left(\begin{matrix}
					1 & -r_{1,r}(\lambda) e^{-\vartheta_{21}} & 0 \\
					r_{1,r}^*(\lambda) e^{\vartheta_{21}} & 1-r_{1,r}(\lambda)r_{1,r}^*(\lambda) & 0 \\
					0 & 0 & 1
				\end{matrix}\right),&&\lambda\in\Sigma_1',\\
				
				&\left(\begin{array}{ccc}
					1 & 0 & 0 \\
					r_{1,a}^*(\lambda) e^{\vartheta_{21}} & 1 & 0 \\
					0 & 0 & 1
				\end{array}\right),&&\lambda\in\Sigma_{1,lower}',\\
				
				&\left(\begin{array}{ccc}
					1 & r_{1,a}(\lambda) e^{-\vartheta_{21}} & 0 \\
					0 & 1 & 0 \\
					0 & 0 & 1
				\end{array}\right),&&\lambda\in\Sigma_{1,upper}',\\
				&\left(\begin{array}{ccc}
					1-r_{2,r}(\lambda)r_{2,r}^*(\lambda) & -r_{2,r}^*(\lambda) e^{-\vartheta_{21}} & 0 \\
					r_{2,r}(\lambda) e^{\vartheta_{21}} & 1 & 0 \\
					0 & 0 & 1
				\end{array}\right),&&\lambda\in\Sigma_4',\\
				
				&\left(\begin{array}{ccc}
					1 & -r_{2,a}^*(\lambda) e^{-\vartheta_{21}} & 0 \\
					0 & 1 & 0 \\
					0 & 0 & 1
				\end{array}\right),&&\lambda\in\Sigma_{4,upper},\\
				
				&\left(\begin{array}{ccc}
					1 & 0 & 0 \\
					-r_{2,a}(\lambda) e^{\vartheta_{21}} & 1 & 0 \\
					0 & 0 & 1
				\end{array}\right),&&\lambda\in\Sigma_{4,lower},\\
			\end{aligned}
		\end{cases}
		$$
		and at the same time the other jump matrices can be gotten by the symmetries.
		\par		
		For function $f(\xi)$, define the Cauchy operator $\mathcal{C}$ on $\Sigma$ as
		$$
		(\mathcal{C}f)(z)=\frac{1}{2\pi i}\int_{\Sigma}\frac{f(\xi)}{\xi-z}d\xi.
		$$
		Denote $\dot L^3(\Sigma)$ as $(1+|\xi|)^{\frac{1}{3}}f(\xi)\in L^3(\Sigma)$, i.e., the weighted $L^3$ space, and Cauchy operator maps the function in $\dot L^3(\Sigma)$ into an analytic function except for $\lambda\in\Sigma$, while the boundary valued functions $\mathcal{C}_{\pm}f(z)$ belong to $\dot L^3(\Sigma)$.  Suppose that $w\in L^{\infty}$, and define $\mathcal{C}_w$ as $\mathcal{C}_w(f):=\mathcal{C}_-(wf)$, it follows that $\mathcal{C}_w$ is a bounded linear operator on $\dot L^3(\Sigma)$. Let $\mu$ satisfy the singular integral equation
		$$
		\mu=I+\mathcal{C}_w\mu,\ \mu \in I+\dot L^3(\Sigma).
		$$
		Consequently, if $(I-\mathcal{C}_w)$ is invertible, there exits a unique solution  for a RH problem with jump contour $\Sigma$ and the jump condition is $I+w$. Furthermore, the solution of this RH problem is
		$$
		M(x,t;\lambda)=I+C(\mu w).
		$$
		Suppose $w(x,t;\lambda)=V(x,t;\lambda)-I$, it is immediate to know that $w(x,t;\lambda)\in\dot L^3(\Sigma')$ and $\|w(x,t;\lambda)\|_{\infty}\to0$ rapidly as $t\to\infty$, which implies that $(I-\mathcal{C}_w)$ is invertible. Reminding that the solution of RH problem \ref{MRHp} exists and is unique, moreover, $M_{\pm}\to I$ rapidly in $\dot L^3(\Sigma)$ norm as $t\to\infty$, thus it is concluded that $M\to I$ and  $u(x,t)\to 0$ as $t\to\infty$ by the reconstruction formula $u(x,t)=\lim_{\lambda\to0}\log \left[ (\omega,\omega^2,1) M(x,t,\lambda) \right]_{13}$.
	\end{proof}

	For $\xi<-1$, the signature of $\vartheta_{21}$ is opposite to the case $\xi>1$ and the factorizations of jump matrices are
	$$
	\begin{aligned}
		v_1
		&=\left(\begin{array}{ccc}
			1 & -\frac{r_1(\lambda) e^{-\vartheta_{21}}}{1-|r_1(\lambda)|^2} & 0 \\
			0 & 1 & 0 \\
			0 & 0 & 1
		\end{array}\right)
		\left(\begin{array}{ccc}
			(1-\left|r_1(\lambda)\right|^2)^{-1} & 0 & 0 \\
			0 & 1-\left|r_1(\lambda)\right|^2 & 0 \\
			0 & 0 & 1
		\end{array}\right)
		\left(\begin{array}{ccc}
			1 & 0 & 0 \\
			\frac{r_1^*(\lambda) e^{\vartheta_{21}}}{1-|r_1(\lambda)|^2} & 1 & 0 \\
			0 & 0 & 1
		\end{array}\right),
	\end{aligned}
	$$
	and
	$$
	\begin{aligned}
		v_4&=\left(\begin{array}{ccc}
			1 & 0 & 0 \\
			\frac{r_2^*(\lambda)}{1-\left|r_2(\lambda)\right|^2} e^{\vartheta_{21}} & 1 & 0 \\
			0 & 0 & 1
		\end{array}\right)\left(\begin{array}{ccc}
			1-\left|r_2(\lambda)\right|^2 & 0 & 0 \\
			0 & {({1-\left|r_2(\lambda)\right|^2})^{-1}} & 0 \\
			0 & 0 & 1
		\end{array}\right)\left(\begin{array}{ccc}
			1 & -\frac{r_2(\lambda)}{1-\left|r_2(\lambda)\right|^2} e^{-\vartheta_{21}} & 0 \\
			0 & 1 & 0 \\
			0 & 0 & 1
		\end{array}\right).
	\end{aligned}
	$$
	In this case, it is necessary to introduce the functions $\delta_1$ and $\delta_4$ of the forms
	$$
	\left\{\begin{aligned}
		\delta_{1+}(\lambda) & =\delta_{1-}(\lambda)\left(1-|r_1(\lambda)|^2\right), & & \lambda\in\R_+, \\
		& =\delta_{1-}(z), & & \lambda\in\CC\setminus\R_+, \\
		\delta_1(\lambda) & \rightarrow 1 & & \text { as } \lambda \rightarrow \infty,
	\end{aligned}\right.
	$$
	and
	$$
	\left\{\begin{aligned}
		\delta_{4+}(\lambda) & =\delta_{4-}(\lambda)\left(1-|r_2(\lambda)|^2\right), & & k\in\R_-, \\
		& =\delta_{4-}(z), & & k\in\CC\setminus\R_-, \\
		\delta_4(\lambda) & \rightarrow 1 & & \text { as } k \rightarrow \infty.
	\end{aligned}\right.
	$$
	Thus, it follows that
	$$
	\delta_1(\lambda)=\exp \left\{\frac{1}{2 \pi i} \int_{0}^{\infty} \frac{\ln \left(1-\left|r_1(s)\right|^2\right)}{s-\lambda} d s\right\}, \quad \lambda \in \mathbb{C} \backslash \R_+ ,
	$$
	and
	$$
	\delta_4(\lambda)=\exp \left\{\frac{1}{2 \pi i} \int^{-\infty}_0 \frac{\ln \left(1-\left|r_2(s)\right|^2\right)}{s-\lambda} d s\right\}, \quad \lambda \in \mathbb{C} \backslash \R_-,
	$$
	where we have denoted $\ln(z)$ as the real-valued function.
	\par
	Using the symmetries of functions $\delta_1$ and $\delta_4$ to define $\delta_j$ for $j=2,3,5,6$ as follows
	$$
	\begin{array}{ll}
		\delta_3( \lambda)=\delta_1\left( \omega^2 \lambda\right), & \lambda \in \mathbb{C} \backslash  \omega\R_+, \\
		\delta_5( \lambda)=\delta_1( \omega \lambda), & \lambda \in \mathbb{C} \backslash \omega\R_-,\\
		\delta_2( \lambda)=\delta_4\left( \omega^2 \lambda\right), & \lambda \in \mathbb{C} \backslash \omega^2\R_-, \\
		\delta_6( \lambda)=\delta_4( \omega \lambda), & \lambda \in \mathbb{C} \backslash \omega^2\R_+.
	\end{array}
	$$
	\par
	Define the matrix-valued function $\Delta$ as
	$$
	\Delta( \lambda)=\left(\begin{array}{ccc}
		\frac{\delta_1( \lambda)\delta_6( \lambda)}{\delta_3( \lambda)\delta_4( \lambda)} & 0 & 0 \\
		0 & \frac{\delta_5( \lambda)\delta_4( \lambda)}{\delta_1( \lambda)\delta_2( \lambda)} & 0 \\
		0 & 0 & \frac{\delta_3( \lambda)\delta_2( \lambda)}{\delta_5( \lambda)\delta_6( \lambda)}
	\end{array}\right)=
	\left(\begin{array}{ccc}
		\Delta_1(\lambda) & 0 & 0 \\
		0 & \Delta_2(\lambda) & 0 \\
		0 & 0 & \Delta_3(\lambda)
	\end{array}\right),
	$$
	and take the transformation
	$$
	\tilde M(x,t;\lambda)=M(x,t;\lambda)\Delta(\lambda).
	$$
	It is clear that
	$$
	\begin{aligned}
		\tilde v_1&=\left(\begin{array}{ccc}
			1-|r_1(\lambda)|^2 & -\frac{\tilde\delta_{v_1}}{\delta_{1-}^2}\frac{r_1(\lambda)}{1-|r_1(\lambda)|^2} e^{-\vartheta_{21}} & 0 \\
			\frac{\delta_{1+}^2 }{\tilde\delta_{v_1}} \frac{r_1^*(\lambda)}{1-|r_1(\lambda)|^2} e^{\vartheta_{21}} & 1 & 0 \\
			0 & 0 & 1
		\end{array}\right)\\
		&=\left(\begin{array}{ccc}
			1 & -\frac{\tilde\delta_{v_1}}{\delta_{1-}^2}\frac{r_1(\lambda)}{1-|r_1(\lambda)|^2} e^{-\vartheta_{21}} & 0 \\
			0 & 1 & 0 \\
			0 & 0 & 1
		\end{array}\right)\left(\begin{array}{ccc}
			1 & 0 & 0 \\
			\frac{\delta_{1+}^2 }{\tilde\delta_{v_1}} \frac{r_1^*(\lambda)}{1-|r_1(\lambda)|^2} e^{\vartheta_{21}} & 1 & 0 \\
			0 & 0 & 1
		\end{array}\right):=\tilde v_1^{lower}\tilde v_1^{upper},
	\end{aligned}
	$$
	and
	$$
	\begin{aligned}
		\tilde v_{4}^{(1)}&=\left(\begin{array}{ccc}
			1 & -\frac{\delta_{4+}^2}{\tilde\delta_{v_4}}\frac{r_{2}^{*}(\lambda)}{1-|r_2(\lambda)|^2} e^{-\vartheta_{21}} & 0 \\
			\frac{\tilde\delta_{v_4}}{\delta_{4-}^2}\frac{r_{2}(\lambda)}{1-|r_2(\lambda)|^2} e^{\vartheta_{21}} & 1-\left|r_{2}(\lambda)\right|^{2} & 0 \\
			0 & 0 & 1
		\end{array}\right)\\
		&=\left(\begin{array}{ccc}
			1 & 0 & 0 \\
			\frac{\tilde\delta_{v_4}}{\delta_{4-}^2}\frac{r_{2}(\lambda)}{1-|r_2(\lambda)|^2} e^{\vartheta_{21}} & 1 & 0 \\
			0 & 0 & 1
		\end{array}\right)\left(\begin{array}{ccc}
			1 & -\frac{\delta_{4+}^2}{\tilde\delta_{v_4}}\frac{r_{2}^{*}(\lambda)}{1-|r_2(\lambda)|^2} e^{-\vartheta_{21}} & 0 \\
			0 & 1 & 0 \\
			0 & 0 & 1
		\end{array}\right):=\tilde v_4^{upper}\tilde v_4^{lower},
	\end{aligned}
	$$
	where $\tilde\delta_{v_1}=\frac{\delta_3\delta^2_4\delta_5}{\delta_6\delta_2}$ and $\tilde\delta_{v_4}=\frac{\delta_1^2\delta_2\delta_6}{\delta_5\delta_3}$.
	%
	
	As in the analysis for $\xi > 1$, the solution of the RH problem for $\tilde{M}$ behaves $\tilde{M} \to I$ rapidly, and note that as $\lambda \to 0$, we have $r_1(\lambda), r_2(\lambda) \to 0$, which implies that the delta functions $\delta_1$ and $\delta_4$ are continuous at $\lambda = 0$. By using the symmetric properties, it follows that $\delta_1(0) = \delta_3(0) = \delta_5(0)$ and $\delta_4(0) = \delta_6(0) = \delta_2(0)$. It is clear that $\Delta \to I$ as $\lambda \to 0$, and furthermore, $u(x, t) \to 0$ as $t \to \infty$ for $\xi < -1$. So the following lemma holds immediately.

	\begin{lemma}\label{-1xi}
		For $\xi<-1$, as $t\to\infty$, the solution of RH problem \ref{MRHp} satisfies
		{
		$$
		\|M(x,t;\lambda)-I\|=\mathcal{O}\left(t^{-N}\right),
		$$}
		and the solution of the Tzitz\'eica equation (\ref{Tt}) decays to zero rapidly.
	\end{lemma}
	
	\subsection{Proof of the Sector I \& II in Theorem \ref{uasy}} \label{SectorI}
	Building upon the  Lemma \ref{1xi} and Lemma \ref{-1xi},  under the condition $\left|\frac{x}{t}\right| > 1$ and as $t \to \infty$, the function $M(x, t, \lambda)$ converges to the identity matrix { $I+\mathcal{O}(t^{-N})$ for any positive integer $N$.} While for $\left|\frac{x}{t}\right|$ tends to $\infty$, the phase function in the jump matrix can be rewritten as $\vartheta_{21}:=\frac{-\mathrm{sgn}(x)\sqrt{3} i |x|(\lambda-\lambda^{-1})}{2}$, while retaining the signature of the phase function. Therefore, the function $M(x, t, \lambda)$ behaves like { $I+\mathcal{O}(|x|^{-N})$.} Moreover, leveraging the explicit expression for the solution to the Tzitz\'eica equation (\ref{Tt}), as articulated in \eqref{usolution}, the result in the Sector I \& II of Theorem \ref{uasy} is proved.

	\subsection{Long-time behavior inside the light cone $|\frac{x}{t}|<1$}
	In this subsection, the long-time asymptotics of Tzitz\'eica equation (\ref{Tt}) inside the light cone $|\frac{x}{t}|<1$ is formulated by deforming RH problem \ref{MRHp} step by step.
	
	\subsubsection{First transformation}
	Recall the signature of $\vartheta_{21}$ for $|\xi|<1$, and if $|\lambda|>|\lambda_0|$, for $\lambda\in\R$ the jump matrix has the same factorization as in the case of $\xi>1$.  On the other hand, if $\lambda$ lies inside of the circle in Figure \ref{sign}, the factorization is the same as in the case of $\xi<-1$. Similarly, introduce the scalar RH problem to cancel the diagonal matrix inside the circle, abusing the same symbol for the convenience. Find functions $\delta_1$ and $\delta_4$ satisfy
	$$
	\left\{\begin{aligned}
		\delta_{1+}(\lambda) & =\delta_{1-}(\lambda)\left(1-|r_1(\lambda)|^2\right), & & 0<\lambda<\lambda_0, \\
		& =\delta_{1-}(z), & & \text{elsewhere}, \\
		\delta_1(\lambda) & \rightarrow 1 & & \text { as } \lambda \rightarrow \infty,
	\end{aligned}\right.
	$$
	and
	$$
	\left\{\begin{aligned}
		\delta_{4+}(\lambda) & =\delta_{4-}(\lambda)\left(1-|r_2(\lambda)|^2\right), & & -\lambda_0<\lambda<0, \\
		& =\delta_{4-}(z), & & \text{elsewhere}, \\
		\delta_4(\lambda) & \rightarrow 1 & & \text { as } \lambda \rightarrow \infty.
	\end{aligned}\right.
	$$
	It is immediate to see that
	$$
	\delta_1(\lambda)=\exp \left\{\frac{1}{2 \pi i} \int_{0}^{\lambda_0} \frac{\ln \left(1-\left|r_1(s)\right|^2\right)}{s-\lambda} d s\right\}, \quad \lambda \in \mathbb{C} \backslash (0,\lambda_0 ),
	$$
	and
	$$
	\delta_4(\lambda)=\exp \left\{\frac{1}{2 \pi i} \int^{-\lambda_0}_0 \frac{\ln \left(1-\left|r_2(s)\right|^2\right)}{s-\lambda} d s\right\}, \quad \lambda \in \mathbb{C} \backslash (-\lambda_0,0) .
	$$
	
	Suppose that $\log_{\theta}(z)$ denotes the logarithm of $z$ with the branch cut along $\arg z=\theta$, in detail, for $\theta=0$, that is
	$$
	\log_0(z)=\ln|z|+arg_0(z),\ arg_0(z)\in(0,2\pi).
	$$
	\begin{proposition}
		The functions $\delta_1(\lambda)$ and $\delta_4(\lambda)$ are piecewise analytic for $\lambda\in(0,\lambda_0)$ and $(-\lambda_0,0)$, respectively, and have the  properties as below:
		\begin{enumerate}
			\item The function $\delta_1(\lambda)$ can be rewritten as
			$$
			\delta_1( \lambda)=e^{i \nu_1 \log_{-\pi}\left(\lambda-\lambda_0\right)} e^{-\chi_1( \lambda)},
			$$
			where
			\begin{align}\label{nu1def}
				\nu_1=-\frac{1}{2 \pi} \ln \left(1-\left|r_1\left(\lambda_0\right)\right|^2\right),
			\end{align}
			and
			$$
			\chi_1( \lambda)=\frac{1}{2 \pi i} \int_{0}^{\lambda_0} \log_{-\pi}(\lambda-s) d \ln \left(1-\left|r_1(s)\right|^2\right) .
			$$
			Besides, the function $\delta_4(\lambda)$ can be rewritten as
			$$
			\delta_4( \lambda)=e^{i \nu_4 \log_{0}\left(\lambda+\lambda_0\right)} e^{-\chi_4( \lambda)},
			$$
			where
			\begin{align}\label{nu4def}
				\nu_4=-\frac{1}{2 \pi} \ln \left(1-\left|r_2\left(-\lambda_0\right)\right|^2\right),
			\end{align}
			and
			$$
			\chi_4( \lambda)=\frac{1}{2 \pi i} \int_{0}^{-\lambda_0} \log_0(\lambda-s) d \ln \left(1-\left|r_2(s)\right|^2\right).
			$$
			
			\item The functions $\delta_1( \lambda)$ and $\delta_4( \lambda)$ are piecewise analytic expect for $\lambda\in(0,\lambda_0)$ and $(-\lambda_0,0)$, with continuous boundary values on $[0,\lambda_0) $ and $(-\lambda_0,0]$, respectively.
			
			\item  The functions $\delta_1( \lambda)$ and $\delta_4( \lambda)$ satisfy the conjugate symmetry and are bounded in $|\lambda|<\lambda_0$, in addition, the properties below hold
			$$
			\delta_1( \lambda)={(\overline{\delta_1( \bar{\lambda})})^{-1}}, \  k \in \mathbb{C} \backslash (0,\lambda_0),\quad
			\delta_4( \lambda)={(\overline{\delta_4( \bar{\lambda})})^{-1}}, \  k \in \mathbb{C} \backslash (-\lambda_0,0),
			$$
			and
			$$
			|\delta_{1\pm}(\lambda)|<\infty,\quad 0<\lambda<\lambda_0,\quad
			|\delta_{4\pm}(\lambda)|<\infty,\quad -\lambda_0<\lambda<0.
			$$
			Moreover, $|\delta_j(\lambda)|$ and $|\delta_j^{-1}(\lambda)|~(j=1,4)$ are bounded.
			
			\item  As $\lambda\to \pm \lambda_0$ along with paths which are not tangential to $\lambda<\lambda_0$ and $\lambda>-\lambda_0$, respectively, we have
			\begin{align*}
				& \left|\chi_1( \lambda)-\chi_1\left( \lambda_0\right)\right| \leq C\left|\lambda-\lambda_0\right|\left(1+|\ln | \lambda-\lambda_0||\right), \\
				& \left|\chi_4( \lambda)-\chi_4\left(-\lambda_0\right)\right| \leq C\left|\lambda+\lambda_0\right|\left(1+|\ln | \lambda+\lambda_0||\right). \\
			\end{align*}
		\end{enumerate}
	\end{proposition}
	\begin{proof}
		The proof of (1), (3) and (4) is standard analysis. For the (2), we claim that both $\delta_1$ and $\delta_4$ are continuous at $\lambda=0$. First of all, notice that $r_j(0)=0$ for $j=1,2$, so that the jump condition is continuous at $\lambda=0$. Then substituting $\lambda=0$ into the $\delta_1(\lambda)$, it yields
		$$
		\delta_1(0)=\exp \left\{\frac{1}{2 \pi i} \int_{0}^{\lambda_0} \frac{\ln \left(1-\left|r_1(s)\right|^2\right)}{s} d s\right\}.
		$$
		By $\frac{\ln(1-|r_1(s)|^2)}{s}=-\frac{|r_1(s)|^2}{s}$ and $|r_1(s)|$ decays rapidly as $s\to0$,  it is concluded that $\delta_1(0)$ is bounded, which implies that $\delta_1$ can be continuous to $\lambda=0$.
	\end{proof}
	Further, introduce the functions $\delta_j(\lambda)$ for $j=2,3,5,6$ as follows
	$$
	\begin{array}{ll}
		\delta_3( \lambda)=\delta_1\left( \omega^2 \lambda\right), & \lambda \in \mathbb{C} \backslash (0,\omega\lambda_0), \\
		\delta_5( \lambda)=\delta_1( \omega \lambda), & \lambda \in \mathbb{C} \backslash (0,\omega^2\lambda_0),\\
		\delta_2( \lambda)=\delta_4\left( \omega \lambda\right), & \lambda \in \mathbb{C} \backslash (-\omega^2\lambda_0,0), \\
		\delta_6( \lambda)=\delta_4( \omega^2 \lambda), & \lambda \in \mathbb{C} \backslash (-\omega\lambda_0,0),
	\end{array}
	$$
	which satisfy the jump conditions
	$$
	\begin{array}{ll}
		\delta_{3+}( \lambda)=\delta_{3-}( \lambda)\left(1-\left|r_1\left(\omega^2 \lambda\right)\right|^2\right), & \lambda \in (0,\omega\lambda_0), \\
		\delta_{5+}( \lambda)=\delta_{5-}( \lambda)\left(1-\left|r_1(\omega \lambda)\right|^2\right), & \lambda \in (0,\omega^2\lambda_0),\\
		\delta_{2+}( \lambda)=\delta_{2-}( \lambda)\left(1-\left|r_2\left(\omega \lambda\right)\right|^2\right), & \lambda \in (-\omega^2\lambda_0,0), \\
		\delta_{6+}( \lambda)=\delta_{6-}( \lambda)\left(1-\left|r_2(\omega^2 \lambda)\right|^2\right), & \lambda \in (-\omega\lambda_0,0) .
	\end{array}
	$$
	\par
	Now, define the matrix-valued function $M^{(1)}(x,t;\lambda)$ to deform that original RH problem \ref{MRHp} for function $M(x,t;\lambda)$ as
	$$
	M^{(1)}(x,t;\lambda):=M(x,t;\lambda)D(\lambda),
	$$
	where the function $D(\lambda)$ is defined by
	$$
	D(\lambda)=\left(\begin{array}{ccc}
		\frac{\delta_1( \lambda)\delta_6( \lambda)}{\delta_3( \lambda)\delta_4( \lambda)} & 0 & 0 \\
		0 & \frac{\delta_5( \lambda)\delta_4( \lambda)}{\delta_1( \lambda)\delta_2( \lambda)} & 0 \\
		0 & 0 & \frac{\delta_3( \lambda)\delta_2( \lambda)}{\delta_5( \lambda)\delta_6( \lambda)}
	\end{array}\right).
	$$
	In fact, since both $\delta_1(\lambda)$ and $\delta_4(\lambda)$ are continuous at $\lambda=0$, it follows that $\delta_1(0)=\delta_3(0)=\delta_5(0)$ and $\delta_2(0)=\delta_4(0)=\delta_6(0)$ by using  the symmetries as above which implies that $D(\lambda)=I+\mathcal{O}(\lambda)$ as $\lambda\to 0$.
	Then the jump matrices for function $M^{(1)}(x,t;\lambda)$ are
	$$
	V^{(1)}(x,t;\lambda)=
	\begin{cases}
		\begin{aligned}
			&\left(\begin{array}{ccc}
				1-|r_1(\lambda)|^2 & -\frac{\tilde\delta_{v_1}}{\delta_{1-}^2}\frac{r_1(\lambda)}{1-|r_1(\lambda)|^2} e^{-\vartheta_{21}} & 0 \\
				\frac{\delta_{1+}^2 }{\tilde\delta_{v_1}} \frac{r_1^*(\lambda)}{1-|r_1(\lambda)|^2} e^{\vartheta_{21}} & 1 & 0 \\
				0 & 0 & 1
			\end{array}\right), && 0<\lambda<\lambda_0,\\
			
			&\left(\begin{array}{ccc}
				1 & -\frac{\tilde\delta_{v1} }{\delta_{1}^2} r_1(\lambda) e^{-\vartheta_{21}} & 0 \\
				\frac{\delta_{1}^2 }{\tilde\delta_{v1}} r_1^*(\lambda) e^{\vartheta_{21}} & 1-r_1(\lambda) r_1^*(\lambda)& 0 \\
				0 & 0 & 1
			\end{array}\right), &&\lambda_0<\lambda, \\
			
			&\left(\begin{array}{ccc}
				1 & -\frac{\delta_{4+}^2}{\tilde\delta_{v_4}}\frac{r_{2}^{*}(\lambda)}{1-|r_2(\lambda)|^2} e^{-\vartheta_{21}} & 0 \\
				\frac{\tilde\delta_{v_4}}{\delta_{4-}^2}\frac{r_{2}(\lambda)}{1-|r_2(\lambda)|^2} e^{\vartheta_{21}} & 1-\left|r_{2}(\lambda)\right|^{2} & 0 \\
				0 & 0 & 1
			\end{array}\right),&& -\lambda_0<\lambda<0,\\
			
			&\left(\begin{array}{ccc}
				1-\left|r_{2}(\lambda)\right|^{2} & -\frac{\delta_{4}^2}{\tilde\delta_{v_4}}{r_{2}^{*}(\lambda)} e^{-\vartheta_{21}} & 0 \\
				\frac{\tilde\delta_{v_4}}{\delta_{4}^2}{r_{2}(\lambda)} e^{\vartheta_{21}} & 1 & 0 \\
				0 & 0 & 1
			\end{array}\right),&&\lambda<-\lambda_0,\\
		\end{aligned}
	\end{cases}
	$$
	with $\tilde\delta_{v_1}=\frac{\delta_3\delta^2_4\delta_5}{\delta_6\delta_2}$ and $\tilde\delta_{v_4}=\frac{\delta_1^2\delta_2\delta_6}{\delta_5\delta_3}$, and other jump matrices can be gotten by the symmetries. Moreover, the jump contour of the RH problem for function $M^{(1)}(x,t;\lambda)$
	is $\Sigma^{(1)}=\Sigma.$

	\subsubsection{Second transformation}
	
	Suppose that
	$$
	\rho_1(\lambda)=\frac{r_1(\lambda)}{1-r_1(\lambda)r_1^{*}(\lambda)}~~{\rm for}~~\ 0<\lambda<\lambda_0, \quad
	\rho_2(\lambda)=\frac{r_2^*(\lambda)}{1-r_2(\lambda)r_2^{*}(\lambda)}~~{\rm for}~~\,  -\lambda_0<\lambda<0.
	$$
	\begin{lemma}\label{lemma-decomposition}
		The functions $r_j(\lambda)$ and $\rho_j(\lambda)$ for $j=1,2$ have the following decompositions
		$$
		\begin{array}{ll}
			r_1(\lambda)=r_{1, a}(x, t, \lambda)+r_{1, r}(x, t, \lambda), & \lambda \in\left[\lambda_0, \infty\right), \\
			r_2^*(\lambda)=r_{2, a}^*(x, t, \lambda)+r_{2, r}^*(x, t, \lambda), & \lambda \in\left( -\infty,-\lambda_0\right], \\
			\rho_1(\lambda)=\rho_{1, a}(x, t, \lambda)+\rho_{1, r}(x, t, \lambda), & \lambda \in\left[0, \lambda_0\right),\\
			\rho_2(\lambda)=\rho_{2, a}(x, t, \lambda)+\rho_{2, r}(x, t, \lambda), & \lambda \in(-\lambda_0,0].
		\end{array}
		$$
		Additionally, the decomposition functions have the properties below:
		\begin{enumerate}
			\item  The functions $r_{1,a}$ and $r_{2,a}^*$ can be analytically continuous to $\lambda_0<\lambda,\Im \lambda>0$ and  $\lambda<-\lambda_0,\Im \lambda>0$, respectively. The functions $\rho_{1,a},\rho_{2,a}$  are able to analytically continuous to $0<\lambda<\lambda_0,\Im \lambda<0$ and  $-\lambda_0<\lambda<0,\Im \lambda<0$, respectively.
			
			\item  Recall that $\vartheta_{21}=\frac{(\omega^2-\omega)t}{2}[(\lambda-\lambda^{-1}) \xi+(\lambda+\lambda^{-1})]$ and denote $\theta_{21}=\frac{-\sqrt{3}it[(\lambda-\lambda^{-1})\xi+(\lambda+\lambda^{-1})]}{4}$, then the functions $r_{1,a},r_{2,a}^*$ and $\rho_{1,a},\rho_{2,a}$ satisfy the following properties:
			\begin{align*}
				&\left|r_{1, a}(x, t, \lambda)-r_1(\lambda_0)\right| \leq C{e^{-\frac{t}{4}\operatorname{Re} \theta_{21}( \lambda)}}|\lambda-\lambda_0|,\quad \Re\lambda>\lambda_0 \ \text{and}\ \Im(\lambda)>0,\\
				&\left|r_{2, a}^*(x, t, \lambda)-r_2^*(-\lambda_0)\right| \leq C{e^{-\frac{t}{4}\operatorname{Re} \theta_{21}( \lambda)}}|\lambda+\lambda_0|,\quad \Re\lambda<-\lambda_0\ \text{and}\ \Im(\lambda)>0.\\
			\end{align*}
			Moreover, it follows
			\begin{align*}
				&\left|\rho_{1, a}(x, t, \lambda)-\rho_1(\lambda_0)\right| \leq {C e^{-\frac{t}{4}\operatorname{Re} \theta_{21}( \lambda)}|\lambda-\lambda_0|}, \quad \lambda\in B_{\lambda_0}(0)\cap \Im\lambda>0\ \text{and}\ \frac{\lambda_0}{2}<\Re\lambda<\lambda_0,\\
				&\small \left|\rho_{2, a}(x, t, \lambda)-\rho_2(-\lambda_0)\right| \leq {C e^{-\frac{t}{4}\operatorname{Re} \theta_{21}( \lambda)}|\lambda+\lambda_0|},\lambda\in B_{\lambda_0}(0)\cap \Im\lambda>0\ \text{and}\ -\lambda_0<\Re\lambda<\frac{-\lambda_0}{2},
			\end{align*}
			and
			\begin{align*}
				&\left|\rho_{1, a}(x, t, \lambda)\right| \leq {C|\lambda|e^{-\frac{t}{4}\operatorname{Re} \theta_{21}( \lambda)}}, \quad \lambda\in B_{\lambda_0}(0)\cap \Im\lambda>0\ \text{and}\ 0<\Re\lambda<\frac{\lambda_0}{2},\\
				&\left|\rho_{2, a}(x, t, \lambda)\right| \leq {C|\lambda|e^{-\frac{t}{4}\operatorname{Re} \theta_{21}( \lambda)}}, \quad \lambda\in B_{\lambda_0}(0)\cap \Im\lambda>0\ \text{and}\ \frac{-\lambda_0}{2}<\Re\lambda<0.
			\end{align*}

			\item  For $1\le p\le\infty$,  the functions  $r_{j,r} $ and $\rho_{j,r}$ satisfy
			\begin{align*}
				&\left\| (1+|\cdot|)r_{1, r}(x, t, \lambda)e^{-2t\theta_{21}}\right\|_{L^p{(\lambda_0,\infty)}} \leq \frac{c}{t^{3/2}}\quad \lambda_0<\lambda,\\
				&\left\| (1+|\cdot|)r_{2, r}^*(x, t, \lambda)e^{-2t\theta_{21}}\right\|_{L^p(-\infty,-\lambda_0)} \leq \frac{c}{t^{3/2}}\quad \lambda<-\lambda_0,
			\end{align*}
			and
			$$
			\left\|(1+{|\cdot|})\rho_{1, r}(x, t, \cdot)e^{2t\theta_{21}}\right\|_{L^p(0,\lambda_0)} \leq \frac{c}{t^{3/2}},\
			\left\|(1+|\cdot|)\rho_{2, r}(x, t, \cdot)e^{2t\theta_{21}}\right\|_{L^p(-\lambda_0,0)} \leq \frac{c}{t^{3/2}}.
			$$
			In particular, { $\rho_{j,r}(x,t;\lambda)=\mathcal{O}(t^{-N})$} for $0<|\lambda|<\frac{\lambda_0}{2}$.
		\end{enumerate}

	\end{lemma}
	
	\begin{proof}
		Notice that
		$$
		\begin{aligned}
			\theta_{21}=\frac{-\sqrt{3}it[(\lambda-\lambda^{-1})\xi+(\lambda+\lambda^{-1})]}{4}
			=\frac{-\sqrt{3}it}{2(1+\lambda_0^2)}(\lambda+\lambda_0^2\lambda^{-1}):=-\sqrt{3}it\theta(\lambda;\lambda_0).
		\end{aligned}
		$$
		The subsequent segment of proof parallels the analytical framework employed in studying the sine-Gordon equation by Cheng, Venakides and Zhou \cite{Zhou-Cpde-singordan}.
	\end{proof}
	
	Open lenses of the contour $\Sigma^{(1)}$ from the critical points, then a new jump contour $\Sigma^{(2)}$ is obtained in Figure \ref{second tranform}. Thus it is ready to take the next deformation.
	Define the matrix-valued function \(M^{(2)}(x,t;\lambda)\) in the vicinity of point \(\lambda_0\) to deform the RH problem for function $M^{(1)}(x,t;\lambda)$ as
	%
	\begin{equation}\label{M12-1}
		M^{(2)}(x,t;\lambda)=\begin{cases}\begin{aligned}
				& M^{(1)}(x,t;\lambda) \left(\begin{array}{ccc}
					1 & 0 & 0 \\
					\frac{\delta_{1}^2 }{\tilde\delta_{v1}}r^*_{1,a} e^{ \vartheta_{21}} & 1 & 0 \\
					0 & 0 & 1
				\end{array}\right) && \text{below }\Sigma_6^{(2)} \text{and  above } \Sigma_4^{(2)},\\
				
				& M^{(1)}(x,t;\lambda) \left(\begin{array}{ccc}
					1 & \frac{\tilde\delta_{v1} }{\delta_{1}^2}r_{1,a} e^{- \vartheta_{21}} & 0 \\
					0 & 1 & 0 \\
					0 & 0 & 1
				\end{array}\right) && \text{below }\Sigma_1^{(2)} \text{and  above } \Sigma_6^{(2)},\\
				
				& M^{(1)}(x,t;\lambda) \left(\begin{array}{ccc}
					1 & -\frac{\tilde\delta_{v_1}}{\delta_{1-}^2}\rho_{1,a} e^-{ \vartheta_{21}} & 0 \\
					0 & 1 & 0 \\
					0 & 0 & 1
				\end{array}\right) && \text{below }\Sigma_5^{(2)} \text{and  above } \Sigma_3^{(2)},\\
				
				& M^{(1)}(x,t;\lambda) \left(\begin{array}{ccc}
					1 & 0  & 0 \\
					-\frac{\delta_{1+}^2 }{\tilde\delta_{v_1}} \rho^*_{1,a}e^{ \vartheta_{21}} & 1 & 0 \\
					0 & 0 & 1
				\end{array}\right) && \text{below }\Sigma_2^{(2)} \text{and  above } \Sigma_5^{(2)},\\
				
				&M^{(1)}(x,t;\lambda) && \text{otherelse}.
			\end{aligned}
		\end{cases}
	\end{equation}
	Similarly, for near $-\lambda_0$, the definition of function \(M^{(2)}(x,t;\lambda)\) is given by
	\begin{equation}\label{M12-2}
		M^{(2)}(x,t;\lambda)=\begin{cases}\begin{aligned}
				& M^{(1)}(x,t;\lambda) \left(\begin{array}{ccc}
					1 & 0 & 0 \\
					-\frac{\tilde\delta_{v_4}}{\delta_{4}^2}r_{2,a} e^{ \vartheta_{21}} & 1 & 0 \\
					0 & 0 & 1
				\end{array}\right) && \text{below }\Sigma_{12}^{(2)} \text{and  above } \Sigma_8^{(2)},\\
				
				& M^{(1)}(x,t;\lambda) \left(\begin{array}{ccc}
					1 & -\frac{\delta_{4}^2}{\tilde\delta_{v_4}}r_{2,a}^* e^{- \vartheta_{21}} & 0 \\
					0 & 1 & 0 \\
					0 & 0 & 1
				\end{array}\right) && \text{below }\Sigma_7^{(2)} \text{and  above } \Sigma_{12}^{(2)},\\
				
				& M^{(1)}(x,t;\lambda) \left(\begin{array}{ccc}
					1 & \frac{\delta_{4+}^2}{\tilde\delta_{v_4}}\rho_{2,a}^{*} e^{- \vartheta_{21}} & 0 \\
					0 & 1 & 0 \\
					0 & 0 & 1
				\end{array}\right) && \text{below }\Sigma_{11}^{(2)} \text{and  above } \Sigma_9^{(2)},\\
				
				& M^{(1)}(x,t;\lambda) \left(\begin{array}{ccc}
					1 & 0  & 0 \\
					\frac{\tilde\delta_{v_4}}{\delta_{4-}^2}\rho_{2,a}e^{ \vartheta_{21}} & 1 & 0 \\
					0 & 0 & 1
				\end{array}\right) && \text{below }\Sigma_{10}^{(2)} \text{and  above } \Sigma_{11}^{(2)},\\
				
				&M^{(1)}(x,t;\lambda) && \text{otherelse}.
			\end{aligned}
		\end{cases}
	\end{equation}
	
	Further factorizations for the situations in the vicinity of \(\pm\omega\lambda_0, \pm\omega^2\lambda_0\) can be similarly defined. However, the details of all transformations are not presented here for simplicity.
	
	\begin{lemma}
		The matrices in transformations (\ref{M12-1}) and (\ref{M12-2}) are uniformly bounded for $\CC\setminus{\Sigma^{(2)}}$ for $t>0$ and as $\lambda\to\infty$ and behave as $I+\mathcal{O}(\frac{1}{\lambda})$. Moreover, as $\lambda\to 0$, they converge to $I$, rapidly.
	\end{lemma}
	
	\begin{proof}
		Since the $\delta_j$ is uniformly bounded for $\CC\setminus\Sigma^{(2)}$ and combining the properties of $r_{j,a}$ and $\rho_{j,a}$ in Lemma \ref{lemma-decomposition}, it follows that the entries of matrices in transformations (\ref{M12-1}) and (\ref{M12-2}) are uniformly bounded and behave $I+\mathcal{O}(\frac{1}{\lambda})$ for $\lambda\to\infty$. Similarly, since $r_j(\lambda)\to0$ rapidly as $\lambda\to0$, it follows that these matrices tend to $I$ rapidly.
	\end{proof}
	\begin{figure}[h]
		\centering
		\begin{tikzpicture}[>=latex]
			\draw[very thick] (-4,0) to (4,0) node[black,right]{$\mathbb{R}$};
			\draw[very thick] (-2,-1.732*2) to (2,1.732*2)  node[black,above]{$\omega^2\mathbb{R}$};
			\draw[very thick] (2,-1.732*2) to (-2,1.732*2)    node[black,above]{$\omega\mathbb{R}$};
			\filldraw[mred] (1.6,0) node[black,below=0.1mm]{$\lambda_{0}$} circle (1.5pt);
			\filldraw[mred] (-1.6,0) node[black,below=0.1mm]{$-\lambda_{0}$} circle (1.5pt);
			\filldraw[mblue] (.8,1.732*0.8) node[black,right=1mm]{$-\omega^{2}\lambda_{0}$} circle (1.5pt);
			\filldraw[mblue] (-.8,-1.732*0.8) node[black,right=1mm]{$\omega^{2}\lambda_{0}$} circle (1.5pt);
			\filldraw[mgreen] (.8,-1.732*0.8) node[black,right=1mm]{$-\omega \lambda_{0}$} circle (1.5pt);
			\filldraw[mgreen] (-.8,1.732*0.8) node[black,right=1mm]{$\omega \lambda_{0}$} circle (1.5pt);
			
			\draw[->,very thick,rotate=60,mblue] (1.6,0)  to (3.2,0.8) ;
			\draw[-,very thick,rotate=60,mblue] (1.6,0)  to (4,1.2);
			\draw[->,very thick,rotate=60,mblue] (1.6,0)  to (3.2,-0.8) ;
			\draw[->,very thick,rotate=60,mblue] (-1.6,0)  to (-3.2,0.8) ;
			\draw[-,very thick,rotate=60,mblue] (-1.6,0)  to (-4,1.2);
			\draw[->,very thick,rotate=60,mblue] (-1.6,0)  to (-3.2,-0.8);
			\draw[-,very thick,rotate=60,mblue] (-1.6,0)  to (-4,-1.2);
			\draw[-,very thick,rotate=60,mblue] (1.6,0)  to (4,-1.2);
			\draw[-,very thick,rotate=60,mblue] (-1.6,0)  to (-1.0,0.3);
			\draw[-,very thick,rotate=60,mblue] (1.6,0)  to (1.0,0.3);
			\draw[-,very thick,rotate=60,mblue] (1.0,-0.3)  to (0,0);
			\draw[->,very thick,rotate=60,mblue] (1.0,-0.3)  to (0.5,-0.15);
			\draw[-,very thick,rotate=60,mblue] (1.0,0.3)  to (0,0);
			\draw[->,very thick,rotate=60,mblue] (1.0,0.3)  to (0.5,0.15);
			\draw[-,very thick,rotate=60,mblue] (0.7,0)  to (0.9,0);
			\draw[-,very thick,rotate=60,mblue] (-1.6,0)  to (-1.2,-0.2) ;
			\draw[-,very thick,rotate=60,mblue] (-1.6,0)  to (-1.0,-0.3);
			\draw[-,very thick,rotate=60,mblue] (-1.0,-0.3)  to (0,0);
			\draw[->,very thick,rotate=60,mblue] (-1.0,-0.3)  to (-0.5,-0.15);
			\draw[-,very thick,rotate=60,mblue] (1.6,0)  to (1.0,-0.3);
			\draw[->,very thick,rotate=60,mblue] (-1.0,0.3)  to (-0.5,0.15);
			\draw[-,very thick,rotate=60,mblue] (-1.0,0.3)  to (0,0);
			
			\draw[->,very thick,rotate=120,mgreen] (1.6,0)  to (3.2,0.8) ;
			\draw[-,very thick,rotate=120,mgreen] (1.6,0)  to (4,1.2);
			\draw[->,very thick,rotate=120,mgreen] (1.6,0)  to (3.2,-0.8) ;
			\draw[->,very thick,rotate=120,mgreen] (-1.6,0)  to (-3.2,0.8) ;
			\draw[-,very thick,rotate=120,mgreen] (-1.6,0)  to (-4,1.2);
			\draw[->,very thick,rotate=120,mgreen] (-1.6,0)  to (-3.2,-0.8);
			\draw[-,very thick,rotate=120,mgreen] (-1.6,0)  to (-4,-1.2);
			\draw[-,very thick,rotate=120,mgreen] (1.6,0)  to (4,-1.2);
			\draw[-,very thick,rotate=120,mgreen] (-1.6,0)  to (-1.0,0.3);
			\draw[-,very thick,rotate=120,mgreen] (1.6,0)  to (1.0,0.3);
			\draw[-,very thick,rotate=120,mgreen] (1.0,-0.3)  to (0,0);
			\draw[->,very thick,rotate=120,mgreen] (1.0,-0.3)  to (0.5,-0.15);
			\draw[-,very thick,rotate=120,mgreen] (1.0,0.3)  to (0,0);
			\draw[->,very thick,rotate=120,mgreen] (1.0,0.3)  to (0.5,0.15);
			\draw[-,very thick,rotate=120,mgreen] (0.7,0)  to (0.9,0);
			\draw[-,very thick,rotate=120,mgreen] (-1.6,0)  to (-1.2,-0.2) ;
			\draw[-,very thick,rotate=120,mgreen] (-1.6,0)  to (-1.0,-0.3);
			\draw[-,very thick,rotate=120,mgreen] (-1.0,-0.3)  to (0,0);
			\draw[->,very thick,rotate=120,mgreen] (-1.0,-0.3)  to (-0.5,-0.15);
			\draw[-,very thick,rotate=120,mgreen] (1.6,0)  to (1.0,-0.3);
			\draw[->,very thick,rotate=120,mgreen] (-1.0,0.3)  to (-0.5,0.15);
			\draw[-,very thick,rotate=120,mgreen] (-1.0,0.3)  to (0,0);
			
			\draw[->,very thick,mred] (1.6,0)  to (3.2,0.8) node[above,black] {$\small 1$};
			\draw[-,very thick,mred] (1.6,0)  to (4,1.2);
			\draw[->,very thick,mred] (1.6,0)  to (3.2,-0.8) node[below,black] {$\small 4$};
			\draw[->,very thick,mred] (-1.6,0)  to (-3.2,0.8) node[above,black] {$\small 7$};
			\draw[-,very thick,mred] (-1.6,0)  to (-4,1.2);
			\draw[->,very thick,mred] (-1.6,0)  to (-3.2,-0.8) node[below,black] {$\small 8$};
			\draw[-,very thick,mred] (-1.6,0)  to (-4,-1.2);
			\draw[-,very thick,mred] (-1.6,0)  to (-1.0,0.3)node[above,black] {\small $10$};
			\draw[->,very thick,mred] (-1.0,0.3)  to (-0.5,0.15);
			\draw[-,very thick,mred] (-1.0,0.3)  to (0,0);
			\draw[-,very thick,mred] (1.6,0)  to (4,-1.2);
			\draw[-,very thick,mred] (1.6,0)  to (1.2,0.2) node[above,black] {$\small 2$};
			\draw[-,very thick,mred] (1.6,0)  to (1.0,0.3);
			\draw[-,very thick,mred] (1.0,0.3)  to (0,0);
			\draw[->,very thick,mred] (1.0,0.3)  to (0.5,0.15);
			\draw[-,very thick,mred] (1.6,0)  to (1.2,-0.2) node[below,black] {\small $3$};
			\draw[-,very thick,mred] (1.6,0)  to (1.0,-0.3);
			\draw[-,very thick,mred] (1.0,-0.3)  to (0,0);
			\draw[->,very thick,mred] (1.0,-0.3)  to (0.5,-0.15);
			\draw[-,very thick] (0.7,0)  to (0.9,0)node[right] {\small $5$};
			\draw[<->,very thick] (-1,0)  to (1,0);
			\draw[<->,very thick] (-3,0) node[below] {\small $12$} to (3,0) node[above] {\small $6$};
			\draw[-,very thick,mred] (-1.6,0)  to (-1.2,-0.2) node[below,black] {\small $9$};
			\draw[-,very thick,mred] (-1.6,0)  to (-1.0,-0.3);
			\draw[-,very thick,mred] (-1.0,-0.3)  to (0,0);
			\draw[->,very thick,mred] (-1.0,-0.3)  to (-0.5,-0.15);
			\draw[-,very thick] (-0.7,0)  to (-1.4,0)node[right] {\small $11$};
			
			\draw[<->,very thick] (-.5,-1.732*.5)  to (.5,1.732*.5);
			
			%
			
			\draw[<->,very thick] (-1.5,-1.732*1.5)  to (1.5,1.732*1.5) ;
			\draw[<->,very thick] (-.5,1.732*.5)  to (.5,-1.732*.5) ;
			\draw[<->,very thick] (-1.5,1.732*1.5)  to (1.5,-1.732*1.5);
			
		\end{tikzpicture}
		\caption{{\protect\small
				The jump contour $\Sigma^{(2)}$ and saddle points $\pm\omega^j\lambda_0$ for $j=0,1,2$.}}
		\label{second tranform}
	\end{figure}
	
	The jump matrices near $\lambda_0$ are defined by
	$$
	\begin{aligned}
		v^{(2)}_1&=\left(\begin{array}{ccc}
			1 & -\frac{\tilde\delta_{v1} }{\delta_{1}^2}r_{1,a} e^{-\vartheta_{21}} & 0 \\
			0 & 1 & 0 \\
			0 & 0 & 1
		\end{array}\right),\
		v^{(2)}_2=\left(\begin{array}{ccc}
			1 & 0 & 0 \\
			-\frac{\delta_{1+}^2 }{\tilde\delta_{v_1}} \rho^*_{1,a} e^{ \vartheta_{21}} & 1 & 0 \\
			0 & 0 & 1
		\end{array}\right),\\	
		v^{(2)}_3&=\left(\begin{array}{ccc}
			1 & \frac{\tilde\delta_{v_1}}{\delta_{1-}^2}\rho_{1,a} e^{-\vartheta_{21}} & 0 \\
			0 & 1 & 0 \\
			0 & 0 & 1
		\end{array}\right),\
		v^{(2)}_4=\left(\begin{array}{ccc}
			1 & 0 & 0 \\
			\frac{\delta_{1}^2 }{\tilde\delta_{v1}}r^*_{1,a} e^{\vartheta_{21}} & 1 & 0 \\
			0 & 0 & 1
		\end{array}\right),\\
		v^{(2)}_5&=\left(\begin{array}{ccc}
			1-\frac{\delta_{+}^2}{\delta_{1-}^2}\rho_{1,r}\rho^*_{1,r}(\lambda) & -\frac{\tilde\delta_{v_1}}{\delta_{1-}^2}\rho_{1,r} e^{-\vartheta_{21}} & 0 \\
			\frac{\delta_{1+}^2 }{\tilde\delta_{v_1}} \rho_{1,r}^*  e^{\vartheta_{21}} & 1 & 0 \\
			0 & 0 & 1
		\end{array}\right),\
		v^{(2)}_6=\left(\begin{array}{ccc}
			1 & -\frac{\tilde\delta_{v1} }{\delta_{1}^2} r_{1,r} e^{-\vartheta_{21}} & 0 \\
			\frac{\delta_{1}^2 }{\tilde\delta_{v1}}r_{1,r}^{*} e^{\vartheta_{21}} & 1-r_{1,r}r^*_{1,r} & 0 \\
			0 & 0 & 1
		\end{array}\right),\\
	\end{aligned}
	$$
	and jump matrices near $-\lambda_0$ are defined by
	$$
	\begin{aligned}
		v^{(2)}_7&=\left(\begin{array}{ccc}
			1 & -\frac{\delta_{4}^2}{\tilde\delta_{v_4}}r_{2,a}^* e^{-\vartheta_{21}} & 0 \\
			0 & 1 & 0 \\
			0 & 0 & 1
		\end{array}\right),\
		v^{(2)}_8=\left(\begin{array}{ccc}
			1 & 0 & 0 \\
			\frac{\tilde\delta_{v_4}}{\delta_{4}^2}r_{2,a} e^{\vartheta_{21}} & 1 & 0 \\
			0 & 0 & 1
		\end{array}\right),\\
		v^{(2)}_9&=\left(\begin{array}{ccc}
			1 & \frac{\delta_{4+}^2}{\tilde\delta_{v_4}}\rho_{2,a}^{*} e^{-\vartheta_{21}} & 0 \\
			0 & 1 & 0 \\
			0 & 0 & 1
		\end{array}\right),\
		v^{(2)}_{10}=\left(\begin{array}{ccc}
			1 & 0 & 0 \\
			-\frac{\tilde\delta_{v_4}}{\delta_{4-}^2}\rho_{2,a} e^{\vartheta_{21}} & 1 & 0 \\
			0 & 0 & 1
		\end{array}\right),\\
		v^{(2)}_{11}&=\left(\begin{array}{ccc}
			1 & -\frac{\delta_{4+}^2}{\tilde\delta_{v_4}}\rho^*_{2,r} e^{-\vartheta_{21}} & 0 \\
			\frac{\tilde\delta_{v_4}}{\delta_{4-}^2}\rho_{2,r} e^{\vartheta_{21}} & 1-\frac{\delta_{4+}^2}{\delta_{4-}^2}\rho_{2,r}\rho^*_{2,r} & 0 \\
			0 & 0 & 1
		\end{array}\right),\
		v^{(2)}_{12}=\left(\begin{array}{ccc}
			1-r_{2,r}r^*_{2,r} & -\frac{\delta_{4}^2}{\tilde\delta_{v_4}}r_{2,r}^{*} e^{-\vartheta_{21}} & 0 \\
			\frac{\tilde\delta_{v_4}}{\delta_{4}^2}r_{2,r} e^{\vartheta_{21}} & 1 & 0 \\
			0 & 0 & 1
		\end{array}\right).\\
	\end{aligned}
	$$
	
	\begin{lemma}
		The jump matrices of RH problem for function $M^{(2)}(x,t;\lambda)$ converge to $I$ as $t\to\infty$ uniformly for $\lambda\in\Sigma^{(2)}$ except for $\lambda$ closed to the critical points $\pm\omega^j \lambda_0$ with $j=0,1,2$. In particular, the jump matrices $v^{(2)}_{5,6,11,12}$ satisfy the following estimate
		\begin{equation}\label{lemma-estimate}
			\|(1+|\cdot|)(v^{(2)}-I)\|_{L^1\cap L^{\infty}(\Sigma^{(2)}_5\cup\Sigma^{(2)}_6\cup \Sigma^{(2)}_{11}\cup\Sigma^{(2)}_{12})}\le Ct^{-\frac{3}{2}}.
		\end{equation}
	\end{lemma}
	
	\begin{proof}
		
		Because the $\delta$ functions are bounded except for the jump contour and $\delta_{\pm}$ are still bounded for $\lambda$ in the boundary, it follows that the components $\frac{\delta_{j,+}}{\delta_{v_j}}$ is uniformly bounded. On the other hand, the estimates for $r_{j,a},r_{j,r}$ and $\rho_{j,a},\rho_{j,r}$ imply the estimate (\ref{lemma-estimate}) of the lemma. See e.g. \cite{Charlier-Lenells-2021} for details.
	\end{proof}

	\subsubsection{Local parametrix at critical points}

	In the analysis of the second transformation, it is observed that the jump matrix approaches the identity matrix \(I\) as \(t \to \infty\), except near the jump contours around critical points \(\pm \omega^j \lambda_0\) for \(j = 0, 1, 2\). In this context, we aim to construct a parametrix in the vicinity of \(\pm \lambda_0\) utilizing parabolic cylinder functions.
	
	Consider a small disc \(B_{\epsilon}(\lambda_0)\) centered at \(\lambda_0\) with radius \(\epsilon\). For \(\lambda_0\) satisfying \(0 < \lambda_0 < M\), define a conformal map by
	\begin{align}\label{z1def}
		z_1=3^{\frac{1}{4}}(\lambda-\lambda_0)\sqrt{\frac{2t}{(1+\lambda_0^2)\lambda_0}}.
	\end{align}
	
	Let us denote the intersection of the jump contour \(\Sigma^{(2)}\) with the discs centered at \(\lambda_0\) and \(-\lambda_0\) by \(\Sigma^{(\lambda_0)} = \Sigma^{(2)} \cap B_{\epsilon}(\lambda_0)\) and \(\Sigma^{(-\lambda_0)} = \Sigma^{(2)} \cap B_{\epsilon}(-\lambda_0)\), respectively. The exponential part of the phase function, \(\vartheta_{21}(\lambda)\), can be reformulated as
	$$
	\begin{aligned}
		\vartheta_{21}(\lambda)
		=\frac{(\omega^2-\omega)t}{2}[(\lambda-\lambda^{-1}) \xi+(\lambda+\lambda^{-1})]
		:=\vartheta_{21}(\lambda_0)-\frac{iz_1^2}{2}\left(1-\frac{\lambda_0^4z_1}{\sqrt{2\tau}\eta_1^4}\right),
	\end{aligned}
	$$
	where \(\eta_1 = \lambda_0 + 3^{-\frac{1}{4}}k \sqrt{\frac{(1 + \lambda_0^2) \lambda_0}{2t}}\) for \(0 < k < 1\) and \(\tau = \frac{i \vartheta_{21}(\lambda_0)}{2} = \frac{\sqrt{3} t \lambda_0}{1 + \lambda_0^2}\).
	
	Recall the definition of \(\delta_1(\lambda)\) for \(\lambda \in \mathbb{C} \setminus (0, \lambda_0]\) of the form
	$$
	\delta_1( \lambda)=e^{i \nu_1 \log_{-\pi}\left(\lambda-\lambda_0\right)} e^{-\chi_1( \lambda)},\quad \lambda\in\CC\setminus(0,\lambda_0],
	$$
	where \(\nu_1\) is given by
	$$
	\nu_1=-\frac{1}{2 \pi} \ln \left(1-\left|r_1\left(\lambda_0\right)\right|^2\right),
	$$
	and \(\chi_1( \lambda)\) is defined as
	$$
	\chi_1( \lambda)=\frac{1}{2 \pi i} \int_{0}^{\lambda_0} \log_{-\pi}(\lambda-s) d \ln \left(1-\left|r_1(s)\right|^2\right) .
	$$
	The ratio of the square of \(\delta_{1+}(\lambda)\) over \(\delta_{\tilde v_1}(\lambda)\) is
	%
	%
	$$
	\begin{aligned}
		\frac{\delta_{1+}^{2}(\lambda)}{\delta_{\tilde v_1}(\lambda)}
		=e^{2i \nu_1 \log_{-\pi}\left(z_1\right)}\frac{a^{2i\nu_1}e^{-2\chi_1(\lambda_0)}}{\tilde\delta_{ v_1}(\lambda_0)}
		\frac{e^{2\chi_1(\lambda_0)-2\chi_1(\lambda)}\tilde\delta_{v_1}(\lambda_0)}{\tilde\delta_{ v_1}(\lambda)}
		:=e^{2i \nu_1 \log_{-\pi}\left(z_1\right)}\delta_{\lambda_0}^0\delta_{\lambda_0}^1
	\end{aligned}
	$$
	with $a=3^{-\frac{1}{4}}\sqrt{\frac{(1+\lambda_0^2)\lambda_0}{2t}}$, where $\delta_{\lambda_0}^0=\frac{a^{2i\nu_1}e^{-2\chi_1(\lambda_0)}}{\delta_{\tilde v_1}(\lambda_0)}$ and $\delta_{\lambda_0}^1=\frac{e^{2\chi_1(\lambda_0)-2\chi_1(\lambda)}\delta_{\tilde v_1}(\lambda_0)}{\delta_{\tilde v_1}(\lambda)}$.
	\par
	Similarly, for the small disc centered at \(\lambda = -\lambda_0\), which is denoted as \(B_{\epsilon}(-\lambda_0)\), define the conformal map on \(B_{\epsilon}(-\lambda_0)\) as
	$$
	z_2=3^{\frac{1}{4}}(\lambda+\lambda_0)\sqrt{\frac{2t}{(1+\lambda_0^2)\lambda_0}}.
	$$
	For \(\lambda \in \Sigma^{(-\lambda_0)}\), the exponential part, \(\vartheta_{21}(\lambda)\), is similarly transformed into
	$$
	\begin{aligned}
		\vartheta_{21}(\lambda)
		=\vartheta_{21}(-\lambda_0)+\frac{iz_2^2}{2}\left(1+\frac{\lambda_0^4}{\sqrt{2\tau}\eta^4_2}z_2\right),
	\end{aligned}
	$$
	with $\eta_2=-\lambda_0+3^{-\frac{1}{4}}k\sqrt{\frac{(1+\lambda_0^2)\lambda_0}{2t}}$, $0<k<1$.
	\noindent Furthermore, the \(\delta\) functions on \(\Sigma^{(-\lambda_0)}\) involve \(\delta_4\), defined for \(\lambda \in \mathbb{C} \setminus [-\lambda_0, 0)\) as:
	$$
	\delta_4( \lambda)=e^{i \nu_4 \log_{0}\left(\lambda+\lambda_0\right)} e^{-\chi_4( \lambda)},\quad \lambda\in\CC\setminus[-\lambda_0,0),
	$$
	where \(\nu_4 = -\frac{1}{2 \pi} \ln \left(1 - |r_2(-\lambda_0)|^2\right)\), and
	$$
	\chi_4( \lambda)=\frac{1}{2 \pi i} \int_{0}^{-\lambda_0} \log_{0}(\lambda-s) d \ln \left(1-\left|r_2(s)\right|^2\right) .
	$$
	Moreover, the \({\tilde\delta_{v_4}}/{\delta_{4}^2}\) is expressed as
	$$
	\begin{aligned}
		\frac{\tilde\delta_{v_4}}{\delta_{4}^2}=e^{-2i \nu_4 \log_{\pi}\left(z_2\right)} \frac{\tilde\delta_{v_4}(-\lambda_0)}{a^{2i\nu_4}e^{-2\chi_4( -\lambda_0)}}\frac{\tilde\delta_{ v_4}(\lambda)}{e^{2\chi_4(-\lambda_0)-2\chi_4(\lambda)}\tilde\delta_{ v_4}(-\lambda_0)}
		:=e^{-2i \nu_4 \log_{\pi}\left(z_2\right)}\left(\delta_{-\lambda_0}^0\right)^{-1}\left(\delta_{-\lambda_0}^1\right)^{-1}
	\end{aligned}
	$$
	where $\delta_{-\lambda_0}^0=\frac{a^{2i\nu_4}e^{-2\chi_4(-\lambda_0)}}{\tilde\delta_{ v_4}(-\lambda_0)}$ and $\delta_{-\lambda_0}^1=\frac{e^{2\chi_4(-\lambda_0)-2\chi_4(\lambda)}\delta_{ v_4}(-\lambda_0)}{\tilde\delta_{  v_4}(\lambda)}$.
	\par
	Introduce the diagonal matrix \(H(\pm\lambda_0, t)\) as
	$$
	H(\pm\lambda_0,t)=\diag\left(\left(\delta_{{\pm \lambda_0}}^{0}\right)^{\mp\frac{1}{2}} e^{-\frac{1}{2}\vartheta_{21}(\pm\lambda_0)},\left(\delta_{{\pm\lambda_0}}^{0}\right)^{\pm\frac{1}{2}} e^{\frac{1}{2} \vartheta_{21}(\pm\lambda_0)},1\right),
	$$
	and then further take the deformation as
	$$
	M^{(3,\epsilon)}=M^{(2)}(x,t;\lambda)H(\pm \lambda_0,t),\quad \lambda\in B_{\epsilon}(\pm \lambda_0),
	$$
	which leads to the definition of the jump matrices on \(\Sigma^{(\lambda_0)}\) as follows
	$$
	\begin{aligned}
		v^{(3,\epsilon)}_1&=\left(\begin{array}{ccc}
			1 & -e^{-2i \nu_1 \log_{-\pi}\left(z_1\right)}(\delta_{\lambda_0}^1)^{-1} r_{1,a}(\lambda(z_1)) e^{\frac{iz_1^2}{2}\left(1-\frac{\lambda_0^4z_1}{\sqrt{2\tau}\eta_1^4}\right)} & 0 \\
			0 & 1 & 0 \\
			0 & 0 & 1
		\end{array}\right),\\
		v^{(3,\epsilon)}_2&=\left(\begin{array}{ccc}
			1 & 0 & 0 \\
			-e^{2i \nu_1 \log_{-\pi}\left(z_1\right)}\delta_{\lambda_0}^1 \rho^*_{1,a}(\lambda(z_1)) e^{-\frac{iz_1^2}{2}\left(1-\frac{\lambda_0^4z_1}{\sqrt{2\tau}\eta_1^4}\right)} & 1 & 0 \\
			0 & 0 & 1
		\end{array}\right),\\
		v^{(3,\epsilon)}_3&=\left(\begin{array}{ccc}
			1 & e^{-2i \nu_1 \log_{-\pi}\left(z_1\right)}(\delta_{\lambda_0}^1)^{-1}\rho_{1,a}(\lambda(z_1)) e^{\frac{iz_1^2}{2}\left(1-\frac{\lambda_0^4z_1}{\sqrt{2\tau}\eta_1^4}\right)} & 0 \\
			0 & 1 & 0 \\
			0 & 0 & 1
		\end{array}\right),\\
		v^{(3,\epsilon)}_4&=\left(\begin{array}{ccc}
			1 & 0 & 0 \\
			e^{2i \nu_1 \log_{-\pi}\left(z_1\right)}\delta_{\lambda_0}^1r^*_{1,a}(\lambda(z_1)) e^{-\frac{iz_1^2}{2}\left(1-\frac{\lambda_0^4z_1}{\sqrt{2\tau}\eta_1^4}\right)} & 1 & 0 \\
			0 & 0 & 1
		\end{array}\right).\\
	\end{aligned}
	$$
	Similarly, the jump matrices \(v^{(3,\epsilon)}_7\) to \(v^{(3,\epsilon)}_{10}\) for \(\Sigma^{(-\lambda_0)}\) are defined to account for modifications due to the \(\delta_4\) function and its corresponding transformations. The expressions of \(v^{(3,\epsilon)}_j\) for $j=7,8,9,10$ are given by
	$$
	\begin{aligned}
		v^{(3,\epsilon)}_7&=\left(\begin{array}{ccc}
			1 & -e^{2i \nu_4 \log_{\pi}\left(z\right)}\delta_{-\lambda_0}^1r_{2,a}^* e^{-\frac{iz_2^2}{2}\left(1+\frac{\lambda_0^4}{\sqrt{2\tau}\eta^4_2}z_2\right)} & 0 \\
			0 & 1 & 0 \\
			0 & 0 & 1
		\end{array}\right),\\
		v^{(3,\epsilon)}_8&=\left(\begin{array}{ccc}
			1 & 0 & 0 \\
			e^{-2i \nu_4 \log_{\pi}\left(z\right)}\left(\delta_{-\lambda_0}^1\right)^{-1} r_{2,a} e^{\frac{iz_2^2}{2}\left(1+\frac{\lambda_0^4}{\sqrt{2\tau}\eta^4_2}z_2\right)} & 1 & 0 \\
			0 & 0 & 1
		\end{array}\right),\\
		v^{(3,\epsilon)}_9&=\left(\begin{array}{ccc}
			1 & e^{2i \nu_4 \log_{\pi}\left(z\right)}\delta_{-\lambda_0}^1\rho_{2,a}^{*} e^{-\frac{iz_2^2}{2}\left(1+\frac{\lambda_0^4}{\sqrt{2\tau}\eta^4_2}z_2\right)} & 0 \\
			0 & 1 & 0 \\
			0 & 0 & 1
		\end{array}\right),\\
		v^{(3,\epsilon)}_{10}&=\left(\begin{array}{ccc}
			1 & 0 & 0 \\
			-e^{-2i \nu_4 \log_{\pi}\left(z\right)}\left(\delta_{-\lambda_0}^1\right)^{-1}\rho_{2,a} e^{\frac{iz_2^2}{2}\left(1+\frac{\lambda_0^4}{\sqrt{2\tau}\eta^4_2}z_2\right)} & 1 & 0 \\
			0 & 0 & 1
		\end{array}\right).\\
	\end{aligned}
	$$
	\noindent
	When $z$ is fixed, several limits are observed as $t \to \infty$ for $j = 1, 2$. Specifically, when $t \to \infty$ it follows immediately that $r_{j,a} \to r_j(\lambda_0)$, $\rho_{j,a} \to \frac{r_j(\lambda_0)}{1-|r_j(\lambda_0)|^2}$, and both $\delta_{\lambda_0}^1$ and $\delta_{- \lambda_0}^1$ approach $1$, while $\frac{\lambda_0^4 z_j}{\sqrt{2\tau} \eta_j^4} \to 0$. These limits imply that $v^{3,\epsilon} \to v^X_{\pm\lambda_0}$ as $t \to \infty$, where $v^X_{\pm\lambda_0}$ denotes the jump matrices of the model problem associated with $M^X_{ \pm\lambda_0}$.
	\par	
	To further elaborate, define $M^{(\pm \lambda_0)}(x, t, \lambda)$ within $B_{\epsilon}(\pm \lambda_0)$ as follows
	\begin{align}\label{MxMlambda0}
		M^{(\pm \lambda_0)}(x, t, \lambda) = H(\pm \lambda_0) M^X_{\pm\lambda_0}(x, z(\lambda)) H(\pm \lambda_0 )^{-1},
	\end{align}
	where the prefactor $H(\pm\lambda_0)$ ensures that the jump matrix on $\partial B_{\epsilon}(\pm \lambda_0)$ converges to the identity matrix $I$ as $t \to \infty $.


	\begin{lemma}
		The matrix function $H(\pm \lambda_0,t)$ is uniformly bounded, i.e.,
		$$
		\sup_{t\ge t_0}|H(\pm \lambda_0,t)|\le C, \quad 0<\lambda_0<M.
		$$
		The function $\delta_{\pm \lambda_0}^0,\delta_{\pm \lambda_0}^1$ and $\frac{\lambda_0^4z_j}{\sqrt{2\tau}\eta_j^4}$ satisfy the following properties:
		$$
		|\delta_{\lambda_0}^0|=1,\quad |\delta_{-\lambda_0}^0|=e^{-2\pi\nu_4}, \quad 0<\lambda_0<1\ \text{and}\ t \ge t_0,\\
		$$
		moreover, one has
		\begin{align}
			|\delta_{\pm \lambda_0}^1(\lambda)-1|&\le C|\lambda\mp \lambda_0|(1+|\ln(|\lambda\mp\lambda_0|)|).
		\end{align}
	\end{lemma}
	\begin{proof}
		Recall that the expression for $\delta_{\lambda_0}^0$ is given by $\delta_{\lambda_0}^0 = \frac{a^{-2i\nu_1}e^{-2\chi_1(\lambda_0)}}{\tilde \delta_{v_1}(\lambda_0)}$. Direct calculation shows that
		$$
		|a^{-2i\nu_1}|=\left|\left(3^{-\frac{1}{4}}\sqrt{\frac{(1+\lambda_0^2)\lambda_0}{2t}}\right)^{2i\nu_1}\right|=|e^{2i\nu_1 ln(a)}|=1.
		$$
		Furthermore, due to the relationship $\delta_{1,4}(\lambda) = {(\overline{\delta_{1,4}(\bar{k})})^{-1}}$ and the symmetry between $\delta_{1}$ (respectively, $\delta_{4}$) and $\delta_{3,5}$ (respectively, $\delta_{2,6}$), we have
		$
		|\tilde \delta_{v_1}(\lambda_0)|=1.
		$
		Additionally, the real part of $\chi_1$ at $\lambda_0$ is
		$$
		\operatorname{Re}\chi_1(\lambda_0)=\frac{1}{2 \pi} \int^{\lambda_0}_{0} 0 d \ln \left(1-\left|r_1(s)\right|^2\right)=0,
		$$
		implying that
		$$
		|\delta_{\lambda_0}^0|=\left|\frac{a^{2i\nu}e^{-2\chi_1(\lambda_0)}}{\delta_{\tilde v_1}(\lambda_0)}\right|=1.
		$$
		In a similar way, it is obvious that
		$$
		\operatorname{Re}\chi_4(-\lambda_0)=\frac{1}{2 \pi} \int^{0}_{-\lambda_0} \pi d \ln \left(1-\left|r_2(s)\right|^2\right)=-\pi\nu_4,
		$$
		and consequently
		$$
		|\delta_{-\lambda_0}^0|=\left|\frac{a^{2i\nu_4}e^{-2\chi_4(\lambda_0)}}{\tilde\delta_{ v_4}(\lambda_0)}\right|=e^{2\pi \nu_4}.
		$$
		\noindent Recalling the definition of $\delta_{\lambda_0}^1$, i.e., $\delta_{\lambda_0}^1 = \frac{e^{2\chi_1(\lambda_0)-2\chi_1(\lambda)}\tilde\delta_{ v_1}(\lambda_0)}{\tilde\delta_{ v_1}(\lambda)}$, it becomes evident that
		$$
		|e^{2\chi_1(\lambda_0)-2\chi_1(\lambda)}-1|\le C|\chi_1(\lambda_0)-\chi_1(\lambda)|\le C|\lambda-\lambda_0|(1+|\ln(|\lambda-\lambda_0|)|).
		$$
	\end{proof}
	
	\begin{lemma}\label{V2lambda0} Denote the jump matrices as $V^{(\pm \lambda_0)}$ on $\Sigma^{(\pm \lambda_0)}$, respectively. The matrix functions $M^{(\pm \lambda_0)}$ are analytic for $\lambda\in B_{\epsilon}(\pm \lambda_0)\setminus\Sigma^{(\pm \lambda_0)}$.  For $0<\lambda_0<M$ and $t$ large enough, it follows
		$$
		\|V^{(2)}-V^{(\pm \lambda_0)}\|_{L^1(\Sigma^{(\pm \lambda_0)})}\le C\frac{\ln t}{t},
		$$
		and
		$$
		\|V^{(2)}-V^{(\pm \lambda_0)}\|_{L^{\infty}(\Sigma^{(\pm\lambda_0)})}\le C\frac{\ln t}{t^{\frac{1}{2}}}.
		$$
		Furthermore, we have
		\begin{align}
			&\left\|M^{(\pm \lambda_0)}(x, t, \zeta)^{-1}-I\right\|_{L^{\infty}(\partial B_{\epsilon}\left(\pm \lambda_0\right))}  =\mathcal{O}\left(t^{-1 / 2}\right), \label{Mlambda0esti} \\
			\label{Masy}& \int_{\partial B_{\epsilon}\left(\pm \lambda_0\right)}\left(M^{(\pm \lambda_0)}(x, t, \zeta)^{-1}-I\right)\zeta^{-1}  \frac{d \zeta}{2 \pi i} =\mp\frac{H(\pm \lambda_0, t) \left(M^X_{\pm\lambda_0}(y)\right)_{1} H(\pm \lambda_0, t)^{-1}}{3^{\frac{1}{4}}\sqrt{\frac{2t\lambda_0}{(1+\lambda_0^2)}}}+\mathcal{O}\left(t^{-1}\right),
		\end{align}
		where $\left(M^X_{\pm\lambda_0}(y)\right)_{1} $ is defined in \eqref{Mlambda0}.
	\end{lemma}

	\begin{proof}
		Notice that
		$$
		V^{(2)}-V^{(\lambda_0)}=H(\lambda_0,t)\left(V^{(3,\epsilon)}-V^{X}_{\lambda_0}\right)H(\lambda_0,t),
		$$
		and since $H(\lambda_0,t)^{\pm1}$ is uniformly bounded, it suffices to show that
		$$
		\begin{aligned}
			& \left\|V^{(3,\epsilon)}(x, t; \cdot)-V^{X}_{\lambda_0}(x, t, z(\lambda_0, \cdot))\right\|_{L^1\left(\Sigma^{(\lambda_0)}\right)} \leq C t^{-1} \ln t, \\
			& \left\|V^{(3,\epsilon)}(x, t; \cdot)-V^{X}_{\lambda_0}(x, t, z(\lambda_0, \cdot))\right\|_{L^{\infty}\left(\Sigma^{(\lambda_0)}\right)} \leq C t^{-1 / 2} \ln t.
		\end{aligned}
		$$
		Recall that
		
		\begin{align*}
			&v^{X}_1(z_1,y;\lambda_0)=\left(\begin{array}{ccc}
				1 & -y z_1^{-2 i \nu_1(y)} e^{\frac{i z_1^2}{2}} & 0 \\
				0 & 1 & 0 \\
				0 & 0 & 1
			\end{array}\right),\\
			&v^{(3,\epsilon)}_1=\left(\begin{array}{ccc}
				1 & -e^{-2i \nu_1 \log_{-\pi}\left(z_1\right)}(\delta_{\lambda_0}^1)^{-1} r_{1,a}(\lambda(z_1)) e^{\frac{iz_1^2}{2}\left(1-\frac{\lambda_0^4z_1}{\sqrt{2\tau}\eta_1^4}\right)} & 0 \\
				0 & 1 & 0 \\
				0 & 0 & 1
			\end{array}\right),
		\end{align*}
		then it follows
		$$
		\begin{aligned}
			&\left|e^{-2i \nu_1 \log_{-\pi} \left(z_1\right)}(\delta_{\lambda_0}^1)^{-1} r_{1,a}(\lambda(z_1)) e^{\frac{iz_1^2}{2} \left(1-\frac{\lambda_0^4z_1}{\sqrt{2\tau}\eta_1^4}\right)}-y z_1^{-2 i \nu_1(y)} e^{\frac{i z_1^2}{2}}\right|\\
			&\le C\left|((\delta_{\lambda_0}^1)^{-1}-1)r_{1,a} e^{\frac{iz_1^2}{2} \left(1-\frac{\lambda_0^4z_1}{\sqrt{2\tau}\eta_1^4}\right)}+(e^{\frac{iz_1^2}{2} \left(1-\frac{\lambda_0^4z_1}{\sqrt{2\tau}\eta_1^4}\right)}-e^{\frac{iz_1^2}{2}}) r_{1,a}+(r_{1,a}(\lambda)-r_{1}(\lambda_0))e^{\frac{iz_1^2}{2}}\right|\\
			&\le C\left(|\lambda-\lambda_0|\ln(|\lambda-\lambda_0|)+|\lambda-\lambda_0| \right)e^{-\frac{t}{2}\operatorname{Re}\Phi_{21}}|e^{\frac{iz^2}{2}}|\\
			&\le C |\lambda-\lambda_0|(1+\ln(|\lambda-\lambda_0|))e^{-ct|\lambda-\lambda_0|^2},
		\end{aligned}
		$$
		which implies that
		$$
		\left\|\left(v^{(3,\epsilon)}-v^{X_A}(\lambda_0)\right)_{12}\right\|_{L^1\left(\Sigma_{\lambda_0}\right)} \leq C \int_0^{\infty} s(1+|\ln s|) e^{-c t u^2} d s \leq C t^{-1} \ln t,
		$$
		and
		$$
		\left\|\left(v^{(3,\epsilon)}-v^{X_A}\right)_{12}\right\|_{L^{\infty}\left(\Sigma_{\lambda_0}\right)} \leq C \sup _{s \geq 0} s(1+|\ln s|) e^{-c t s^2} \leq C t^{-1 / 2} \ln t.
		$$
		\par
		Observing the expression for $z_1$ defined in \eqref{z1def}, it is evident that for $\lambda \in \partial B_{\epsilon}(\lambda_0)$, the value of $z_1$ tends towards infinity as $t \rightarrow \infty$. Combining this behavior with the WKB expansion for $M^{X}_{\lambda_0}$ shown in \eqref{-Mlmabda0}, yields
		\[M^{X}_{\lambda_0}(y,z)=I+\frac{(M^{X}_{\lambda_0}(y))_1}{3^{\frac{1}{4}}\sqrt{\frac{2t}{(1+\lambda_0^2)\lambda_0}}(\lambda-\lambda_0)}+\mathcal{O}\left(\frac{1}{t}\right),\]
		and considering the equation \eqref{MxMlambda0}, it leads to the derivation that
		\begin{align}\label{Mlambda0miu}
			M^{\lambda_0}(y,z)^{-1}-I=\frac{ H( \lambda_0)(M^{X}_{\lambda_0}(y))_1  H( \lambda_0)^{-1}}{3^{\frac{1}{4}}\sqrt{\frac{2t}{(1+\lambda_0^2)\lambda_0}}(\lambda-\lambda_0)}+\mathcal{O}\left(\frac{1}{t}\right), \quad \text{as } t\to \infty, \quad |y| \leq C,
		\end{align}
		thereby the validation of \eqref{Mlambda0esti} in the scenario of $\lambda_0$ completes. A similar approach can be employed to prove the case for $-\lambda_0$. Ultimately, the establishment of \eqref{Masy} is directly obtained through \eqref{Mlambda0miu} and Cauchy integral formula.
		
	\end{proof}

	\subsubsection{Small norm RH problem}
	
	By using the symmetries of the RH problem, one can define the local RH problem near other critical points $\pm\omega^j \lambda_0$ by the way
	$$
	\tilde M^{(\pm\omega \lambda_0)}(x,t;\lambda)=\mathcal{A}M^{(\pm \lambda_0)}(x,t;\omega \lambda)\mathcal{A}^{-1}.
	$$
	Denote $\tilde B^{(\pm \lambda_0)}_{\epsilon}=\cup_{j=0}^{2} B_{\epsilon}(\pm \omega^j\lambda_0) $ and define the matrix-valued function  $\tilde M(x,t;\lambda)$ as
	$$
	\tilde M(x,t;\lambda):=\begin{cases}\begin{aligned}
			&M^{(2)}\left(\tilde M^{(\pm\lambda_0)}\right)^{-1}, &&\lambda\in \tilde B^{(\pm\lambda_0)}_{\epsilon},\\
			&M^{(2)}, && \text{otherelse}.
		\end{aligned}
	\end{cases}
	$$
	\par
	In conclusion, denote the jump  contour $\tilde{\Sigma}:=\Sigma^{(2)}\cup\partial \tilde B^{(\pm\lambda_0)}_{\epsilon}$ and then the jump matrix is defined by
	$$
	\tilde{V}:=\begin{cases}\begin{aligned}
			&V^{(2)}, && \lambda \in \tilde{\Sigma}\setminus \overline {\left(\tilde B^{(\pm \lambda_0)}_{\epsilon}\right)},\\
			&(\tilde M^{(\pm\lambda_0)})^{-1}, && \lambda \in \partial \tilde B^{(\pm\lambda_0)}_{\epsilon},\\
			
			&\tilde M^{(\pm\lambda_0)}_{-}V^{(2)}(\tilde M^{(\pm\lambda_0)}_+)^{-1}, && \lambda \in \tilde B^{(\pm\lambda_0)}_{\epsilon}\cap {\Sigma}^{(2)}.
			
	\end{aligned}\end{cases}
	$$
	\par
	Let $\tilde\Sigma^{(\pm \lambda_0)}:=\cup_{j=0}^{2}\Sigma^{(\pm\omega^j\lambda_0)}$, then the jump contour of the new RH problem near $\pm\lambda_0$ is depicted in Figure \ref{new RH}.
	The estimates for the jump condition are listed in the following lemma.

    \begin{figure}[H]
		\centering
		\begin{tikzpicture}[>=latex]
			\draw[very thick] (-4,0) to (4,0) node[black,right]{$\mathbb{R}$};
			\draw[very thick] (-2,-1.732*2) to (2,1.732*2)  node[black,above]{$\omega^2\mathbb{R}$};
			\draw[very thick] (2,-1.732*2) to (-2,1.732*2)    node[black,above]{$\omega\mathbb{R}$};
			\filldraw[mred] (1.6,0) node[black,below=0.1mm]{$\lambda_{0}$} circle (1.5pt);
			\filldraw[mred] (-1.6,0) node[black,below=0.1mm]{$-\lambda_{0}$} circle (1.5pt);
			\filldraw[mblue] (.8,1.732*0.8) node[black,right=1mm]{$-\omega^{2}\lambda_{0}$} circle (1.5pt);
			\filldraw[mblue] (-.8,-1.732*0.8) node[black,right=1mm]{$\omega^{2}\lambda_{0}$} circle (1.5pt);
			\filldraw[mgreen] (.8,-1.732*0.8) node[black,right=1mm]{$-\omega \lambda_{0}$} circle (1.5pt);
			\filldraw[mgreen] (-.8,1.732*0.8) node[black,right=1mm]{$\omega \lambda_{0}$} circle (1.5pt);
			
			\draw[->,very thick,rotate=60,mblue] (1.6,0)  to (3.2,0.8) ;
			\draw[-,very thick,rotate=60,mblue] (1.6,0)  to (4,1.2);
			\draw[->,very thick,rotate=60,mblue] (1.6,0)  to (3.2,-0.8) ;
			\draw[->,very thick,rotate=60,mblue] (-1.6,0)  to (-3.2,0.8) ;
			\draw[-,very thick,rotate=60,mblue] (-1.6,0)  to (-4,1.2);
			\draw[->,very thick,rotate=60,mblue] (-1.6,0)  to (-3.2,-0.8);
			\draw[-,very thick,rotate=60,mblue] (-1.6,0)  to (-4,-1.2);
			\draw[-,very thick,rotate=60,mblue] (1.6,0)  to (4,-1.2);
			\draw[-,very thick,rotate=60,mblue] (-1.6,0)  to (-1.0,0.3);
			\draw[-,very thick,rotate=60,mblue] (1.6,0)  to (1.0,0.3);
			\draw[-,very thick,rotate=60,mblue] (1.0,-0.3)  to (0,0);
			\draw[->,very thick,rotate=60,mblue] (1.0,-0.3)  to (0.5,-0.15);
			\draw[-,very thick,rotate=60,mblue] (1.0,0.3)  to (0,0);
			\draw[->,very thick,rotate=60,mblue] (1.0,0.3)  to (0.5,0.15);
			\draw[-,very thick,rotate=60,mblue] (0.7,0)  to (0.9,0);
			\draw[-,very thick,rotate=60,mblue] (-1.6,0)  to (-1.2,-0.2) ;
			\draw[-,very thick,rotate=60,mblue] (-1.6,0)  to (-1.0,-0.3);
			\draw[-,very thick,rotate=60,mblue] (-1.0,-0.3)  to (0,0);
			\draw[->,very thick,rotate=60,mblue] (-1.0,-0.3)  to (-0.5,-0.15);
			\draw[-,very thick,rotate=60,mblue] (1.6,0)  to (1.0,-0.3);
			\draw[->,very thick,rotate=60,mblue] (-1.0,0.3)  to (-0.5,0.15);
			\draw[-,very thick,rotate=60,mblue] (-1.0,0.3)  to (0,0);
			
			\draw[->,very thick,rotate=120,mgreen] (1.6,0)  to (3.2,0.8) ;
			\draw[-,very thick,rotate=120,mgreen] (1.6,0)  to (4,1.2);
			\draw[->,very thick,rotate=120,mgreen] (1.6,0)  to (3.2,-0.8) ;
			\draw[->,very thick,rotate=120,mgreen] (-1.6,0)  to (-3.2,0.8) ;
			\draw[-,very thick,rotate=120,mgreen] (-1.6,0)  to (-4,1.2);
			\draw[->,very thick,rotate=120,mgreen] (-1.6,0)  to (-3.2,-0.8);
			\draw[-,very thick,rotate=120,mgreen] (-1.6,0)  to (-4,-1.2);
			\draw[-,very thick,rotate=120,mgreen] (1.6,0)  to (4,-1.2);
			\draw[-,very thick,rotate=120,mgreen] (-1.6,0)  to (-1.0,0.3);
			\draw[-,very thick,rotate=120,mgreen] (1.6,0)  to (1.0,0.3);
			\draw[-,very thick,rotate=120,mgreen] (1.0,-0.3)  to (0,0);
			\draw[->,very thick,rotate=120,mgreen] (1.0,-0.3)  to (0.5,-0.15);
			\draw[-,very thick,rotate=120,mgreen] (1.0,0.3)  to (0,0);
			\draw[->,very thick,rotate=120,mgreen] (1.0,0.3)  to (0.5,0.15);
			\draw[-,very thick,rotate=120,mgreen] (0.7,0)  to (0.9,0);
			\draw[-,very thick,rotate=120,mgreen] (-1.6,0)  to (-1.2,-0.2) ;
			\draw[-,very thick,rotate=120,mgreen] (-1.6,0)  to (-1.0,-0.3);
			\draw[-,very thick,rotate=120,mgreen] (-1.0,-0.3)  to (0,0);
			\draw[->,very thick,rotate=120,mgreen] (-1.0,-0.3)  to (-0.5,-0.15);
			\draw[-,very thick,rotate=120,mgreen] (1.6,0)  to (1.0,-0.3);
			\draw[->,very thick,rotate=120,mgreen] (-1.0,0.3)  to (-0.5,0.15);
			\draw[-,very thick,rotate=120,mgreen] (-1.0,0.3)  to (0,0);
			
			\draw[->,very thick,mred] (1.6,0)  to (3.2,0.8) ;
			\draw[-,very thick,mred] (1.6,0)  to (4,1.2);
			\draw[->,very thick,mred] (1.6,0)  to (3.2,-0.8) ;
			\draw[->,very thick,mred] (-1.6,0)  to (-3.2,0.8) ;
			\draw[-,very thick,mred] (-1.6,0)  to (-4,1.2);
			\draw[->,very thick,mred] (-1.6,0)  to (-3.2,-0.8) ;
			\draw[-,very thick,mred] (-1.6,0)  to (-4,-1.2);
			\draw[-,very thick,mred] (-1.6,0)  to (-1.0,0.3);
			\draw[->,very thick,mred] (-1.0,0.3)  to (-0.5,0.15);
			\draw[-,very thick,mred] (-1.0,0.3)  to (0,0);
			\draw[-,very thick,mred] (1.6,0)  to (4,-1.2);
			\draw[-,very thick,mred] (1.6,0)  to (1.2,0.2) ;
			\draw[-,very thick,mred] (1.6,0)  to (1.0,0.3);
			\draw[-,very thick,mred] (1.0,0.3)  to (0,0);
			\draw[->,very thick,mred] (1.0,0.3)  to (0.5,0.15);
			\draw[-,very thick,mred] (1.6,0)  to (1.2,-0.2) ;
			\draw[-,very thick,mred] (1.6,0)  to (1.0,-0.3);
			\draw[-,very thick,mred] (1.0,-0.3)  to (0,0);
			\draw[->,very thick,mred] (1.0,-0.3)  to (0.5,-0.15);
			\draw[-,very thick] (0.7,0)  to (0.9,0);
			\draw[<->,very thick] (-1,0)  to (1,0);
			\draw[<->,very thick] (-3,0)  to (3,0);
			\draw[-,very thick,mred] (-1.6,0)  to (-1.2,-0.2) ;
			\draw[-,very thick,mred] (-1.6,0)  to (-1.0,-0.3);
			\draw[-,very thick,mred] (-1.0,-0.3)  to (0,0);
			\draw[->,very thick,mred] (-1.0,-0.3)  to (-0.5,-0.15);
			\draw[-,very thick] (-0.7,0)  to (-1.4,0);
			
			\draw[<->,very thick] (-.5,-1.732*.5)  to (.5,1.732*.5);
			
			\draw[<->,very thick] (-1.5,-1.732*1.5)  to (1.5,1.732*1.5) ;
			\draw[<->,very thick] (-.5,1.732*.5)  to (.5,-1.732*.5) ;
			\draw[<->,very thick] (-1.5,1.732*1.5)  to (1.5,-1.732*1.5);
			\draw[very thick,mred] (1.6,0) circle [radius=0.5cm];
			\draw[->,very thick,mblue,rotate=60] (1.61,0.5) to (1.50,0.5);
			\draw[very thick,mblue,rotate=60] (1.6,0) circle [radius=0.5cm];
			\draw[->,very thick,mred] (1.61,0.5) to (1.50,0.5);
			\draw[->,very thick,mgreen,rotate=120] (1.61,0.5) to (1.50,0.5);
			\draw[very thick,mgreen,rotate=120] (1.6,0) circle [radius=0.5cm];
			\draw[->,very thick,mred,rotate=180] (1.61,0.5) to (1.50,0.5);
			\draw[very thick,mred,rotate=180] (1.6,0) circle [radius=0.5cm];
			\draw[->,very thick,mgreen,rotate=-60] (1.61,0.5) to (1.50,0.5);
			\draw[very thick,mgreen,rotate=-60] (1.6,0) circle [radius=0.5cm];
			\draw[->,very thick,mblue,rotate=-120] (1.61,0.5) to (1.50,0.5);
			\draw[very thick,mblue,rotate=-120] (1.6,0) circle [radius=0.5cm];
		\end{tikzpicture}
		\caption{{\protect\small
				The jump contour $\tilde{\Sigma}:=\Sigma^{(2)}\cup\partial \tilde{B}_{\epsilon}^{(\pm \lambda_0)}$ with circles oriented anticlockwise.}}
		\label{new RH}
	\end{figure}
	\begin{lemma}\label{westimate}
		
		Define $W=\tilde V-I$, then the estimate of jump matrix below is uniformly for $t$ large enough and $0<\lambda_0<M$.
		
		\begin{align}
			& \left\|(1+|\zeta|)  {W}\right\|_{\left(L^1 \cap L^{\infty}\right)(\Sigma)} \leq \frac{C}{ t}, \\
			& \left\|(1+|\zeta|)  W\right\|_{\left(L^1 \cap L^{\infty}\right)\left(\Sigma^{(2)}\setminus{\Sigma\cup\tilde\Sigma^{(\pm \lambda_0)}}\right)} \leq C e^{-c t}, \\
			\label{west}	& \left\| W\right\|_{\left(L^1 \cap L^{\infty}\right)(\partial \tilde{B}_{\epsilon}^{(\pm \lambda_0)})} \leq C t^{-1 / 2}, \\
			& \left\| W\right\|_{L^1\left(\tilde{\Sigma}^{\pm\lambda_0}\right)} \leq C t^{-1} \ln t ,\\
			& \left\| W\right\|_{L^{\infty}\left(\tilde{\Sigma}^{\pm\lambda_0}\right)} \leq C t^{-1 / 2} \ln t.
		\end{align}
	\end{lemma}
	
	\begin{proof}	
		It follows from the estimates of the jump matrix $V^{(2)}$ that the first two inequalities are verified, and the third inequality follows from the estimate \eqref{Mlambda0esti} in Lemma \ref{V2lambda0}.
		Moreover, the last two inequalities result from the estimate for $V^{(2)}-V^{(\pm \lambda_0)}$, which is shown in Lemma \ref{V2lambda0}.
		
	\end{proof}
	
	Consider the Cauchy operator defined on the complex plane excluding the contour $\tilde{\Sigma}$, denoted as $\mathcal{C}$, where
	$$
	\left(\mathcal{C} f\right)(z)=\int_{\tilde\Sigma} \frac{f(\zeta)}{\zeta-z} \frac{\mathrm{d} \zeta}{2 \pi i}, \quad z \in \CC\setminus \tilde \Sigma.
	$$
	For a function $f$ satisfying $(1+|z|)^{\frac{1}{3}}f(z) \in L^3(\tilde{\Sigma})$, the operator $\mathcal{C}f$ maps $\mathbb{C} \setminus \tilde{\Sigma}$ to $\mathbb{C}$ analytically. In any component $D$ of $\mathbb{C} \setminus \tilde{\Sigma}$, consider curves $\{C_n\}_{n=1}^{\infty}$ enclosing each compact subset of $D$, which fulfill the condition
	$$
	\sup_{n\ge1}\int_{C_n}(1+|z|)|f(z)|^3|dz|<\infty,
	$$
	Furthermore, $\mathcal{C}_{\pm}f$ exist almost everywhere on $\tilde{\Sigma}$, and both $(1+|z|)^{\frac{1}{3}}\mathcal{C}_{+}f(z)$ and $(1+|z|)^{\frac{1}{3}}\mathcal{C}_{-}f(z)$ belong to $L^3(\tilde{\Sigma})$. The operators $\mathcal{C}_{\pm}$ are bounded from the weighted $L^3(\tilde{\Sigma})$ space to itself, denoted as $\dot{L}^{3}(\tilde{\Sigma})$, and satisfy the relation $\mathcal{C}_+ - \mathcal{C}_- = I$.
	%
	%
	By the Lemma \ref{westimate} and the Riesz-Thorin interpolation inequality, it yields that
	\begin{align}\label{Wesit}
		\|(1+|\cdot|)W\|_{L^p(\tilde\Sigma)}\le C t^{-\frac{1}{2}}(\ln t)^{\frac{1}{p}},
	\end{align}
	which indicates that $W$ belongs to both the weighted $L^3(\tilde{\Sigma})$ and $L^{\infty}(\tilde{\Sigma})$ spaces.
	
	Define the operator $\mathcal{C}_{W}$ as
	$$
	\mathcal {C}_{W} f=\mathcal {C}_{+}\left(fW_{-}\right)+\mathcal {C}_{-}\left(fW_{+}\right),
	$$
	where $\mathcal{C}_{W}$ maps from $\dot{L}^3(\tilde{\Sigma}) \cap L^{\infty}(\tilde{\Sigma})$ to $\dot{L}^3(\tilde{\Sigma})$ and is a bounded linear operator.
	
	\begin{lemma}
		For $t$ sufficiently large and $0<\lambda_0<M$, the operator $I-\mathcal C_W$ is invertible and $(I-\mathcal C_W)^{-1}$is a bounded linear operator from $\dot L^3(\tilde\Sigma)$ to itself.
	\end{lemma}
	Let $\mu\in I+\dot L^3(\tilde \Sigma)$ satisfy the integral equation
	$
	\mu=I+\mathcal C_W\mu,
	$
	that is $\mu=I+(I-\mathcal C_W)^{-1}C_WI$. Therefore, the following lemma holds
	\begin{lemma}\label{muest}
		For $t$ large enough and $0<\lambda_0<M$, the RH problem $\tilde M$ has a unique solution as follows
		$$
		\tilde M(x,t;\lambda)=I+\mathcal C(\mu W)=I+\int_{\tilde\Sigma} \frac{\mu(x,t;\zeta)W(x,t;\zeta)}{\zeta-\lambda} \frac{\mathrm{d} \zeta}{2 \pi i}, \quad \lambda \in \CC\setminus \tilde \Sigma.
		$$
		On the other hand, we have
		\begin{align}\label{muesit}
			\|\mu-I\|_{L^p(\tilde \Sigma)}\le  {\frac{C(\ln t)^{\frac{1}{p}}}{t^{\frac{1}{2}}}}.
		\end{align}
		
	\end{lemma}
	
	\begin{proof}
		Let $\|\mathcal{C}_{\pm}\|_p$ denote the sum of the operator norms of $\mathcal{C}_+$ and $\mathcal{C}_-$ from $L^p(\tilde{\Sigma})$ to $L^p(\tilde{\Sigma})$:
		\begin{equation*}
			\|\mathcal{C}_{\pm}\|_p := \|\mathcal{C}_+\|_{L^p(\tilde{\Sigma}) \to L^p(\tilde{\Sigma})} +\|\mathcal{C}_-\|_{L^p(\tilde{\Sigma}) \to L^p(\tilde{\Sigma})}.
		\end{equation*}
		Assume that $t$ is sufficiently large to satisfy the condition $\|W\|_{L^{\infty}(\tilde{\Sigma})} < \|\mathcal{C}_{\pm}\|_p^{-1}$. Then, the following inequality holds
		\begin{align*}
			&\|\mu - I\|_{L^p(\tilde{\Sigma})}
			\leq \sum_{j=1}^{\infty} \|\mathcal{C}_W\|^{j-1}_{L^p(\tilde{\Sigma}) \to L^p(\tilde{\Sigma})} \|\mathcal{C}_W I\|_{L^p(\tilde{\Sigma})}\\
			&\leq \sum_{j=1}^{\infty} \|\mathcal{C}_{\pm}|_p^j \|W\|^{j-1}_{L^{\infty}(\tilde{\Sigma})} \|W\|_{L^p(\tilde{\Sigma})} \
			= \frac{\|\mathcal{C}_{\pm}\|_p \|W\|_{L^p(\tilde{\Sigma})}}{1 - \|\mathcal{C}_{\pm}\|_p \|W\|_{L^{\infty}(\tilde{\Sigma})}}.
		\end{align*}
		Noting that the operators $\mathcal{C}_{\pm}$ are bounded and using the estimate \eqref{Wesit}, the inequality \eqref{muesit} follows.
	\end{proof}
	\par	
	It remains to consider the existence of a non-tangential limit as $\lambda \to 0$ that is defined as
	\begin{align}\label{Qdef}
		Q(x,t):=\lim_{\lambda\to0}\tilde M(x,t;\lambda)=I+\frac{1}{2\pi i}\int_{\tilde\Sigma}\frac{\mu (x,t;\zeta)W(x,t;\zeta)}{\zeta}d\zeta.
	\end{align}
	\par	
	Utilizing the equations \eqref{Masy}, \eqref{west} and \eqref{muesit}, the contributions from $\partial B_{\epsilon}(\lambda_0)$ and   $\partial B_{\epsilon}(-\lambda_0)$ to the right-hand side of \eqref{Qdef} are estimated
	\begin{align*}
		\frac{1}{2\pi i}\int_{\partial B_{\epsilon}(\lambda_0)}W(x,t;\zeta)\frac{d\zeta}{\zeta}+\frac{1}{2\pi i}\int_{\partial B_{\epsilon}(\lambda_0)}\frac{\left(\mu (x,t;\zeta)-I\right)W(x,t;\zeta)}{\zeta}d\zeta.\\
		=-\frac{H(\lambda_0, t) \left(M^X_{\lambda_0}(y)\right)_{1} H( \lambda_0, t)^{-1}}{3^{\frac{1}{4}}\sqrt{\frac{2t\lambda_0}{(1+\lambda_0^2)}}}+\mathcal{O}\left(\ln t\ t^{-1}\right),
	\end{align*}
	and
	\begin{align*}
		\frac{1}{2\pi i}\int_{\partial B_{\epsilon}(-\lambda_0)}W(x,t;\zeta)\frac{d\zeta}{\zeta}+\frac{1}{2\pi i}\int_{\partial B_{\epsilon}(-\lambda_0)}\frac{\left(\mu (x,t;\zeta)-I\right)W(x,t;\zeta)}{\zeta}d\zeta\\
		=\frac{H(-\lambda_0, t) \left(M^X_{-\lambda_0}(y)\right)_{1} H(-\lambda_0, t)^{-1}}{3^{\frac{1}{4}}\sqrt{\frac{2t\lambda_0}{(1+\lambda_0^2)}}}+\mathcal{O}\left(\ln t\ t^{-1}\right).
	\end{align*}

	\noindent Recalling that
	$
	\tilde M(x,t;\lambda)=\mathcal{A}\tilde M(x,t;\omega \lambda)\mathcal{A}^{-1},\,\lambda\in\CC\setminus\tilde \Sigma,
	$
	it follows that $\mu$ and $W$ also satisfy this symmetry, which facilitates the analysis of their properties across the spectrum of $\lambda$ values. Thus it is derived that
	\begin{align}\label{W-BBB}
		\begin{aligned}
			\frac{1}{2\pi i}\int_{\partial\tilde B_\epsilon{(\lambda_0)}\cup\partial\tilde B_\epsilon{(-\lambda_0)}}&W(x,t;\lambda)d\lambda\\
			&=\frac{1}{2\pi i}
			\left(\int_{\partial B_{\epsilon}(\lambda_0)}+\int_{\partial B_{\epsilon}(\omega \lambda_0)}+\int_{\partial B_{\epsilon}(\omega^2\lambda_0)}\right)
			W(x,t;\zeta)\zeta^{-1}d\zeta\\
			&+\frac{1}{2\pi i}
			\left(\int_{\partial B_{\epsilon}(-\lambda_0)}+\int_{\partial B_{\epsilon}(-\omega \lambda_0)}+\int_{\partial B_{\epsilon}(-\omega^2\lambda_0)}\right)
			W(x,t;\zeta)\zeta^{-1}d\zeta\\
			&=Q(x,t;\lambda_0)+ \mathcal{A}^{-1}Q(x,t;\lambda_0)\mathcal{A}+ \mathcal{A}^{-2}Q(x,t;\lambda_0)\mathcal{A}^2\\
			&+Q(x,t;-\lambda_0)+ \mathcal{A}^{-1}Q(x,t;-\lambda_0)\mathcal{A}+ \mathcal{A}^{-2}Q(x,t;-\lambda_0)\mathcal{A}^2.
		\end{aligned}
	\end{align}
	\par
	Consequently, the limit in equation (\ref{Qdef}) can be written as
	\begin{align}\label{Q-limit}
		\lim_{\lambda\to0}\tilde M(x,t;\lambda)=&I-\left(\frac{\sum_{j=0}^{2} \mathcal{A}^{-j}H(\lambda_0, t) \left(M^X_{\lambda_0}(y)\right)_{1} H( \lambda_0, t)^{-1}\mathcal{A}^{j}}{3^{\frac{1}{4}}\sqrt{\frac{2t\lambda_0}{(1+\lambda_0^2)}}}\right)\nonumber\\
		&\quad +\left(\frac{\sum_{j=0}^{2} \mathcal{A}^{-j}H(-\lambda_0, t) \left(M^X_{-\lambda_0}(y)\right)_{1} H( -\lambda_0, t)^{-1}\mathcal{A}^{j}}{3^{\frac{1}{4}}\sqrt{\frac{2t\lambda_0}{(1+\lambda_0^2)}}}\right)\nonumber\\	
		&=\frac{1}{3^{\frac{1}{4}}\sqrt{\frac{2t\lambda_0}{(1+\lambda_0^2)}}}\left(\begin{matrix}
			1 & \alpha & \beta\\
			\beta & 1 & \alpha\\
			\alpha & \beta & 1
		\end{matrix}\right),
	\end{align}
	with
	$\alpha=\delta_{-\lambda_0}^0e^{-\vartheta_{21}(\lambda_0)}\beta_{12}^{-\lambda_0}-(\delta_{\lambda_0}^0)^{-1}e^{-\vartheta_{21}(\lambda_0)}\beta_{12}^{\lambda_0}$ and $\beta=(\delta_{-\lambda_0}^0)^{-1}e^{\vartheta_{21}(-\lambda_0)}\beta_{21}^{-\lambda_0}-\delta_{\lambda_0}^0e^{\vartheta_{21}(\lambda_0)}\beta_{21}^{\lambda_0}$.
	\begin{remark}
		It follows from the condition $r_{j}(0)=0~(j=1,2)$ that the estimate of $W(x,t;\zeta)\zeta^{-1}$ in (\ref{W-BBB}) is regular at $\zeta=0$.
	\end{remark}

	\subsubsection{Proof of the Sector {\rm IV} in Theorem \ref{uasy}} \label{SectorII}
	
	Reminding the reconstruction formula (\ref{usolution}), the limit in equation (\ref{Q-limit}) and series of deformations above, the asymptotic solution of the Tzitz\'eica equation (\ref{Tt}) for $\left|\frac{x}{t}\right|< 1$ can be formulated by
	$$
	\begin{aligned}
		u(x,t)&=\lim_{\lambda\to0}\log[(\omega,\omega^2,1)M(x,t;\lambda)]_{13}\\
		&=\lim_{\lambda\to0}\log[(\omega,\omega^2,1)\tilde M(x,t;\lambda)]_{13}+ {\mathcal{O}\left(\frac{\ln t}{t}\right)}\\
		&=\log(1+\frac{2\Re(\omega^2\delta^0_{-\lambda_0}e^{\vartheta_{21}(-\lambda_0)}\beta_{12}^{-\lambda_0})-2\Re(\omega\delta^0_{\lambda_0}e^{\vartheta_{21}(\lambda_0)}\beta_{21}^{\lambda_0})}{3^{\frac{1}{4}}\sqrt{\frac{2t\lambda_0}{(1+\lambda_0^2)}}})+{\mathcal{O}\left(\frac{\ln t}{t}\right)}\\	
		&=\log(1+3^{-\frac{1}{4}}\sqrt{\frac{2(1+\lambda_0^2)}{t\lambda_0}}\left(\sqrt{\nu_1}\cos\left(\frac{5\pi i}{12}-\frac{2\sqrt{3}t\lambda_0}{1+\lambda_0^2}-\nu_1\ln\left(\frac{6\sqrt{3}t\lambda_0}{1+\lambda_0^2}\right)+s_1\right)\right.\\
		&\left.\quad -\sqrt{\nu_4}\cos\left(\frac{13\pi i}{12}-\frac{2\sqrt{3}t\lambda_0}{1+\lambda_0^2}
		-\nu_4\ln\left(\frac{6\sqrt{3}t\lambda_0}{1+\lambda_0^2}\right)+s_2\right)\right)+\mathcal O\left(\frac{\ln t}{t}\right),
	\end{aligned}
	$$
	where
	\begin{equation}\label{s}
		\begin{aligned}
			s_1=-(\arg y_1+\arg \Gamma(-i\nu_1)+\nu_4\ln4)+\frac{1}{\pi}\int_0^{-\lambda_0}\log_0\frac{|s-\omega\lambda_0|}{|s-\lambda_0|}d\ln(1-|r_2(s)|^2)\\
			+\frac{1}{\pi}\int_0^{\lambda_0}\log_{-\pi}\frac{|s-\lambda_0|}{|s-\omega\lambda_0|}d\ln(1-|r_1(s)|^2),\\
			s_2=-(\arg y_2+\arg \Gamma(-i\nu_4)+\nu_1\ln4)+\frac{1}{\pi}\int_0^{\lambda_0}\log_{-\pi}\frac{|s+\omega\lambda_0|}{|s+\lambda_0|}d\ln(1-|r_1(s)|^2)\\
			+\frac{1}{\pi}\int_0^{-\lambda_0}\log_{0}\frac{|s+\lambda_0|}{|s+\omega\lambda_0|}d\ln(1-|r_2(s)|^2).
		\end{aligned}
	\end{equation}
	\par
	Henceforth, this establishes the asymptotic formula of Theorem \ref{uasy} in Sector II, culminating the proof in its entirety.
\subsubsection{Long-time behavior near the light cone $|\frac{x}{t}|\to1$ from the interior region}
{In the case where $|\frac{x}{t}|\to 1$ from within the light cone, the critical point $\lambda_0$ approaches $0$ for $x>0$ and diverges to $\infty$ for $x<0$. Moreover, the reflection coefficient $r_1(\lambda)$ vanishes to all orders as $\lambda\to0$ and $\lambda\to\infty$. As a result, the estimates in Lemma \ref{lemma-decomposition} can be further refined by
\begin{lemma}\label{lemma-decomposition-new}
Under the same assumptions as in Lemma \ref{lemma-decomposition}, as $\lambda_0\to0$ or $\lambda\to\infty$, there exists a nonnegative smooth function $C_N(\lambda)$ that vanishes to all orders at $\lambda=0$ and $\lambda=\infty$, then we have
    \begin{enumerate}
    \item
    For $Re\lambda>\lambda_0 \ \text{and}\ \Im(\lambda)>0$,
			$$
			   \left|r_{1, a}(x, t, \lambda)-\sum_{n=0}^N \frac{r_1^{(n)}(\lambda_0)(\lambda-\lambda_0)^n}{n!}\right| \leq C_N(\lambda_0){e^{-\frac{t}{4}\operatorname{Re} \theta_{21}( \lambda)}}|\lambda-\lambda_0|^{N+1},
			$$
			Moreover, for $\lambda\in B_{\lambda_0}(0),\Im\lambda>0\ \text{and}\ \frac{\lambda_0}{2}<\Re\lambda<\lambda_0$, it follows
			\begin{align*}
				&\left|\rho_{1, a}(x, t, \lambda)-\sum_{n=0}^{N}\frac{\rho_1^{(n)}(\lambda_0)(\lambda-\lambda_0)^n}{n!}\right| \leq {C_N(\lambda_0) e^{-\frac{t}{4}\operatorname{Re} \theta_{21}( \lambda)}|\lambda-\lambda_0|^{N+1}}, 
			\end{align*}
			\item  For $1\le p\le\infty$,  the functions  $r_{j,r} $ and $\rho_{j,r}$ satisfy
			$$
				\left\| (1+|\cdot|)r_{1, r}(x, t, \lambda)e^{-2t\theta_{21}}\right\|_{L^p{(\lambda_0,\infty)}} \leq \frac{c}{t^{N+1/2}}\quad \lambda_0<\lambda,
			$$
			and  $\rho_{1,r}(x,t;\lambda)=\mathcal{O}(t^{-N-\frac{1}{2}})$ for $0<|\lambda|<{\lambda_0}$.
                The estimates for $r_2$ and $\rho_2$ are similar and are omitted for brevity.  
		\end{enumerate}
\end{lemma}
The other transformation follows a similar analysis. As $t\to\infty$ and $\frac{|x|}{t}\to1$ from within the light cone, the function $u(x,t)$ retains the same leading term, while the error term is refined to $\mathcal{O}\left(t^{-N}+\frac{C_N(\lambda_0)\ln t}{t}\right)$.}

	\appendix
	
	\section{\bf The model Riemann-Hilbert problem}\label{modelRH}
	
	The jump contour of the model RH problem for function $M^X_{\pm\lambda_0}$ is shown in Figure \ref{model}.
	
	\begin{figure}[!h]
		\centering
		\begin{overpic}[width=.55\textwidth]{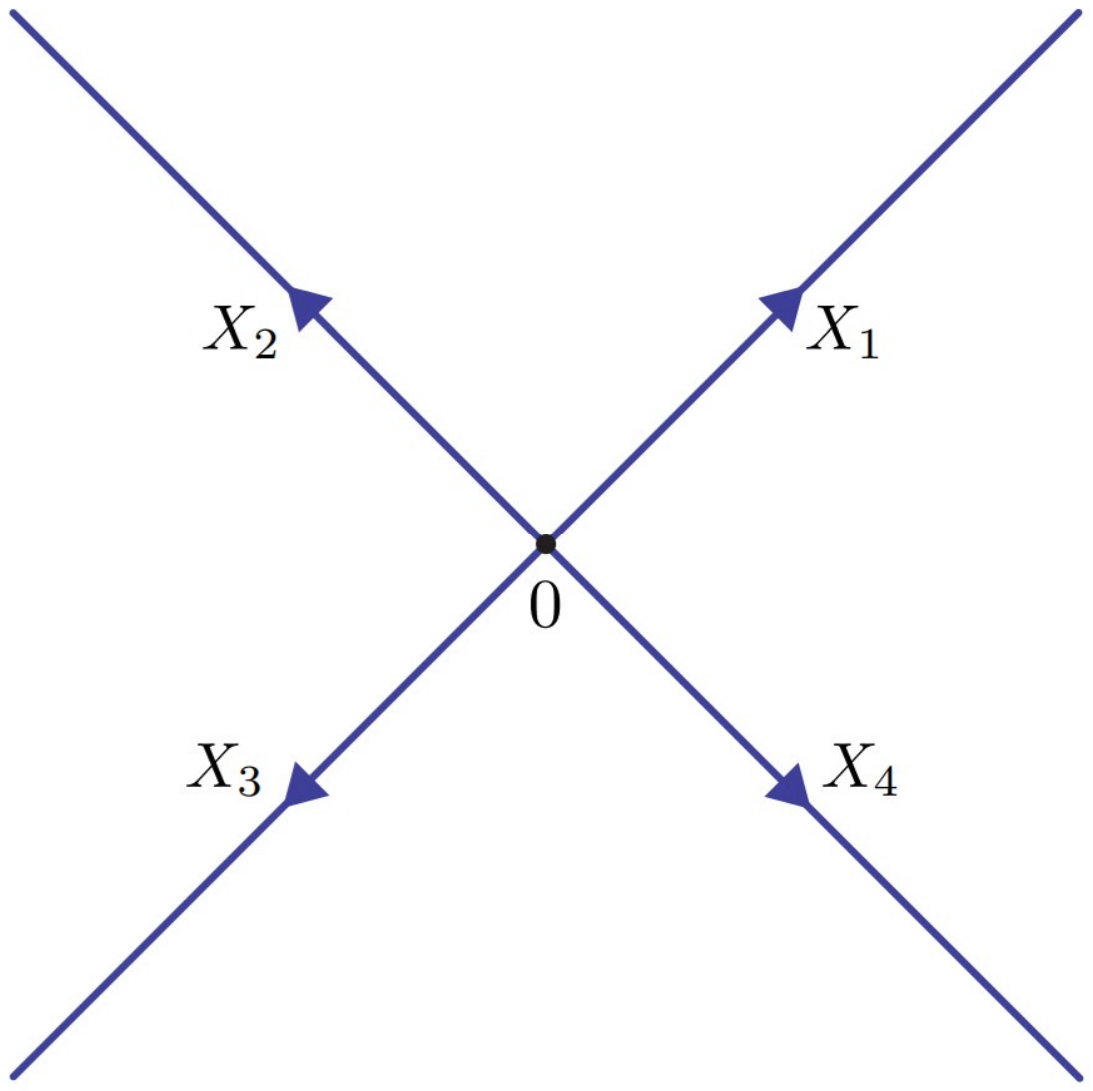}
		\end{overpic}
		\caption{{\protect\small
				The jump contour of the model problem for function $M^X_{\pm\lambda_0}$.}}
		\label{model}
	\end{figure}
	Here the jump matrices are
	\begin{align*}
		&X_1=\{z\in\CC:z=re^{\frac{\pi i}{4}},0\le r\le\infty\},\quad &X_2=\{z\in\CC:z=re^{\frac{3\pi i}{4}},0\le r\le\infty\},\\
		&X_3=\{z\in\CC:z=re^{\frac{5\pi i}{4}},0\le r\le\infty\},\quad &X_4=\{z\in\CC:z=re^{\frac{7\pi i}{4}},0\le r\le\infty\},
	\end{align*}
	which are oriented away from the origin. Denote $X=\cup_{j=1}^4X_j$ and the functions $\nu_{1}(y)=-\frac{1}{2 \pi} \ln \left(1-\left|r_1(y)\right|^2\right)$ and $\nu_4(y)=-\frac{1}{2 \pi} \ln \left(1-\left|r_2\left(y\right)\right|^2\right)$. The model RH problem for function $M^X_{\lambda_0}$ is defined below.
	
	\begin{RHproblem}\label{RHpM0}
		The three-order matrix-valued function $M^X_{\lambda_0}$ satisfies the properties:
		\begin{enumerate}
			\item  $M^X_{\lambda_0}(\cdot\ ,y):~~\CC\setminus X\to\CC^{3\times3}$  is analytic for $z \in\CC\setminus X$.
			
			\item The function $M^X_{\lambda_0}(z,y)$ is continuous on $X\setminus\{0\}$ and satisfies the jump condition
			$$
			(M^X_{\lambda_0}(z,y))_+=(M^X_{\lambda_0}(z,y))_-v^{X}_{\lambda_0}(z,y),\quad z\in\CC\setminus \{0\},
			$$
			where the jump matrix $v^X_{\lambda_0}(z,y)$ is defined by
			$$
			\begin{aligned}
				& \left(\begin{array}{ccc}
					1 & -y z^{-2 i \nu_1(y)} e^{\frac{i z^2}{2}} & 0 \\
					0 & 1 & 0 \\
					0 & 0 & 1
				\end{array}\right) \quad \text { if } z \in X_1, \quad\left(\begin{array}{ccc}
					1 & 0 & 0 \\
					-{\frac{\bar{y}}{1-|y|^2}} z^{2 i \nu_1(y)} e^{-\frac{i z^2}{2}} & 1 & 0 \\
					0 & 0 & 1
				\end{array}\right) \text { if } z \in X_2, \\
				& \left(\begin{array}{ccc}
					1 & \frac{y}{1-|y|^2} z^{-2i \nu_1 }  e^{\frac{iz^2}{2}} & 0 \\
					0 & 1 & 0 \\
					0 & 0 & 1
				\end{array}\right) \text { if } z \in X_3, \quad\left(\begin{array}{ccc}
					1 & 0 & 0 \\
					\bar y z^{2i \nu_1 } e^{-\frac{iz^2}{2}} & 1 & 0 \\
					0 & 0 & 1
				\end{array}\right) \text { if } z \in X_4,
			\end{aligned}
			$$
			with $z^{2i\nu_1(y)}=e^{2i\nu_1(y){log_{-\pi}(z)}}$ .
			\item  $M^X_{\lambda_0}(z ,y)\to I$  as $z\to\infty$.
			\item  $M^X_{\lambda_0}(z ,y)\to \mathcal{O}(1)$ as $z\to 0$.
		\end{enumerate}
	\end{RHproblem}
	\par
	The RH problem \ref{RHpM0} has a unique solution. For $|y|<1$, the solution $M^X_{\lambda_0}$ of RH problem \ref{RHpM0} satisfies the following expansion:
	\begin{align}\label{Mlambda0}		M^X_{\lambda_0}(y,z)=I+\frac{\left(M^X_{\lambda_0}(y)\right)_1}{z}+\mathcal{O}\left(\frac{1}{z^2}\right).
	\end{align}
	where
	$$
	\left(M_{\lambda_0}^X(y)\right)_1=\left(\begin{array}{ccc}
		0 & \beta_{12}^{\lambda_0} & 0 \\
		\beta_{21}^{\lambda_0} & 0 & 0 \\
		0 & 0 & 0
	\end{array}\right), \quad |y|<1,
	$$
	and
	$$
	\beta_{12}^{\lambda_0}=-\frac{\sqrt{2\pi}e^{\frac{\pi i}{4}}e^{-\frac{\pi\nu_1}{2}}}{ \bar y\Gamma(i\nu_1)}, \quad \beta_{21}^{\lambda_0}=-\frac{\sqrt{2\pi}e^{-\frac{\pi i}{4}}e^{-\frac{\pi\nu_1}{2}}}{ y\Gamma(-i\nu_1)}.
	$$
	
	Similarly, the model RH problem for function $M^X_{-\lambda_0}$ is defined below.	
	
	\begin{RHproblem}\label{RHpM-0}
		The three-order matrix-valued function $M^X_{-\lambda_0}$ satisfies the properties:
		\begin{enumerate}
			\item $M^X_{-\lambda_0}(\cdot,y):~~\CC\setminus X\to\CC^{3\times3}$  is analytic for $z \in\CC\setminus X$.
			\item The function $M^X_{-\lambda_0}(z,y)$ is continuous on $X\setminus\{0\}$ and satisfies the jump condition
			$$
			(M^X_{-\lambda_0}(z,y))_+=(M^X_{-\lambda_0}(z,y))_-v^{X}_{-\lambda_0}(z,y),\quad z\in\CC\setminus \{0\},
			$$
			where the jump matrix $v^X_{-\lambda_0}(z,y)$ is defined by
			$$
			\begin{aligned}
				& \left(\begin{array}{ccc}
					1 & 0 & 0 \\
					-\frac{y}{1-|y|^2} z^{-2 i \nu_4(y)} e^{\frac{i z^2}{2}} & 1 & 0 \\
					0 & 0 & 1
				\end{array}\right) \quad \text { if } z \in X_1, \quad\left(\begin{array}{ccc}
					1 & -\bar{y} z^{2 i \nu_4(y)} e^{-\frac{i z^2}{2}} & 0 \\
					0 & 1 & 0 \\
					0 & 0 & 1
				\end{array}\right) \text { if } z \in X_2, \\
				& \left(\begin{array}{ccc}
					1 & 0 & 0 \\
					y z^{-2i \nu_4 }  e^{\frac{iz^2}{2}} & 1 & 0 \\
					0 & 0 & 1
				\end{array}\right) \text { if } z \in X_3, \quad\left(\begin{array}{ccc}
					1 & \frac{\bar y}{1-|y|^2} z^{2i \nu_4 } e^{-\frac{iz^2}{2}} & 0 \\
					0 & 1 & 0 \\
					0 & 0 & 1
				\end{array}\right) \text { if } z \in X_4,
			\end{aligned}
			$$
			with $z^{2i\nu_4(y)}=e^{2i\nu_4(y)}e^{\log_{0}(z)}$ .
			\item $M^X_{-\lambda_0}(z ,y)\to I$  as $z\to\infty$.
			\item $M^X_{-\lambda_0}(z ,y)\to \mathcal{O}(1)$ as $z\to 0$.
		\end{enumerate}
	\end{RHproblem}
	\par
	The RH problem \ref{RHpM-0} has a unique solution.  For $|y|<1$, the solution $M^X_{-\lambda_0}$ of RH problem \ref{RHpM-0} satisfies the following expansion
	\begin{align}\label{-Mlmabda0}
		M^X_{-\lambda_0}(y,z)=I+\frac{\left(M^X_{-\lambda_0}(y)\right)_1}{z}+\mathcal{O}\left(\frac{1}{z^2}\right),
	\end{align}
	where
	$$
	\left(M_{-\lambda_0}^X(y)\right)_1=\left(\begin{array}{ccc}
		0 & \beta_{12}^{-\lambda_0} & 0 \\
		\beta_{21}^{-\lambda_0} & 0 & 0 \\
		0 & 0 & 0
	\end{array}\right), \quad |y|<1,
	$$
	and
	$$
	\beta_{12}^{-\lambda_0}=-\frac{\sqrt{2 \pi} e^{-\frac{\pi i}{4}} e^{-\frac{5 \pi \nu_4}{2}}}{{y} \Gamma(-i \nu_4)},\quad \beta_{21}^{-\lambda_0}=-\frac{\sqrt{2 \pi} e^{\frac{\pi i}{4}} e^{\frac{3 \pi \nu_4}{2}}}{\bar y \Gamma(i \nu_4)}.
	$$
	\begin{remark}
		The methodology employed for proving the solvability and delineating the expansions of the Riemann-Hilbert problems \ref{RHpM0} and \ref{RHpM-0} adheres to a conventional framework. Detailed expositions of this proof process are accessible in the literatures, notably within the paper by Charlier, Lenells and Wang \cite{Charlier-Lenells-Wang-2021}, and further reference provided in \cite{Ds-Xd-2022}.
	\end{remark}
	\noindent{\bf Conflict of interest declaration.} We declare we have no competing interests.
	\subsection*{Acknowledgements}
	The authors are grateful to Jonatan Lenells for valuable comments on the initial manuscript.
	Support is acknowledged from the National Natural Science Foundation of China, Grant No.12371249 , No.12371247 \& No.12431008.  The third author acknowledges the GNFM--INDAM group and the research project 
\textit{Mathematical Methods in NonLinear Physics (MMNLP)}, Gruppo~4 -- 
Fisica Teorica of INFN and the support of the scholarship provided by the China Scholarship Council (CSC) under Grant No. 202406040149.

	\bibliographystyle{amsplain}

\end{document}